\documentclass[twoside,english,most, usenames, dvipsnames]{article}

\usepackage[T1]{fontenc}
\usepackage[utf8]{inputenc}
\usepackage{xcolor}
\usepackage{babel}
\usepackage{array}
\usepackage{latexsym}
\usepackage{varioref}
\usepackage{refstyle}
\usepackage{float}
\usepackage{url}
\usepackage{tcolorbox}
\usepackage{varwidth}
\usepackage{tabularx}
\usepackage{dsfont}
\usepackage{amsmath}
\usepackage{amsthm}
\usepackage{amssymb}
\usepackage{cancel}
\usepackage{stmaryrd}
\usepackage{graphicx}
\usepackage{geometry}
\PassOptionsToPackage{normalem}{ulem}
\usepackage{ulem}
\usepackage[bookmarks=true,bookmarksnumbered=true,bookmarksopen=true,bookmarksopenlevel=1,
 breaklinks=true,pdfborder={0 0 1},backref=false,colorlinks=true,pdfpagemode=FullScreen]
 {hyperref}
\hypersetup{pdftitle={Quantum Petri Nets with Quantum Event Structures semantics},
 pdfauthor={Julien Saan JOACHIM},
 pdfsubject={A semantically grounded framework for Petri Nets compatible with Quantum Event Structures},
 pdfkeywords={Quantum Petri Net, Quantum Event Structure, Quantum Occurrence Net, Quantum Concurrency, Quantum Computing, Drop Condition, Drop Function, }}

\makeatletter

\AtBeginDocument{\providecommand\lemref[1]{\ref{lem:#1}}}
\providecommand{\tabularnewline}{\\}
\newenvironment{cellvarwidth}[1][t]
    {\begin{varwidth}[#1]{\linewidth}}
    {\@finalstrut\@arstrutbox\end{varwidth}}
\floatstyle{ruled}
\newfloat{algorithm}{tbp}{loa}
\providecommand{\algorithmname}{Algorithm}
\floatname{algorithm}{\protect\algorithmname}
\RS@ifundefined{subsecref}
  {\newref{subsec}{name = \RSsectxt}}
  {}
\RS@ifundefined{thmref}
  {\def\RSthmtxt{theorem~}\newref{thm}{name = \RSthmtxt}}
  {}
\RS@ifundefined{lemref}
  {\def\RSlemtxt{lemma~}\newref{lem}{name = \RSlemtxt}}
  {}

\numberwithin{equation}{section}
\numberwithin{figure}{section}
\theoremstyle{plain}
\newtheorem{thm}{\protect\theoremname}
\theoremstyle{definition}
\newtheorem{defn}[thm]{\protect\definitionname}
\theoremstyle{definition}
\newtheorem{example}[thm]{\protect\examplename}
\theoremstyle{definition}
\newtheorem{problem}[thm]{\protect\problemname}
\newenvironment{lyxcode}
	{\par\begin{list}{}{
		\setlength{\rightmargin}{\leftmargin}
		\setlength{\listparindent}{0pt}
		\raggedright
		\setlength{\itemsep}{0pt}
		\setlength{\parsep}{0pt}
		\normalfont\ttfamily}%
	 \item[]}
	{\end{list}}
\theoremstyle{plain}
\newtheorem{prop}[thm]{\protect\propositionname}
\theoremstyle{plain}
\newtheorem{lem}[thm]{\protect\lemmaname}
\theoremstyle{remark}
\newtheorem*{rem*}{\protect\remarkname}
\theoremstyle{remark}
\newtheorem{rem}[thm]{\protect\remarkname}

\usepackage[T1]{fontenc}
\usepackage{textcomp}
\usepackage{lmodern}

\usepackage{wrapfig}
\usepackage{pgfplots}
\pgfplotsset{compat=1.17}
\usepackage{physics}
\usepackage{amsmath}
\usepackage{tikz}
\usepackage{mathdots}
\usepackage{mathtools}
\usepackage{yhmath}
\usepackage{cancel}
\usepackage{color}
\usepackage{siunitx}
\usepackage{array}
\usepackage{multirow}
\usepackage{amssymb}
\usepackage{gensymb}
\usepackage{tabularx}
\usepackage{extarrows}
\usepackage{booktabs}

\usepackage{breqn}
\usetikzlibrary{fadings}
\usetikzlibrary{patterns}
\usetikzlibrary{shadows.blur}
\usetikzlibrary{shapes}

\ifdefined\showcaptionsetup
 \PassOptionsToPackage{caption=false}{subfig}
\fi
\usepackage{subfig}
\AtBeginDocument{

}

\makeatother

\usepackage[style=authoryear]{biblatex}
\providecommand{\definitionname}{Definition}
\providecommand{\examplename}{Example}
\providecommand{\lemmaname}{Lemma}
\providecommand{\problemname}{Problem}
\providecommand{\propositionname}{Proposition}
\providecommand{\remarkname}{Remark}
\providecommand{\theoremname}{Theorem}

\begin{document}
\title{Quantum Petri Nets with Quantum Event Structures semantics}
\author{Julien Saan JOACHIM\thanks{\protect\url{julien.joachim@ens-paris-saclay.fr} Ecole Normale Supérieure
Paris-Saclay, France}\and Marc de Visme\thanks{\protect\url{mdevisme@lmf.cnrs.fr} Université Paris-Saclay, CNRS,
ENS Paris-Saclay, INRIA, Laboratoire Méthodes Formelles, 91190,Gif-sur-Yvette,
France}\and Stefan Haar\thanks{\protect\url{stefan.haar@inria.fr} INRIA, France}}

\maketitle

\newcommand{\circled}[1]{%
  \tikz[baseline=(char.base)]\node[anchor=south west, draw,rectangle, rounded corners=7pt, inner sep=0pt, minimum size=5mm,
    text height=2mm](char){\ensuremath{#1}} ;}

\newcommand{\ssubset}{
  \mathrel{
    \mathrlap{\subset}
    \hphantom{\ll}
    \mathllap{\subset}
  }
}
\newcommand\widerc[1]{%
\tikz[baseline=(wideArcAnchor.base)]{
    \node[inner sep=0] (wideArcAnchor) {$#1$}; 
    \coordinate (wideArcAnchorA) at ($(wideArcAnchor.north west) + (0.1em,0.0ex)$);
    \coordinate (wideArcAnchorB) at ($(wideArcAnchor.north east) + (-0.1em,0.0ex)$);
    \draw[line width=0.1ex,line cap=round,out=45,in=135] (wideArcAnchorA) to (wideArcAnchorB);
}}

\global\long\def\ext{\relbar\joinrel\subset}%

\global\long\def\ci#1{\circled{#1}}%

\global\long\def\tr{\mathbf{tr}}%
\global\long\def\Id#1{\mathbf{Id}_{#1}}%

\global\long\def\minext{\relbar\joinrel\ssubset}%

\global\long\def\upclo#1{\left\uparrow #1\right\uparrow }%
\global\long\def\fanout#1#2{#1fanout#2}%
\global\long\def\fanin#1#2{#1fanin#2}%

\global\long\def\pre#1{\phantom{}^{\bullet}#1}%
\global\long\def\post#1{#1^{\bullet}}%

\global\long\def\m#1{\mathbf{m}_{#1}}%
\global\long\def\widearc#1{\widehat{#1}}%

\global\long\def\qot{\widetilde{Q_{0}}}%

\global\long\def\int#1#2{\left[#1;#2\right]}%

\global\long\def\d#1#2{d\left[#1;#2\right]}%

\global\long\def\dv#1#2{\dfrac{d}{v}\left[#1;#2\right]}%

\global\long\def\qd#1#2{\mathtt{d}\left[#1;#2\right]}%

\global\long\def\qdv#1#2{\mathtt{d}\left[#1;#2\right]}%
\global\long\def\operator#1{\mathbb{O}\left(#1\right)}%
\global\long\def\rest#1{E_{\left|#1\right.}}%
\global\long\def\cutpn#1{\widearc{\mathbf{\mathbf{#1}}}}%
\global\long\def\cut#1{\widearc{#1}}%

\global\long\def\restg#1#2{{#1}_{\left|#2\right.}}%

\begin{abstract}
Classical Petri nets provide a canonical model of concurrency, with
unfolding semantics linking nets, occurrence nets, and event structures.
No comparable framework exists for quantum concurrency: existing ``quantum
Petri nets'' lack rigorous concurrent and sound quantum semantics,
analysis tools, and unfolding theory.

We introduce \emph{Quantum Petri Nets} (QPNs), Petri nets equipped
with a quantum valuation compatible with the quantum event structure
semantics of Clairambault, De Visme, and Winskel (2019). Our contributions
are: (i) a local definition of \emph{Quantum Occurrence Nets} (LQONs)
compatible with quantum event structures, (ii) a construction of QPNs
with a well-defined unfolding semantics, (iii) a compositional framework
for QPNs.

This establishes a semantically well grounded model of quantum concurrency,
bridging Petri net theory and quantum programming.
\end{abstract}

\section{Introduction}

Petri nets offer a rigorous and compositional foundation for modeling
concurrency, causality, and synchronization in classical systems.
Their structural clarity and rich semantic theory--linking safe Petri
nets, occurrence nets, and event structures via a chain of categorical
co-reflections called the ``unfolding semantics'' \autocite{winskelPetriNetsAlgebras1987,winskelEventStructures1987,nielsenPetriNetsEvent1981}--have
made them central to the study of distributed computation.

In contrast, quantum concurrency still lacks a similarly grounded
semantic framework. Recent progress in quantum programming languages--notably
the $Q\Lambda$ quantum $\lambda$-calculus \autocite{selingerLambdaCalculusQuantum2006}
and its full abstraction \autocite{clairambaultFullAbstractionQuantum2020}--opens
the door to such a foundation. Notably, Quantum Event Structure have
enabled the development of a concurrent game semantics for $Q\Lambda$,
by encoding quantum computations as strategies over event structures
annotated with quantum valuations \autocites{clairambaultConcurrentQuantumStrategies2019}{clairambaultGameSemanticsQuantum2019}.
They hence capture the causal relations while still enabling quantum
non-local behaviors in a compact and compositional way. 

\paragraph{Previous Works}
Two previous attempts have sought to define “quantum Petri nets” by
adapting classical nets to quantum systems--typically by associating
quantum states with tokens and unitary or measurement operations with
transitions \autocite{letiaQuantumPetriNets2021,schmidtHowBakeQuantum2021}.
However, these models suffer from theoretical limitations, of which
a lack of rigorous handling of concurrency and entanglement, weak
or absent support for composition, limited analytical or verification
tools, and critically no formal unfolding semantics. As a result,
none of these models have emerged as viable foundations for reasoning
about concurrent and parallel quantum systems.

\paragraph{Problem.}

There is currently no definition of a compositional, semantically
well-grounded model for quantum concurrency parallelism akin to classical
Petri nets.

\paragraph{Contributions.}

This work addresses the above gap by introducing a new model of Quantum
Petri Nets (QPNs) that is grounded in the concurrent game semantics
of $\lambda$-calculus \autocite{clairambaultGameSemanticsQuantum2019},
and is equipped with a well-defined unfolding semantics -- based
on Quantum Event Structures. Our contributions include:
\begin{itemize}
\item A definition of \emph{Local Quantum Occurrence Nets (LQONs)}, as Petri
Nets annotated with local quantum valuations, in direct correspondence
with Quantum Event Structures. The key semantic constraint--the drop-condition--ensures
that those valuations correspond to sub-density operators, whose traces
are probability valuations on quantum states. {[}Section \ref{sec:Local Quantum Occurrence Nets}{]}
\item A construction of Quantum Petri Nets, as Petri nets annotated with
local quantum valuations, preserving semantic correspondence through
an unfolding construction. {[}Section \ref{sec:Quantum Petri Nets}{]}
\item A simple criterion for checking (almost locally) the ``drop-condition''
with a reasonable combinatorial complexity. {[}Section \ref{subsec:Local Drop Condition}{]}
\item A definition of compositional operations (parallel composition and
joins) for QPNs, and a sufficient condition under which these preserve
the some required quantum properties. {[}Section \ref{subsec:Interaction of Quantum Petri Nets}{]}
\end{itemize}

\paragraph{Future Work}

Further combinatorial simplifications of the ``drop-condition''
could be explored. A proof of a categorical co-reflection between
Quantum Event Structures and LQONs, or an adjunction between QPN and
its unfolding, would strengthen the formal grounding of the approach.
Finally, one should complete the compositional framework of QPNs to
characterize exactly those compositions that preserve all properties
of QPNs after being applied. 

\paragraph{Outline}

The remainder of this article is structured as follows. Section \ref{sec:Preliminaries}
recalls background material on Petri nets, occurrence nets, and event
structures, along with their relationships via unfolding semantics.
It then introduces de Visme’s notion of Quantum Event Structure, which
serves as the semantic foundation of our model. Sections \ref{sec:Global Quantum Occurrence Nets}
and \ref{sec:Local Quantum Occurrence Nets} introduce the notions
of Global and Local Quantum Occurrence Nets, respectively. Finally,
Section \ref{sec:Quantum Petri Nets} defines Quantum Petri Nets,
and presents their compositional semantics.

\section{Preliminaries \label{sec:Preliminaries}}

We assume a certain familiarity with Petri nets, and only precise
the relevant notions and notations for the rest of the discussion. 

\subsection{Classical Petri Nets}

A Petri net is a tuple of the form $N=\left(P,T,F,\m 0\right)$ ,
with $P\cap T=\emptyset$, where the elements of $P$ are called places
, those of $T$ transitions, where $F\subseteq\left(P\times T\right)\cup\left(T\times P\right)$
is a set of tuples called ``flow relation'', and where $\m 0\overset{multi}{\subseteq}P$
is a multiset called the ``initial marking''. Places are represented
in circles, events in squares -e.g. the place $\ci a\in P$ and transition
$\boxed{e}\in T$. It forms a bipartite graph where $F$ is the set
of edges and $P\sqcup T$ is the vertex partition, with distinguished
vertices specified by $\m 0$. The pre-places of a transition $\boxed{t}\in T$
are denoted by $\pre{\boxed{t}}=\left\{ \ci p\in P|\ci p\rightarrow\boxed{t}\in F\right\} $.
Similarly, its post-places are $\post{\boxed{t}}=\left\{ \ci p\in P|\boxed{t}\rightarrow\ci p\in F\right\} $.
Pre and Post transitions of a place $\ci p$ are denoted in a similar
fashion. A marking of a Petri Net $N$ is a multiset of places $\m{}\overset{multi}{\subseteq}P$,
and reachable markings for the classical discrete semantics are defined
in the usual way. If $N$ is a safe net, $\m{}\subseteq P$ is a subset
of $P$.\\

Let us simply recall that Petri nets are a powerful formalism for
modeling and analyzing concurrent and parallel systems due to their
explicit representation of causality, concurrency, synchronization,
and non-determinism through token flow and transition firing. In their
classical version, they also benefit from several analysis techniques
to check system properties-such as reachability, liveness, and invariants--while
numerous extensions (e.g., continuous, timed, colored, stochastic
Petri nets) enable scalable modeling of real-time, data-rich, and
probabilistic concurrent systems. This motivates the need for a parallel
tool for the analysis of quantum systems, and a robust definition
of Quantum Petri Nets.\\

Occurrence nets are constructions--actually some acyclic Petri Nets--that
are closely linked to Petri nets through a process called the ``Unfolding
semantics''. They directly make explicit the causality, concurrency
and conflict relation between distinct executions of a Petri Net.
More precisely, the concurrency is represented with a partial-order
relation, that alleviates -among other things- the burden of the combinatorial
state explosion cause by the interleaving semantics.

The exact Unfolding process will be further explored in Section \ref{sec:Quantum Petri Nets},
where after having defined Quantum Occurrence Nets, we will parallel
the link between classical Occurrence Nets and Petri net in order
to define Quantum Petri Nets. However, Occurrence nets will benefit
from an immediate introduction, since the elaboration of Quantum Occurrence
Nets is the object of the next section.

\subsection{Occurrence Nets}
\begin{defn}
Occurrence Nets are acyclic Petri Nets, whose places are called ``Conditions''
and transitions ``Events'' (cf. Section \ref{subsec:Event structure}
for the link with Event Structures). 

Fix a safe Petri net $O=\left(C,E,F,C_{0}\right)$ for the rest of
this subsection. We let the causality relation be denoted as $<$,
the transitive closure of $F$; and $\leq$ the reflexive closure
of $<$. Causality thus partially orders the events of $O$. Further,
if $\boxed{e}\in E$ is an event, let $\left[\boxed{e}\right]:=\left\{ \boxed{e'}\in E|\boxed{e'}\le\boxed{e}\right\} $
be the (backward) ``cone'' of $\boxed{e}$, and $\left[\boxed{e}\right):=\left[\boxed{e}\right]\setminus\boxed{e}$
the pre-cone of $\boxed{e}$. Conflict is defined as such.
\end{defn}

\begin{tcolorbox}[breakable,enhanced,title=Conflict \autocite{haarComputingRevealsRelation2013},colback=gray!5!white,
colframe=gray!75!black]

\begin{defn}
Two events $\boxed{x},\boxed{x\prime}$ are in conflict, written $\boxed{x}\#\boxed{x\prime}$
if there exist $\boxed{e},\boxed{e\prime}\in E$ such that:
\end{defn}

\begin{center}
\begin{tabular}{ccc}
1. $\boxed{e}\neq\boxed{e\prime}$ & 2. $\pre{\boxed{e}\cap\pre{\boxed{e'}\neq\emptyset}}$, & 3. $\boxed{e}\leq\boxed{x}$ and $\boxed{e'}\leq\boxed{x'}$.\tabularnewline
\end{tabular}
\par\end{center}
\end{tcolorbox}
\begin{tcolorbox}[breakable,enhanced,title=Occurrence Net \autocite{haarComputingRevealsRelation2013},colback=gray!5!white,
colframe=gray!75!black]

\begin{defn}
A net $O$ is called an \emph{Occurrence Net} if it satisfies the
following properties:
\begin{description}
\item [{No~self-conflict:}] $\forall x\in C\cup E,\hspace*{1em}not\left[x\#x\right]$
\item [{$<$~is~acyclic:}] i.e. $\leq$ is a partial order
\item [{Finite~cones}] all events are causally dependent on a finitely
many events, i.e. $\left|\left[\boxed{e}\right]\right|<\infty$
\item [{No~backward~branching:}] all conditions $\ci c$ satisfy $\left|\pre{\ci c}\right|\leq1$
\item [{Minimal~nodes:$C_{0}\subseteq C$}] is the set of $\leq$-minimal
nodes.
\end{description}
\end{defn}

\end{tcolorbox}

Figures \ref{fig:A-safe-Petri-net} and \ref{fig:Occurrence Net}
represents a (safe) Petri net comprised of $4$ places and $4$ transitions,
and an occurrence net. This Occurrence net turns out to be a prefix
of its unfolding (see Definition \ref{def:Branching-Process-Unfolding}).
\begin{figure}[H]
\centering
\begin{centering}
\subfloat[{\label{fig:A-safe-Petri-net}A safe Petri Net, with initial marking
$\protect\m 0=\left\{ 1,4\right\} $ {[}\autocite{haarComputingRevealsRelation2013}{]}}]{\includegraphics[width=0.35\columnwidth]{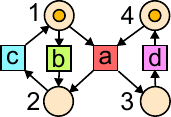}} \subfloat[\label{fig:Occurrence Net}A prefix of the unfolding of the left petri
net. This is an example of Occurrence Net.]{\includegraphics[width=0.5\columnwidth]{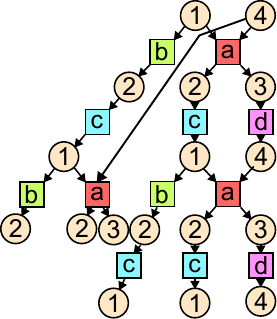}}
\par\end{centering}
\caption{}
\end{figure}

\subsection{Event structure (with polarities)\label{subsec:Event structure}}

Occurrence Nets are tightly linked to Event Structures, which are
partial-orders usually used to denote concurrent systems. This functorial
link was made explicit in \autocite[Theorem 3.4.11]{winskelEventStructures1987}
where Winskel establishes a co-reflection between the categories of
occurrence nets and event structures, which we will not need here.
On a high-level however, an Event structure can be seen as an occurrence
Net where one has discarded the conditions, while remembering the
causality relation along with the events in conflict. 

\begin{tcolorbox}[breakable,enhanced,title=Event Structure, colback=gray!5!white, colframe=gray!75!black]

\begin{defn}
An event structure is a tuple $E=\left(\underline{E},\leq,\#,pol\right)$,
where
\begin{itemize}
\item $\underline{E}$ is a set of elements called events,
\item $\leq$ is a partial order called causality,
\item $\#$ is a symmetric and irreflexive binary relation called conflict,
\item and $pol:\underline{E}\to\left\{ \ominus,0,\oplus\right\} $ associates
a polarity to every event,\\
such that:
\end{itemize}
\begin{description}
\item [{Finite~basis}] Writing $[e]=\left\{ a\in\underline{E}\mid a\leq e\right\} $,
we have $[e]$ finite for all $e\in\underline{E}$.
\item [{Heredity}] Whenever $a\#b\leq c$, we have $a\#c$.
\end{description}
\end{defn}

Many equivalent presentations of an event structure exists. In particular,
instead of providing the causality and conflict, one can provide the
immediate causality and minimal conflict, or one can provide the set
of configurations:
\begin{description}
\item [{Immediate~causality}] We write $a\rightarrowtriangle b$ whenever
$a<b$ and there is no $c$ such that $a<c<b$.
\item [{Minimal~conflict}] We write $a\sim b$ whenever $a\#b$ and there
is no $a'<a$ such that $a'\#b$, and no $b'<b$ such that $a\#b'$.
\item [{(Finite)~Configurations}] We write $C(E)$ for the set of finite
subsets $x\subseteq_{\textup{fin}}\underline{E}$ such that whenever
$a\leq b\in x$ we have $a\in x$, and whenever $a\#b\in x$ we have
$a\notin x$. Equivalently, whenever $a\rightarrowtriangle b\in x$
we have $a\in x$, and whenever $a\sim b\in x$ we have $a\notin x$.

$\vartriangleright$We say that a configuration $x$ \emph{enables}
an event $a\notin x$, and we write $x\ext a$ $x\cup\left\{ a\right\} $
is a configuration
\end{description}
\end{tcolorbox}

\paragraph{Polarities}

In game semantics, the positive polarity ($\oplus$, or ``Player'')
represents the program under analysis, while the negative polarity
($\ominus$, or ``Opponent'') models its environment---effectively,
“every program other than the one being studied.” Under this perspective,
negative events represent information received from the environment,
whereas positive events correspond to the program's output to it.
The neutral polarity ($0$) designates events that do not involve
external interaction, typically encoding internal computation or communication
within the program itself.

When considering sets of events, we write $x^{\ominus}$ for the restriction
of $x\subseteq E$ to only the events of negative polarity, and we
write $x\subseteq^{\ominus}y$ whenever $x\subseteq y$ and all the
events of $y\backslash x$ are negatives. We also use $(-)^{\oplus}$
for positive events and $(-)^{0,\oplus}$ for non-negative events.\\
We are more specifically interested in ``race-free'' event-structures
(see Definition\ref{def:quantum game}.$(1)$) , that guarantee that
the event structure is \emph{``deterministic'', }in the sense given
by \autocite{winskelDeterministicConcurrentStrategies2012}. Among
other properties, race-freeness describe exactly those event-structures
for which the quantum copy-cat strategy $c\mkern-12mu c$ is deterministic,
so that we have a probabilistic identity strategy (w.r.t. composition).
This property will be heavily used in Section \ref{subsec:Interaction of Quantum Petri Nets}
to define the interaction of two Quantum Petri Nets. 
\begin{example}
An example of Event Structure with polarities is given in on the left
of Figure \ref{fig:An Example of Quantum Event Structure}--ignoring
the quantum annotation $Q\left(\ldots\right)$ on its right, which
illustrate a notion seen further down.
\end{example}

\subsection{Quantum Event Structures}

In our case, \autocite{clairambaultConcurrentQuantumStrategies2019,clairambaultFullAbstractionQuantum2020}
assimilate basic Event Structures to game arenas, endowed with some
suitably constrained completely positive operators, in order to define
Quantum Event Structures. The latter are used to denote the paradigmatic
Quantum Lambda Calculus \autocite{selingerLambdaCalculusQuantum2006}
with a Game Semantics.

This justifies a definition of Quantum Occurrence Nets, by porting
the functorial correspondence between classical Occurrence Nets and
Events structures to de Visme's notion of Quantum Event Structures.
As such, we present a new framework for Quantum Occurrence Nets in
Sections \ref{sec:Global Quantum Occurrence Nets} and \ref{sec:Local Quantum Occurrence Nets}
immediately after this final introduction of Quantum Event structures.

\begin{tcolorbox}[breakable,enhanced,title=Quantum Game, colback=blue!5!white, colframe=blue!75!black]

\begin{defn}
\label{def:quantum game}A quantum game is a tuple $E=\left(\underline{E},\leq,\#,p,H\right)$
where $\left(\underline{E},\leq,\#,p\right)$ is a $\text{race-free}^{(1)}$
event structure with polarity, and $H:\underline{E}^{\ominus,\oplus}\to\mathsf{HilbertSpace}$
associates a finite dimensional Hilbert space to every event of non-neutral
polarity.

$(1)$ An event structure is race-free if only negative events can
be in conflict with another negative event. That is, $a\sim b\implies p\left(a\right)=p\left(b\right)=\ominus\,OR\,p\left(a\right)\neq\ominus\neq p\left(b\right)$
\end{defn}

\end{tcolorbox}

Quantum Event Structures decorate quantum games with quantum operators,
that we define as hermitian preserving-Completely Positive-Trace Non
Increasing Maps ( written $CPTNI$). An finite-dimensional operator
$M$ is positive if it has a real non negative spectrum. The Löwner
order $\sqsupseteq$ orders partially the category of Positive Maps,
such that $M\sqsupseteq N\iff M-N\text{ is positive}.$ The map $M$
is Completely positive if for every Hilbert space $H$, and every
(sub)density matrix $\rho$, $\left(M\otimes\Id H\right)\left(\rho\right)$
is positive. $M$ is Trace non increasing if $\rho\mapsto\tr\left(M\left(\rho\right)\right)$
is non increasing. $CPTNI$ maps forms a (non cartesian) symmetric
monoidal category. 
\begin{tcolorbox}[breakable,enhanced,title=Quantum Event Structure \autocite{clairambaultConcurrentQuantumStrategies2019},
colback=blue!5!white, colframe=blue!75!black]

\begin{defn}
\label{def:Quantum Event Structure}A quantum event structure is a
tuple $E=\left(\underline{E},\leq,\#,pol,H,Q\right)$ where $\left(\underline{E},\leq,\#,p,H\right)$
is a quantum game, and $Q$ is an operation on both configurations
and intervals of configurations such that:
\end{defn}

\begin{itemize}
\item For $x\in C(E)$, $Q(x)$ is a finite dimensional Hilbert space.
\item For $x,y\in C(E)$ such that $x\subseteq y$, 
\[
Q\left(x\subseteq y\right)\in CPTNI\left(Q\left(x\right)\otimes H\left(\left(y\backslash x\right)^{\ominus}\right),H\left(\left(y\backslash x\right)^{\oplus}\right)\otimes Q\left(y\right)\right)
\]
\end{itemize}
\begin{description}
\item [{Obliviousness:}] For $x\subseteq^{\ominus}y$, then $Q(y)=Q(x)\otimes{\displaystyle \bigotimes_{e\in y\backslash x}}H(e)$.
Additionally, $Q\left(x\subseteq^{\ominus}y\right)=\Id{Q(x)\otimes H\left(y\setminus x\right)}$.
When $x$ is the smallest configuration negatively included in $y$,
then those two equations are exact and not ``up to an implicit permutation''
\item [{Functoriality:}] For $x\in C(E)$, $Q\left(x\subseteq x\right)=\Id{Q(x)}$.
Additionally, whenever $x\subseteq y\subseteq z$, we have 
\[
Q(x\subseteq z)=\left(\Id{H\left(\left(y\backslash x\right)^{\oplus}\right)}\otimes Q\left(y\subseteq z\right)\right)\circ\left(Q\left(x\subseteq y\right)\otimes\Id{H\left(\left(z\backslash y\right)^{\ominus}\right)}\right)
\]
\item [{Drop~condition}] For$x,y_{1},\dots,y_{n}\in C(E)$, such that
$x\subseteq^{0,\oplus}y_{i}$ for all $1\leq i\leq n$, if we write
$y_{I}=\bigcup_{i\in I}y_{i}$ for $\varnothing\neq I\subseteq\{1,\dots,n\}$,
we have:
\[
d(x;y_{1},\dots,y_{n}):=\sum_{\substack{I\subseteq\{1,\dots,n\}\\
\text{s.t. }y_{I}\in C(E)
}
}(-1)^{|I|}\tr_{Q\left(y_{I}\right)\otimes H\left(y_{I}\backslash x\right)}\circ Q\left(x\subseteq y_{I}\right)\sqsupseteq0
\]
\end{description}
\end{tcolorbox}

The Obliviousness property transcribes the fact that the Event Structure
solely describes the behaviour of the program (i.e. the Player, $\oplus$
events), and is oblivious to the environment (i.e. the Opponent, $\ominus$
events). The functoriality property ensures compositionnality. Finally,
the Drop Condition axiomatizes the condition so that $Q$ defines
a proper quantum valuation. 

Indeed, Quantum Event Structures are derived from Probabilisitc Event
structures where a probability valuation $v:C(E)\rightarrow\left[0,1\right]$
is defined on \emph{every configuratio}n of $C(E)$--in a similar
way the completely positive operator $Q$ is, as they form Scott-open
sets on the directed-complete-partial-order defined by $\left(C(E),\supseteq\right)$.
When this valuation satisfies the drop condition, it gives it the
structure of a \emph{continuous} (resp. quantum) valuation on Scott-open
sets \autocite{winskelProbabilisticQuantumEvent2014}, which makes
$v$ (resp. $Q$) a proper probability (resp. sub-density) distribution
on the Borel algebra formed by configurations\autocite{alvarez-manillaExtensionResultContinuous1998,clairambaultConcurrentQuantumStrategies2019}.
Hence, $\tr Q$ defines a proper probability valuation.
\begin{example}
An example of Quantum Event Structure is given in Figure \ref{fig:An Example of Quantum Event Structure}.
\begin{figure}[H]
\centering
\begin{centering}

\tikzset{every picture/.style={line width=0.75pt}} 

\begin{tikzpicture}[x=0.75pt,y=0.75pt,yscale=-1,xscale=1]

\draw (71,52.4) node [anchor=north west][inner sep=0.75pt]    {$a^{\ominus }$};
\draw (31,102.4) node [anchor=north west][inner sep=0.75pt]    {$b^{0}$};
\draw (111,103.4) node [anchor=north west][inner sep=0.75pt]    {$c^{\oplus }$};
\draw (141,22.4) node [anchor=north west][inner sep=0.75pt]    {$ \begin{array}{l}
Q( \emptyset \subseteq \{a\}) =Id_{\mathbb{C}^{8}}\\
Q(\{a\} \subseteq \{a,b\}) =Id_{\mathbb{C}^{8}}\\
Q(\{a\} \subseteq \{a,c\}) =X\otimes Id_{\mathbb{C}^{4}}
\end{array}$};
\draw (141,102.4) node [anchor=north west][inner sep=0.75pt]    {$ \begin{array}{l}
Q( \emptyset ) =\mathbb{C}^{4}\\
Q( a) =\mathbb{C}^{8} \ \ H( a) =\mathbb{C}^{2}\\
Q( b) =\mathbb{C}^{8}\\
Q( c) =\mathbb{C}^{4} \ \ H( c) =\mathbb{C}^{2}
\end{array}$};
\draw    (69.84,74) -- (51.43,96.45) ;
\draw [shift={(50.16,98)}, rotate = 309.35] [color={rgb, 255:red, 0; green, 0; blue, 0 }  ][line width=0.75]    (10.93,-3.29) .. controls (6.95,-1.4) and (3.31,-0.3) .. (0,0) .. controls (3.31,0.3) and (6.95,1.4) .. (10.93,3.29)   ;
\draw    (90.57,74) -- (108.71,97.42) ;
\draw [shift={(109.93,99)}, rotate = 232.24] [color={rgb, 255:red, 0; green, 0; blue, 0 }  ][line width=0.75]    (10.93,-3.29) .. controls (6.95,-1.4) and (3.31,-0.3) .. (0,0) .. controls (3.31,0.3) and (6.95,1.4) .. (10.93,3.29)   ;
\draw    (51,111.14) .. controls (52.69,109.49) and (54.35,109.51) .. (56,111.2) .. controls (57.65,112.89) and (59.31,112.91) .. (61,111.27) .. controls (62.69,109.62) and (64.35,109.64) .. (66,111.33) .. controls (67.65,113.02) and (69.31,113.04) .. (71,111.39) .. controls (72.69,109.74) and (74.35,109.76) .. (76,111.45) .. controls (77.65,113.14) and (79.31,113.16) .. (81,111.52) .. controls (82.69,109.87) and (84.35,109.89) .. (86,111.58) .. controls (87.65,113.27) and (89.31,113.29) .. (91,111.64) .. controls (92.69,109.99) and (94.35,110.01) .. (96,111.7) .. controls (97.65,113.39) and (99.31,113.41) .. (101,111.76) .. controls (102.69,110.12) and (104.35,110.14) .. (106,111.83) -- (108,111.85) -- (108,111.85) ;

\end{tikzpicture}
\caption{\label{fig:An Example of Quantum Event Structure}An Example of Quantum
Event Structure {[}see Definition \ref{def:Quantum Event Structure}{]}.
It verifies the 3 axioms : Obliviousness; Functoriality--where the
definition of $Q\left(\emptyset\subseteq\left\{ a,b\right\} \right)$
and $Q\left(\emptyset\subseteq\left\{ a,c\right\} \right)$ is left
implicit so that the functoriality is verified; and the drop condition.}
\par\end{centering}
\end{figure}
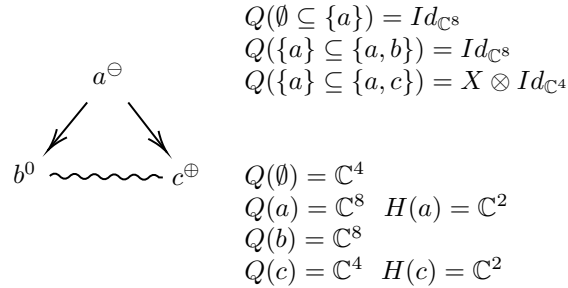
\end{example}

\section{Global Quantum Occurrence Nets\label{sec:Global Quantum Occurrence Nets}}

\global\long\def\ES#1{\mathbf{EventStructure}\left(#1\right)}%

\begin{problem}
Why a global annotation for Quantum Event Structures?\label{prob:Why-a-global}

De Visme's notion's of Quantum Event Structures defines a Global quantum
annotation $Q$, in the sense that the operator $Q$ is defined on
every interval of configuration $x\subseteq y$ instead of single
events. This allows to take into account entanglement phenomena, where
other attempts like Delbeque's in \autocite{delbecqueGameSemanticsQuantum2008}
imputed unreasonable restrictions with respect to this regard. Additionally,
as mentioned in the comment of Definition \ref{def:Quantum Event Structure},
it is the very fact that $Q$ is defined on configurations (i.e. Scott-open
sets) that allows it to be a proper quantum valuation.

However, global properties are unwieldy for semantic and verificational
purposes, as they involve checking a potentially exponential number
of cases. It is furthermore unnatural w.r.t. the reality of concurrent
computation, as one could expect a sensible semantic model for distributed
systems to represent the fact that each node (event of the event structure)
is independant of the one-another, and that the computations should
be local. 

For this reason, after deriving--from de Visme's Quantum Event Structures--a
global definition of Occurrence Nets in (this) Section \ref{sec:Global Quantum Occurrence Nets},
we establish, in Section \ref{sec:Local Quantum Occurrence Nets},
a framework for Local Quantum Occurrence Nets, where global entanglements
is still allowed globally but operations are defined locally.
\end{problem}

Taking advantage of the co-reflection between standard Event Structures
and Occurrence nets \autocite[Theorem 3.4.11]{winskelEventStructures1987},
a first approach is to directly transpose the quantum annotations
$Q$ and $H$ of a Quantum Event Structure onto an Occurrence Net
to obtain a Quantum Occurrence Net. This construction yields a framework
that we call Global Quantum Occurrence Net.

Let us denote by $\ES O$ the Event Structure associated to an occurrence
net $O$, obtained by discarding the conditions, while remembering
the causality relation along with the events in conflict ( forward
functor in \autocite{nielsenPetriNetsEvent1981}). Cuts are the maximal
elements for the causality order within a configuration. It is known
that each configuration corresponds to exactly one cut (and conversely).
But there is also a one-to-one correspondence between the configurations
$x\in C\left(\ES O\right)$ and the markings of $O$. We make use
of the latter in the following.

\begin{tcolorbox}[breakable,enhanced,title=Cut of a configuration and associated marking,
colback=cyan!5!white, colframe=cyan!75!black]

\begin{defn}
\label{def:Cut-of-a} Let $E$ be an event structure. To each configuration
$x\in C(E)$ corresponds an unique cut $\cut x$ - the set of maximal
elements of $x$ w.r.t. $\leq$ in the event structure.

By an abuse of notation, we identify $\cut x$ with the corresponding
reachable marking in the associated Occurrence Net : $\cutpn x:=\bigsqcup_{e\in\widearc x}\post{\boxed{e}}$
(the union is disjoint because in an occurrence net, a condition has
at most one pre-event). It is the set of post-conditions of the events
of the cut $\cut x$. \\

Conversely, if $\cutpn x$ is a marking of an Occurrence net, define
the associated canonical cut $\cut x:=\bigcup_{\ci a\in\cutpn x}\pre{\ci a}$
(the pre-events of the conditions of $\cutpn x$), along with the
associated configuration $x$ (the downward closure of $\cut x$ w.r.t
$\leq$). 
\end{defn}

\end{tcolorbox}

\begin{example}
An example of the correspondence between configurations, cuts and
markings is depicted in Figure \ref{fig:correspondance config and cut}.
\end{example}

\begin{figure}[H]
\centering
\begin{centering}
\subfloat[\label{fig:configurations of an event structure}An event structure,
represented with three configurations $x\subseteq y\subseteq z\in C(E)$.]{

\tikzset{every picture/.style={line width=0.75pt}} 

\begin{tikzpicture}[x=0.75pt,y=0.75pt,yscale=-1,xscale=1]

\draw  [draw opacity=0][fill={rgb, 255:red, 126; green, 211; blue, 33 }  ,fill opacity=1 ][blur shadow={shadow xshift=0.75pt,shadow yshift=0pt, shadow blur radius=2.25pt, shadow blur steps=4 ,shadow opacity=100}] (50.03,63.88) .. controls (62.73,65.28) and (192.33,54.48) .. (204.03,63.88) .. controls (215.73,73.28) and (207.33,225.48) .. (204.03,243.88) .. controls (200.73,262.28) and (116.33,269.48) .. (104.03,253.88) .. controls (91.73,238.28) and (106.73,153.28) .. (104.03,143.88) .. controls (101.33,134.48) and (49.72,142.42) .. (44.03,133.88) .. controls (38.34,125.35) and (37.33,62.48) .. (50.03,63.88) -- cycle ;
\draw  [draw opacity=0][fill={rgb, 255:red, 245; green, 166; blue, 35 }  ,fill opacity=1 ][blur shadow={shadow xshift=0.75pt,shadow yshift=0pt, shadow blur radius=2.25pt, shadow blur steps=4 ,shadow opacity=100}] (56.03,71.88) .. controls (68.73,73.28) and (188.33,62.48) .. (200.03,71.88) .. controls (211.73,81.28) and (196.33,160.48) .. (194.03,173.88) .. controls (191.73,187.28) and (119.73,194.28) .. (114.03,183.88) .. controls (108.33,173.48) and (116.73,143.28) .. (114.03,133.88) .. controls (111.33,124.48) and (59.72,132.42) .. (54.03,123.88) .. controls (48.34,115.35) and (43.33,70.48) .. (56.03,71.88) -- cycle ;
\draw  [draw opacity=0][fill={rgb, 255:red, 80; green, 227; blue, 194 }  ,fill opacity=1 ][blur shadow={shadow xshift=0.75pt,shadow yshift=0pt, shadow blur radius=2.25pt, shadow blur steps=4 ,shadow opacity=100}] (54.03,89.88) .. controls (54.03,86.57) and (56.71,83.88) .. (60.03,83.88) -- (194.03,83.88) .. controls (197.34,83.88) and (200.03,86.57) .. (200.03,89.88) -- (200.03,107.88) .. controls (200.03,111.2) and (197.34,113.88) .. (194.03,113.88) -- (60.03,113.88) .. controls (56.71,113.88) and (54.03,111.2) .. (54.03,107.88) -- cycle ;
\draw [line width=1.5]  [dash pattern={on 1.69pt off 2.76pt}]  (196,100) -- (240,100) ;
\draw [line width=1.5]  [dash pattern={on 1.69pt off 2.76pt}]  (194,150) -- (240,150) ;
\draw [line width=1.5]  [dash pattern={on 1.69pt off 2.76pt}]  (199,230) -- (240,230) ;

\draw (61.03,90.28) node [anchor=north west][inner sep=0.75pt]    {$a$};
\draw (59.03,162.28) node [anchor=north west][inner sep=0.75pt]    {$a'$};
\draw (57.03,234.28) node [anchor=north west][inner sep=0.75pt]    {$a''$};
\draw (115.03,90.28) node [anchor=north west][inner sep=0.75pt]    {$b$};
\draw (113.03,162.28) node [anchor=north west][inner sep=0.75pt]    {$b'$};
\draw (111.03,234.28) node [anchor=north west][inner sep=0.75pt]    {$b''$};
\draw (169.03,90.28) node [anchor=north west][inner sep=0.75pt]    {$c$};
\draw (168.03,162.28) node [anchor=north west][inner sep=0.75pt]    {$c'$};
\draw (165.03,234.28) node [anchor=north west][inner sep=0.75pt]    {$c''$};
\draw  [draw opacity=0][fill={rgb, 255:red, 80; green, 227; blue, 194 }  ,fill opacity=1 ]  (240,79) -- (259,79) -- (259,112) -- (240,112) -- cycle  ;
\draw (243,83.4) node [anchor=north west][inner sep=0.75pt]  [font=\Large]  {$x$};
\draw (307,61) node [anchor=north west][inner sep=0.75pt]  [font=\Large,rotate=-90] [align=left] {some configurations};
\draw  [draw opacity=0][fill={rgb, 255:red, 245; green, 166; blue, 35 }  ,fill opacity=1 ]  (240,129) -- (259,129) -- (259,162) -- (240,162) -- cycle  ;
\draw (243,133.4) node [anchor=north west][inner sep=0.75pt]  [font=\Large]  {$y$};
\draw  [draw opacity=0][fill={rgb, 255:red, 126; green, 211; blue, 33 }  ,fill opacity=1 ]  (241,208) -- (259,208) -- (259,236) -- (241,236) -- cycle  ;
\draw (244,212.4) node [anchor=north west][inner sep=0.75pt]  [font=\large]  {$z$};
\draw    (38,8) -- (218,8) -- (218,42) -- (38,42) -- cycle  ;
\draw (128,25) node  [font=\Large] [align=left] {Event Structure $\displaystyle E$};
\draw    (174.44,181.88) -- (174.13,226.88) ;
\draw [shift={(174.11,229.88)}, rotate = 270.4] [fill={rgb, 255:red, 0; green, 0; blue, 0 }  ][line width=0.08]  [draw opacity=0] (8.93,-4.29) -- (0,0) -- (8.93,4.29) -- cycle    ;
\draw    (174.94,109.88) -- (174.63,154.88) ;
\draw [shift={(174.61,157.88)}, rotate = 270.4] [fill={rgb, 255:red, 0; green, 0; blue, 0 }  ][line width=0.08]  [draw opacity=0] (8.93,-4.29) -- (0,0) -- (8.93,4.29) -- cycle    ;
\draw    (120.86,109.88) -- (120.23,154.88) ;
\draw [shift={(120.19,157.88)}, rotate = 270.8] [fill={rgb, 255:red, 0; green, 0; blue, 0 }  ][line width=0.08]  [draw opacity=0] (8.93,-4.29) -- (0,0) -- (8.93,4.29) -- cycle    ;
\draw    (120.11,181.88) -- (120.42,226.88) ;
\draw [shift={(120.44,229.88)}, rotate = 269.6] [fill={rgb, 255:red, 0; green, 0; blue, 0 }  ][line width=0.08]  [draw opacity=0] (8.93,-4.29) -- (0,0) -- (8.93,4.29) -- cycle    ;
\draw    (65.61,181.88) -- (65.92,226.88) ;
\draw [shift={(65.94,229.88)}, rotate = 269.6] [fill={rgb, 255:red, 0; green, 0; blue, 0 }  ][line width=0.08]  [draw opacity=0] (8.93,-4.29) -- (0,0) -- (8.93,4.29) -- cycle    ;
\draw    (66.78,109.88) -- (65.84,154.88) ;
\draw [shift={(65.78,157.88)}, rotate = 271.19] [fill={rgb, 255:red, 0; green, 0; blue, 0 }  ][line width=0.08]  [draw opacity=0] (8.93,-4.29) -- (0,0) -- (8.93,4.29) -- cycle    ;

\end{tikzpicture}
} \subfloat[\label{fig:cuts of configuration}The associated respective cuts $\protect\cut x$,
$\protect\cut y$ and $\protect\cut z$]{

\tikzset{every picture/.style={line width=0.75pt}} 

\begin{tikzpicture}[x=0.75pt,y=0.75pt,yscale=-1,xscale=1]

\draw [color={rgb, 255:red, 126; green, 211; blue, 33 }  ,draw opacity=1 ][line width=2.25]  [dash pattern={on 2.25pt off 2.25pt}]  (356,103) .. controls (358.33,102.53) and (359.74,103.44) .. (360.25,105.72) .. controls (360.9,108.03) and (362.35,108.87) .. (364.6,108.22) .. controls (366.85,107.53) and (368.35,108.32) .. (369.12,110.58) .. controls (369.83,112.78) and (371.35,113.51) .. (373.7,112.76) .. controls (375.87,111.91) and (377.4,112.57) .. (378.27,114.75) .. controls (379.24,116.94) and (380.74,117.53) .. (382.77,116.53) .. controls (384.96,115.56) and (386.6,116.14) .. (387.69,118.29) .. controls (388.67,120.36) and (390.26,120.86) .. (392.45,119.79) .. controls (394.41,118.62) and (395.93,119.04) .. (397.02,121.06) .. controls (398.37,123.11) and (400,123.49) .. (401.93,122.21) .. controls (403.98,120.92) and (405.62,121.26) .. (406.84,123.23) .. controls (407.01,125.38) and (408.18,126.62) .. (410.34,126.97) .. controls (412.55,128.08) and (413.06,129.67) .. (411.88,131.74) .. controls (410.51,133.43) and (410.78,135.12) .. (412.69,136.8) .. controls (414.51,138.36) and (414.64,139.88) .. (413.07,141.36) .. controls (411.49,143.31) and (411.55,144.99) .. (413.26,146.41) .. controls (414.95,148.46) and (414.95,150.29) .. (413.28,151.89) .. controls (411.59,153.5) and (411.56,155.17) .. (413.19,156.89) .. controls (414.8,158.67) and (414.75,160.41) .. (413.02,162.12) .. controls (411.3,163.24) and (411.23,164.75) .. (412.81,166.64) .. controls (414.38,168.58) and (414.28,170.44) .. (412.51,172.22) .. controls (410.77,173.38) and (410.67,174.96) .. (412.22,176.97) .. controls (413.8,178.36) and (413.7,179.96) .. (411.91,181.79) .. controls (410.12,183.63) and (410.01,185.26) .. (411.58,186.67) .. controls (413.11,188.74) and (413,190.38) .. (411.25,191.57) .. controls (409.46,193.43) and (409.35,195.07) .. (410.93,196.49) .. controls (412.46,198.56) and (412.36,200.19) .. (410.61,201.39) .. controls (408.82,203.24) and (408.72,204.87) .. (410.31,206.26) .. controls (411.87,208.28) and (411.77,210.2) .. (410,212.03) .. controls (408.27,213.22) and (408.2,214.79) .. (409.78,216.75) .. controls (411.37,218.67) and (411.32,220.21) .. (409.61,221.38) .. controls (407.89,223.15) and (407.85,224.94) .. (409.48,226.76) .. controls (411.13,228.52) and (411.11,230.25) .. (409.44,231.94) .. controls (407.79,233.63) and (407.82,235.27) .. (409.52,236.86) .. controls (411.25,238.36) and (411.32,239.91) .. (409.73,241.51) .. controls (408.17,243.11) and (408.31,244.78) .. (410.15,246.52) .. controls (412.02,248.04) and (412.27,249.74) .. (410.9,251.63) .. controls (409.61,253.55) and (409.96,255.01) .. (411.97,256.02) .. controls (414.07,256.92) and (414.76,258.48) .. (414.03,260.7) .. controls (414.04,263.31) and (415.18,264.44) .. (417.45,264.07) .. controls (419.71,263.43) and (421.21,264.17) .. (421.96,266.29) .. controls (423.12,268.41) and (424.79,268.87) .. (426.96,267.68) .. controls (428.73,266.3) and (430.3,266.53) .. (431.66,268.38) .. controls (433.23,270.17) and (434.94,270.27) .. (436.77,268.68) .. controls (438.54,267.02) and (440.18,267.01) .. (441.67,268.64) .. controls (443.66,270.21) and (445.36,270.1) .. (446.78,268.32) .. controls (448.55,266.47) and (450.31,266.27) .. (452.08,267.74) .. controls (453.92,269.16) and (455.52,268.92) .. (456.89,267.02) .. controls (458.21,265.09) and (459.84,264.79) .. (461.77,266.12) .. controls (463.76,267.4) and (465.39,267.04) .. (466.68,265.05) .. controls (467.91,263.04) and (469.55,262.63) .. (471.58,263.83) .. controls (473.65,264.99) and (475.27,264.53) .. (476.44,262.46) .. controls (477.55,260.38) and (479.14,259.88) .. (481.23,260.97) .. controls (483.34,262.02) and (484.9,261.48) .. (485.91,259.36) .. controls (486.84,257.23) and (488.35,256.66) .. (490.44,257.64) .. controls (492.55,258.59) and (494.17,257.91) .. (495.32,255.61) .. controls (496.01,253.47) and (497.38,252.84) .. (499.43,253.71) .. controls (501.82,254.38) and (503.41,253.56) .. (504.2,251.25) .. controls (504.83,248.98) and (506.27,248.13) .. (508.51,248.72) .. controls (510.76,249.24) and (512.13,248.29) .. (512.62,245.88) -- (516,243) ;
\draw [color={rgb, 255:red, 80; green, 227; blue, 194 }  ,draw opacity=1 ][line width=2.25]    (360,102) .. controls (371.47,109.78) and (388.55,115.75) .. (407.98,119.14) .. controls (427.41,122.53) and (499.94,131.06) .. (530,102) ;
\draw [color={rgb, 255:red, 245; green, 166; blue, 35 }  ,draw opacity=1 ][line width=2.25]  [dash pattern={on 3.75pt off 3pt on 7.5pt off 1.5pt}]  (360,102) .. controls (371.47,109.78) and (390.57,118.61) .. (410,122) .. controls (429.43,125.39) and (423.88,166.68) .. (420,192) .. controls (416.12,217.32) and (511.95,189.45) .. (530,172) ;

\draw (382.03,97.86) node [anchor=north west][inner sep=0.75pt]    {$a$};
\draw (380.03,169.86) node [anchor=north west][inner sep=0.75pt]    {$a'$};
\draw (378.03,241.86) node [anchor=north west][inner sep=0.75pt]    {$a''$};
\draw (436.03,97.86) node [anchor=north west][inner sep=0.75pt]    {$b$};
\draw (434.03,169.86) node [anchor=north west][inner sep=0.75pt]    {$b'$};
\draw (432.03,241.86) node [anchor=north west][inner sep=0.75pt]    {$b''$};
\draw (490.03,97.86) node [anchor=north west][inner sep=0.75pt]    {$c$};
\draw (489.03,169.86) node [anchor=north west][inner sep=0.75pt]    {$c'$};
\draw (486.03,241.86) node [anchor=north west][inner sep=0.75pt]    {$c''$};
\draw  [draw opacity=0][fill={rgb, 255:red, 80; green, 227; blue, 194 }  ,fill opacity=1 ]  (533,69) -- (552,69) -- (552,104) -- (533,104) -- cycle  ;
\draw (536,73.4) node [anchor=north west][inner sep=0.75pt]  [font=\Large]  {$\hat{x}$};
\draw    (378,19) -- (526,19) -- (526,52) -- (378,52) -- cycle  ;
\draw (381,23) node [anchor=north west][inner sep=0.75pt]  [font=\Large] [align=left] {associated cuts};
\draw  [draw opacity=0][fill={rgb, 255:red, 245; green, 166; blue, 35 }  ,fill opacity=1 ]  (533,147) -- (552,147) -- (552,182) -- (533,182) -- cycle  ;
\draw (536,151.4) node [anchor=north west][inner sep=0.75pt]  [font=\Large]  {$\hat{y}$};
\draw  [draw opacity=0][fill={rgb, 255:red, 126; green, 211; blue, 33 }  ,fill opacity=1 ]  (528,218) -- (546,218) -- (546,248) -- (528,248) -- cycle  ;
\draw (531,222.4) node [anchor=north west][inner sep=0.75pt]  [font=\large]  {$\hat{z}$};
\draw    (495.44,189.46) -- (495.13,234.46) ;
\draw [shift={(495.11,237.46)}, rotate = 270.4] [fill={rgb, 255:red, 0; green, 0; blue, 0 }  ][line width=0.08]  [draw opacity=0] (8.93,-4.29) -- (0,0) -- (8.93,4.29) -- cycle    ;
\draw    (495.94,117.46) -- (495.63,162.46) ;
\draw [shift={(495.61,165.46)}, rotate = 270.4] [fill={rgb, 255:red, 0; green, 0; blue, 0 }  ][line width=0.08]  [draw opacity=0] (8.93,-4.29) -- (0,0) -- (8.93,4.29) -- cycle    ;
\draw    (441.86,117.46) -- (441.23,162.46) ;
\draw [shift={(441.19,165.46)}, rotate = 270.8] [fill={rgb, 255:red, 0; green, 0; blue, 0 }  ][line width=0.08]  [draw opacity=0] (8.93,-4.29) -- (0,0) -- (8.93,4.29) -- cycle    ;
\draw    (441.11,189.46) -- (441.42,234.46) ;
\draw [shift={(441.44,237.46)}, rotate = 269.6] [fill={rgb, 255:red, 0; green, 0; blue, 0 }  ][line width=0.08]  [draw opacity=0] (8.93,-4.29) -- (0,0) -- (8.93,4.29) -- cycle    ;
\draw    (386.61,189.46) -- (386.92,234.46) ;
\draw [shift={(386.94,237.46)}, rotate = 269.6] [fill={rgb, 255:red, 0; green, 0; blue, 0 }  ][line width=0.08]  [draw opacity=0] (8.93,-4.29) -- (0,0) -- (8.93,4.29) -- cycle    ;
\draw    (387.78,117.46) -- (386.84,162.46) ;
\draw [shift={(386.78,165.46)}, rotate = 271.19] [fill={rgb, 255:red, 0; green, 0; blue, 0 }  ][line width=0.08]  [draw opacity=0] (8.93,-4.29) -- (0,0) -- (8.93,4.29) -- cycle    ;

\end{tikzpicture}
}
\par\end{centering}
\caption{\label{fig:correspondance config and cut}}
\end{figure}

\begin{center}
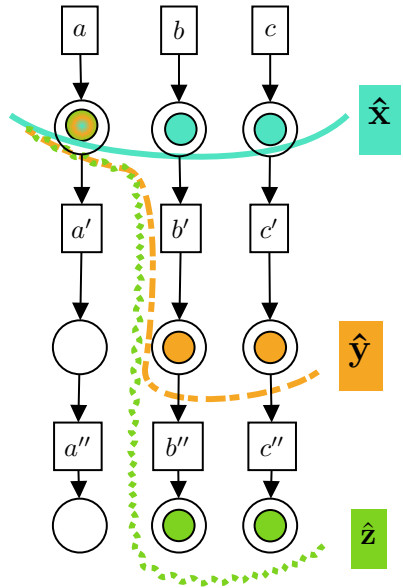
\begin{figure}[H]
\centering

  
\tikzset {_0ipcj92r0/.code = {\pgfsetadditionalshadetransform{ \pgftransformshift{\pgfpoint{0 bp } { 0 bp }  }  \pgftransformscale{1 }  }}}
\pgfdeclareradialshading{_smds8cf5a}{\pgfpoint{0bp}{0bp}}{rgb(0bp)=(0.31,0.89,0.76);
rgb(0bp)=(0.31,0.89,0.76);
rgb(12.5bp)=(0.96,0.65,0.14);
rgb(25bp)=(0.49,0.83,0.13);
rgb(400bp)=(0.49,0.83,0.13)}
\tikzset{every picture/.style={line width=0.75pt}} 

\begin{tikzpicture}[x=0.75pt,y=0.75pt,yscale=-1,xscale=1]

\path  [shading=_smds8cf5a,_0ipcj92r0] (47.98,79.88) .. controls (47.98,75.5) and (51.53,71.95) .. (55.91,71.95) .. controls (60.29,71.95) and (63.84,75.5) .. (63.84,79.88) .. controls (63.84,84.26) and (60.29,87.81) .. (55.91,87.81) .. controls (51.53,87.81) and (47.98,84.26) .. (47.98,79.88) -- cycle ; 
 \draw   (47.98,79.88) .. controls (47.98,75.5) and (51.53,71.95) .. (55.91,71.95) .. controls (60.29,71.95) and (63.84,75.5) .. (63.84,79.88) .. controls (63.84,84.26) and (60.29,87.81) .. (55.91,87.81) .. controls (51.53,87.81) and (47.98,84.26) .. (47.98,79.88) -- cycle ; 

\draw  [fill={rgb, 255:red, 80; green, 227; blue, 194 }  ,fill opacity=1 ] (97.98,81.88) .. controls (97.98,77.5) and (101.53,73.95) .. (105.91,73.95) .. controls (110.29,73.95) and (113.84,77.5) .. (113.84,81.88) .. controls (113.84,86.26) and (110.29,89.81) .. (105.91,89.81) .. controls (101.53,89.81) and (97.98,86.26) .. (97.98,81.88) -- cycle ;
\draw  [fill={rgb, 255:red, 80; green, 227; blue, 194 }  ,fill opacity=1 ] (143.11,82.01) .. controls (143.11,77.63) and (146.66,74.08) .. (151.04,74.08) .. controls (155.43,74.08) and (158.98,77.63) .. (158.98,82.01) .. controls (158.98,86.39) and (155.43,89.95) .. (151.04,89.95) .. controls (146.66,89.95) and (143.11,86.39) .. (143.11,82.01) -- cycle ;
\draw  [fill={rgb, 255:red, 245; green, 166; blue, 35 }  ,fill opacity=1 ] (96.98,192.01) .. controls (96.98,187.63) and (100.53,184.08) .. (104.91,184.08) .. controls (109.29,184.08) and (112.84,187.63) .. (112.84,192.01) .. controls (112.84,196.39) and (109.29,199.95) .. (104.91,199.95) .. controls (100.53,199.95) and (96.98,196.39) .. (96.98,192.01) -- cycle ;
\draw  [fill={rgb, 255:red, 245; green, 166; blue, 35 }  ,fill opacity=1 ] (143.11,192.01) .. controls (143.11,187.63) and (146.66,184.08) .. (151.04,184.08) .. controls (155.43,184.08) and (158.98,187.63) .. (158.98,192.01) .. controls (158.98,196.39) and (155.43,199.95) .. (151.04,199.95) .. controls (146.66,199.95) and (143.11,196.39) .. (143.11,192.01) -- cycle ;
\draw  [fill={rgb, 255:red, 126; green, 211; blue, 33 }  ,fill opacity=1 ] (96.98,281.88) .. controls (96.98,277.5) and (100.53,273.95) .. (104.91,273.95) .. controls (109.29,273.95) and (112.84,277.5) .. (112.84,281.88) .. controls (112.84,286.26) and (109.29,289.81) .. (104.91,289.81) .. controls (100.53,289.81) and (96.98,286.26) .. (96.98,281.88) -- cycle ;
\draw  [fill={rgb, 255:red, 126; green, 211; blue, 33 }  ,fill opacity=1 ] (143.11,282.01) .. controls (143.11,277.63) and (146.66,274.08) .. (151.04,274.08) .. controls (155.43,274.08) and (158.98,277.63) .. (158.98,282.01) .. controls (158.98,286.39) and (155.43,289.95) .. (151.04,289.95) .. controls (146.66,289.95) and (143.11,286.39) .. (143.11,282.01) -- cycle ;
\draw [color={rgb, 255:red, 80; green, 227; blue, 194 }  ,draw opacity=1 ][line width=2.25]    (20,74.81) .. controls (31.47,82.58) and (48.55,88.56) .. (67.98,91.95) .. controls (84.16,94.77) and (137.17,101.16) .. (171.81,86.43) .. controls (178.76,83.47) and (184.97,79.67) .. (190,74.81) ;
\draw [color={rgb, 255:red, 245; green, 166; blue, 35 }  ,draw opacity=1 ][line width=2.25]  [dash pattern={on 3.75pt off 3pt on 7.5pt off 1.5pt}]  (27.98,81.95) .. controls (39.45,89.72) and (58.55,98.56) .. (77.98,101.95) .. controls (97.41,105.34) and (91.86,176.63) .. (87.98,201.95) .. controls (84.09,227.26) and (159.92,219.4) .. (177.98,201.95) ;
\draw [color={rgb, 255:red, 126; green, 211; blue, 33 }  ,draw opacity=1 ][line width=2.25]  [dash pattern={on 2.25pt off 2.25pt}]  (27.98,81.95) .. controls (30.31,81.48) and (31.72,82.39) .. (32.23,84.66) .. controls (32.88,86.98) and (34.33,87.81) .. (36.58,87.16) .. controls (38.83,86.47) and (40.33,87.26) .. (41.09,89.52) .. controls (41.8,91.72) and (43.33,92.45) .. (45.68,91.71) .. controls (47.85,90.85) and (49.38,91.51) .. (50.25,93.7) .. controls (51.22,95.89) and (52.72,96.48) .. (54.75,95.48) .. controls (56.94,94.51) and (58.58,95.09) .. (59.67,97.23) .. controls (60.65,99.3) and (62.24,99.81) .. (64.43,98.74) .. controls (66.39,97.57) and (67.91,97.99) .. (68.99,100) .. controls (70.34,102.05) and (71.98,102.43) .. (73.91,101.16) .. controls (75.96,99.87) and (77.59,100.23) .. (78.82,102.22) .. controls (78.73,104.37) and (79.72,105.63) .. (81.79,106) .. controls (83.96,107.21) and (84.45,108.85) .. (83.24,110.9) .. controls (81.93,112.97) and (82.21,114.62) .. (84.07,115.85) .. controls (85.92,117.28) and (86.09,118.93) .. (84.59,120.81) .. controls (83.06,122.73) and (83.19,124.59) .. (84.96,126.38) .. controls (86.68,127.64) and (86.74,129.33) .. (85.15,131.45) .. controls (83.52,132.9) and (83.55,134.34) .. (85.24,135.76) .. controls (86.93,138) and (86.94,139.89) .. (85.27,141.44) .. controls (83.6,143) and (83.59,144.59) .. (85.24,146.2) .. controls (86.89,147.85) and (86.86,149.49) .. (85.17,151.12) .. controls (83.47,152.77) and (83.43,154.46) .. (85.06,156.2) .. controls (86.68,157.98) and (86.63,159.72) .. (84.92,161.42) .. controls (83.19,163.14) and (83.13,164.92) .. (84.74,166.77) .. controls (86.34,168.64) and (86.29,169.99) .. (84.59,170.84) .. controls (82.86,172.61) and (82.78,174.44) .. (84.37,176.35) .. controls (85.96,178.28) and (85.89,180.15) .. (84.14,181.94) .. controls (82.43,182.81) and (82.37,184.22) .. (83.95,186.18) .. controls (85.53,188.15) and (85.44,190.04) .. (83.69,191.86) .. controls (81.98,192.73) and (81.92,194.16) .. (83.5,196.14) .. controls (85.07,198.13) and (84.98,200.04) .. (83.23,201.87) .. controls (81.52,202.75) and (81.46,204.18) .. (83.03,206.16) .. controls (84.61,208.15) and (84.53,210.05) .. (82.78,211.88) .. controls (81.07,212.76) and (81,214.18) .. (82.59,216.15) .. controls (84.17,218.11) and (84.09,220) .. (82.35,221.81) .. controls (80.64,222.68) and (80.59,224.09) .. (82.18,226.02) .. controls (83.77,227.94) and (83.71,229.79) .. (81.98,231.57) .. controls (80.25,233.34) and (80.2,234.71) .. (81.84,235.68) .. controls (83.45,237.53) and (83.4,239.33) .. (81.68,241.07) .. controls (79.97,242.8) and (79.93,244.55) .. (81.56,246.34) .. controls (83.19,248.09) and (83.16,249.81) .. (81.47,251.49) .. controls (79.78,253.15) and (79.77,254.81) .. (81.42,256.48) .. controls (83.09,258.11) and (83.09,259.72) .. (81.43,261.32) .. controls (79.78,262.9) and (79.79,264.45) .. (81.48,265.97) .. controls (83.17,267.43) and (83.22,269.28) .. (81.62,271.51) .. controls (80.01,273) and (80.08,274.39) .. (81.81,275.69) .. controls (83.59,277.57) and (83.7,279.2) .. (82.13,280.59) .. controls (80.64,282.59) and (80.81,284.38) .. (82.66,285.95) .. controls (84.53,287.33) and (84.77,288.9) .. (83.37,290.67) .. controls (82.13,292.9) and (82.54,294.65) .. (84.61,295.9) .. controls (86.66,296.67) and (87.27,298.08) .. (86.46,300.15) .. controls (86.41,302.63) and (87.6,303.73) .. (90.02,303.44) .. controls (92.31,302.85) and (93.79,303.64) .. (94.44,305.8) .. controls (95.41,307.97) and (96.98,308.53) .. (99.16,307.48) .. controls (101.25,306.31) and (102.85,306.7) .. (103.95,308.64) .. controls (105.55,310.61) and (107.27,310.87) .. (109.1,309.44) .. controls (110.87,307.94) and (112.52,308.08) .. (114.05,309.86) .. controls (115.7,311.59) and (117.41,311.64) .. (119.18,310) .. controls (120.89,308.31) and (122.47,308.28) .. (123.92,309.9) .. controls (125.81,311.46) and (127.6,311.34) .. (129.27,309.54) .. controls (130.52,307.74) and (132.14,307.56) .. (134.12,309) .. controls (135.8,310.45) and (137.41,310.2) .. (138.96,308.27) .. controls (140.42,306.3) and (142.02,305.99) .. (143.75,307.34) .. controls (145.88,308.57) and (147.62,308.15) .. (148.96,306.1) .. controls (149.87,304.13) and (151.39,303.69) .. (153.52,304.8) .. controls (155.69,305.85) and (157.32,305.31) .. (158.39,303.16) .. controls (159.33,301.02) and (160.86,300.41) .. (162.99,301.33) .. controls (165.16,302.18) and (166.59,301.5) .. (167.28,299.3) .. controls (168.05,296.99) and (169.6,296.1) .. (171.93,296.64) .. controls (174.06,297.22) and (175.4,296.25) .. (175.95,293.74) -- (177.98,291.95) ;

\draw    (45.98,19.95) -- (63.98,19.95) -- (63.98,43.95) -- (45.98,43.95) -- cycle  ;
\draw (54.98,31.95) node    {$a$};
\draw    (45.98,119.95) -- (65.98,119.95) -- (65.98,143.95) -- (45.98,143.95) -- cycle  ;
\draw (55.98,131.95) node    {$a'$};
\draw    (41.98,229.95) -- (65.98,229.95) -- (65.98,253.95) -- (41.98,253.95) -- cycle  ;
\draw (53.98,241.95) node    {$a''$};
\draw    (55.98, 80.95) circle [x radius= 13.6, y radius= 13.6]   ;
\draw (49.98,73.35) node [anchor=north west][inner sep=0.75pt]    {$$};
\draw    (54.98, 191.95) circle [x radius= 13.6, y radius= 13.6]   ;
\draw (48.98,184.35) node [anchor=north west][inner sep=0.75pt]    {$$};
\draw    (95.98,19.95) -- (113.98,19.95) -- (113.98,43.95) -- (95.98,43.95) -- cycle  ;
\draw (104.98,31.95) node    {$b$};
\draw    (95.98,119.95) -- (115.98,119.95) -- (115.98,143.95) -- (95.98,143.95) -- cycle  ;
\draw (105.98,131.95) node    {$b'$};
\draw    (91.98,229.95) -- (116.98,229.95) -- (116.98,253.95) -- (91.98,253.95) -- cycle  ;
\draw (104.48,241.95) node    {$b''$};
\draw    (104.98, 81.95) circle [x radius= 13.6, y radius= 13.6]   ;
\draw (98.98,74.35) node [anchor=north west][inner sep=0.75pt]    {$$};
\draw    (104.98, 191.95) circle [x radius= 13.6, y radius= 13.6]   ;
\draw (98.98,184.35) node [anchor=north west][inner sep=0.75pt]    {$$};
\draw    (141.98,19.95) -- (159.98,19.95) -- (159.98,43.95) -- (141.98,43.95) -- cycle  ;
\draw (150.98,31.95) node    {$c$};
\draw    (140.98,119.95) -- (159.98,119.95) -- (159.98,143.95) -- (140.98,143.95) -- cycle  ;
\draw (150.48,131.95) node    {$c'$};
\draw    (139.98,229.95) -- (163.98,229.95) -- (163.98,253.95) -- (139.98,253.95) -- cycle  ;
\draw (151.98,241.95) node    {$c''$};
\draw    (150.98, 81.95) circle [x radius= 13.6, y radius= 13.6]   ;
\draw (144.98,74.35) node [anchor=north west][inner sep=0.75pt]    {$$};
\draw    (150.98, 191.95) circle [x radius= 13.6, y radius= 13.6]   ;
\draw (144.98,184.35) node [anchor=north west][inner sep=0.75pt]    {$$};
\draw    (54.98, 281.95) circle [x radius= 13.6, y radius= 13.6]   ;
\draw (48.98,274.35) node [anchor=north west][inner sep=0.75pt]    {$$};
\draw    (104.98, 281.95) circle [x radius= 13.6, y radius= 13.6]   ;
\draw (98.98,274.35) node [anchor=north west][inner sep=0.75pt]    {$$};
\draw    (150.98, 281.95) circle [x radius= 13.6, y radius= 13.6]   ;
\draw (144.98,274.35) node [anchor=north west][inner sep=0.75pt]    {$$};
\draw  [draw opacity=0][fill={rgb, 255:red, 80; green, 227; blue, 194 }  ,fill opacity=1 ]  (195.98,59.95) -- (216.98,59.95) -- (216.98,94.95) -- (195.98,94.95) -- cycle  ;
\draw (198.98,64.35) node [anchor=north west][inner sep=0.75pt]  [font=\Large]  {$\mathbf{\hat{x}}$};
\draw  [draw opacity=0][fill={rgb, 255:red, 245; green, 166; blue, 35 }  ,fill opacity=1 ]  (185.98,178.95) -- (207.98,178.95) -- (207.98,213.95) -- (185.98,213.95) -- cycle  ;
\draw (188.98,183.35) node [anchor=north west][inner sep=0.75pt]  [font=\Large]  {$\mathbf{\hat{y}}$};
\draw  [draw opacity=0][fill={rgb, 255:red, 126; green, 211; blue, 33 }  ,fill opacity=1 ]  (191.98,273.95) -- (209.98,273.95) -- (209.98,303.95) -- (191.98,303.95) -- cycle  ;
\draw (194.98,278.35) node [anchor=north west][inner sep=0.75pt]  [font=\large]  {$\hat{\mathbf{z}}$};
\draw    (55.22,43.95) -- (55.64,64.35) ;
\draw [shift={(55.7,67.35)}, rotate = 268.83] [fill={rgb, 255:red, 0; green, 0; blue, 0 }  ][line width=0.08]  [draw opacity=0] (8.93,-4.29) -- (0,0) -- (8.93,4.29) -- cycle    ;
\draw    (55.98,94.55) -- (55.98,116.95) ;
\draw [shift={(55.98,119.95)}, rotate = 270] [fill={rgb, 255:red, 0; green, 0; blue, 0 }  ][line width=0.08]  [draw opacity=0] (8.93,-4.29) -- (0,0) -- (8.93,4.29) -- cycle    ;
\draw    (55.78,143.95) -- (55.25,175.35) ;
\draw [shift={(55.2,178.35)}, rotate = 270.95] [fill={rgb, 255:red, 0; green, 0; blue, 0 }  ][line width=0.08]  [draw opacity=0] (8.93,-4.29) -- (0,0) -- (8.93,4.29) -- cycle    ;
\draw    (54.71,205.55) -- (54.28,226.95) ;
\draw [shift={(54.22,229.95)}, rotate = 271.15] [fill={rgb, 255:red, 0; green, 0; blue, 0 }  ][line width=0.08]  [draw opacity=0] (8.93,-4.29) -- (0,0) -- (8.93,4.29) -- cycle    ;
\draw    (104.98,43.95) -- (104.98,65.34) ;
\draw [shift={(104.98,68.34)}, rotate = 270] [fill={rgb, 255:red, 0; green, 0; blue, 0 }  ][line width=0.08]  [draw opacity=0] (8.93,-4.29) -- (0,0) -- (8.93,4.29) -- cycle    ;
\draw    (105.25,95.55) -- (105.68,116.95) ;
\draw [shift={(105.74,119.95)}, rotate = 268.85] [fill={rgb, 255:red, 0; green, 0; blue, 0 }  ][line width=0.08]  [draw opacity=0] (8.93,-4.29) -- (0,0) -- (8.93,4.29) -- cycle    ;
\draw    (105.78,143.95) -- (105.25,175.35) ;
\draw [shift={(105.2,178.35)}, rotate = 270.95] [fill={rgb, 255:red, 0; green, 0; blue, 0 }  ][line width=0.08]  [draw opacity=0] (8.93,-4.29) -- (0,0) -- (8.93,4.29) -- cycle    ;
\draw    (104.84,205.55) -- (104.63,226.95) ;
\draw [shift={(104.6,229.95)}, rotate = 270.57] [fill={rgb, 255:red, 0; green, 0; blue, 0 }  ][line width=0.08]  [draw opacity=0] (8.93,-4.29) -- (0,0) -- (8.93,4.29) -- cycle    ;
\draw    (150.98,43.95) -- (150.98,65.34) ;
\draw [shift={(150.98,68.34)}, rotate = 270] [fill={rgb, 255:red, 0; green, 0; blue, 0 }  ][line width=0.08]  [draw opacity=0] (8.93,-4.29) -- (0,0) -- (8.93,4.29) -- cycle    ;
\draw    (150.84,95.55) -- (150.63,116.95) ;
\draw [shift={(150.6,119.95)}, rotate = 270.57] [fill={rgb, 255:red, 0; green, 0; blue, 0 }  ][line width=0.08]  [draw opacity=0] (8.93,-4.29) -- (0,0) -- (8.93,4.29) -- cycle    ;
\draw    (150.58,143.95) -- (150.84,175.35) ;
\draw [shift={(150.86,178.35)}, rotate = 269.52] [fill={rgb, 255:red, 0; green, 0; blue, 0 }  ][line width=0.08]  [draw opacity=0] (8.93,-4.29) -- (0,0) -- (8.93,4.29) -- cycle    ;
\draw    (151.25,205.55) -- (151.68,226.95) ;
\draw [shift={(151.74,229.95)}, rotate = 268.85] [fill={rgb, 255:red, 0; green, 0; blue, 0 }  ][line width=0.08]  [draw opacity=0] (8.93,-4.29) -- (0,0) -- (8.93,4.29) -- cycle    ;
\draw    (151.68,253.95) -- (151.39,265.35) ;
\draw [shift={(151.32,268.35)}, rotate = 271.43] [fill={rgb, 255:red, 0; green, 0; blue, 0 }  ][line width=0.08]  [draw opacity=0] (8.93,-4.29) -- (0,0) -- (8.93,4.29) -- cycle    ;
\draw    (104.63,253.95) -- (104.77,265.35) ;
\draw [shift={(104.81,268.35)}, rotate = 269.28] [fill={rgb, 255:red, 0; green, 0; blue, 0 }  ][line width=0.08]  [draw opacity=0] (8.93,-4.29) -- (0,0) -- (8.93,4.29) -- cycle    ;
\draw    (54.28,253.95) -- (54.56,265.35) ;
\draw [shift={(54.64,268.35)}, rotate = 268.57] [fill={rgb, 255:red, 0; green, 0; blue, 0 }  ][line width=0.08]  [draw opacity=0] (8.93,-4.29) -- (0,0) -- (8.93,4.29) -- cycle    ;

\end{tikzpicture}
\caption{\label{fig:occurrence net with cut and marking}The Occurrence Net
associated with the event structure $E$. Cuts (and thus configurations)
in the event structure directly correspond to markings in the Occurrence
net. The marking composed of the blue dots (along the solid line)
corresponds to the cut $\protect\cutpn x$, etc.}
\end{figure}
\par\end{center}

Intervals of markings will also prove to be useful. If $\m{}\rightarrow^{*}\m{}'$
are two markings of an Occurrence Net, the interval $\int{\m{}}{\m{}'}$
represents the set of successive conditions used to go from $\m{}$
to $\m{}'$. The Occurrence Net structure enforces that the sequence
of successively reachable markings from $\m{}$ to $\m{}'$ is unique
(up to permutations).The interval $\int{\m{}}{\m{}'}$ is hence well
defined :

\begin{tcolorbox}[breakable,enhanced,title=Interval of markings ,colback=cyan!5!white,
colframe=cyan!75!black]

\begin{defn}
Let $\m{}$ and $\m{}'$ be two markings of a Petri Net such that
$\m{}=\m 0\overset{\boxed{e_{0}}}{\rightarrow}\m 1\overset{\boxed{e_{1}}}{\rightarrow}\ldots\overset{\boxed{e_{n-1}}}{\rightarrow}\m n=\m{}'$.
Define the interval of markings $\int{\m{}}{\m{}'}:=\left\{ \ci a\in\m i\mid i\in\left\llbracket 0,n\right\rrbracket \right\} $.
The set of transitions for a given interval is also unique for Occurrence
Nets : $\sigma\int{\m{}}{\m{}'}:=\left\{ \boxed{e_{0}},\ldots,\boxed{e_{n-1}}\right\} $
\end{defn}

\end{tcolorbox}

\begin{example}
In Figure \ref{fig:occurrence net with cut and marking}, $\int{\cutpn x}{\cutpn y}$
is the set of places of $\cutpn x\cup\cutpn y$. Similarily, $\int{\cutpn x}{\cutpn z}=\cutpn x\cup\cutpn y\cup\cutpn z$
.
\end{example}

We can now define Global Quantum Occurrence Nets, as an immediate
parallel with Quantum Event Structures (Defintion \ref{def:Quantum Event Structure}).
\begin{tcolorbox}[breakable,enhanced,title=Global Quantum Occurrence Net\autocite{clairambaultConcurrentQuantumStrategies2019},
colback=cyan!5!white, colframe=cyan!75!black]

\begin{defn}
\label{def:Global Quantum Occurrence Net}A Global Quantum Occurrence
Net is a tuple $O=\left(P,T,F,\text{\ensuremath{\cutpn{\emptyset}} },pol,Q,H\right)$
where $\left(P,T,F,\cutpn{\emptyset}\right)$ is an Occurrence Net
with initial marking $\cutpn{\emptyset}$, and $Q$ is a $CPTNI$
on both markings and intervals of markings such that:
\end{defn}

\begin{itemize}
\item For $\cutpn x\in\mathsf{Marking}\left(O\right)$, $Q\left(\cutpn x\right)$
is a finite dimensional Hilbert space.
\item For $\cutpn x,\cutpn y\in\mathsf{Marking}\left(O\right)$ such that
$\cutpn x\rightarrow^{*}\cutpn y$, writing $\sigma_{x,y}:=\sigma\int{\cutpn x}{\cutpn y}$
for brevity, we have
\[
Q\left(\int{\cutpn x}{\cutpn y}\right)\in CPTNI\left(Q\left(\cutpn x\right)\otimes H\left(\sigma_{x,y}^{\ominus}\right)\rightarrow H\left(\sigma_{x,y}^{\oplus}\right)\otimes Q\left(\cutpn y\right)\right)
\]
\end{itemize}
\begin{description}
\item [{Obliviousness:}] For $\cutpn x\rightarrow^{*,\ominus}\cutpn y$,
then $Q\left(\cutpn y\right)=Q\left(\cutpn x\right)\otimes\bigotimes_{\boxed{e}\in\sigma_{x,y}}H\left(\boxed{e}\right)$.
Additionally, $Q\left(\cutpn x\rightarrow^{*,\ominus}\cutpn y\right)=\Id{Q\left(\cutpn x\right)\otimes H\left(\sigma_{x,y}\right)}$.
\item [{Functoriality:}] For $\cutpn x\in\mathsf{Marking}\left(O\right)$,
$Q\left(\int{\cutpn x}{\cutpn x}\right)=\Id{Q\left(\cutpn x\right)}$.
Additionally, whenever $\cutpn x\rightarrow^{*}\cutpn y\rightarrow^{*}\cutpn z$,
we have 
\[
Q\left(\int{\cutpn x}{\cutpn z}\right)=\left(\Id{H\left(\sigma_{x,y}^{\oplus}\right)}\otimes Q\left(\int{\cutpn y}{\cutpn z}\right)\right)\circ\left(Q\left(\int{\cutpn x}{\cutpn y}\right)\otimes\Id{H\left(\sigma_{x,y}^{\ominus}\right)}\right)
\]
\item [{Drop~condition}] For $\cutpn x,\cutpn{y_{1}},\dots,\cutpn{y_{n}}\in\mathsf{Marking}\left(O\right)$,
such that : $\text{\ensuremath{\forall}}i\in\left\llbracket 1,n\right\rrbracket ,\left\{ \begin{array}{ll}
 & \cutpn x\rightarrow^{*}\cutpn{y_{i}}\\
\forall\boxed{e}\in\sigma\int{\cutpn x}{\cutpn{y_{i}}}, & pol\left(\boxed{e}\right)\in\left\{ 0,\oplus\right\} 
\end{array}\right.$, we have:
\[
\qdv x{y_{1},\ldots,y_{n}}:=\sum_{\substack{I\subseteq\{1,\dots,n\}\\
\text{s.t. }\cutpn{y_{I}}\in\mathsf{Marking}\left(O\right)
}
}(-1)^{|I|}\tr_{Q\left(\cutpn{y_{I}}\right)\otimes H\left(\sigma\int{\cutpn x}{\cutpn{y_{I}}}\right)}\circ Q\left(\int{\cutpn x}{\cutpn{y_{I}}}\right)\sqsupseteq0
\]
 where for $I\subseteq\left\llbracket 1,n\right\rrbracket $, ${\displaystyle \cutpn{y_{I}}:=\mathsf{Maximal_{\leq}}\left(\bigcup_{i\in I}\cutpn{y_{i}}\right)=\left\{ \ci a\in\bigcup_{i\in I}\cutpn{y_{i}}\mid\nexists\ci b\in\bigcup_{i\in I}\cutpn{y_{i}},\ci a<\ci b\right\} }$.
\end{description}
\end{tcolorbox}

\section{Local Quantum Occurrence Nets\label{sec:Local Quantum Occurrence Nets}}

We address here the Problem raised in \ref{prob:Why-a-global} with
global annotations. Localization of operations on events is made possible
possible by the Petri Net structure.

The idea is to define the completely positive annotation $Q$ of a
Global Occurrence Net \emph{from} an association of local CPTNIs $Q_{0}\left(\boxed{e}\right)$
defined on single events $\boxed{e}$. $Q_{0}$ needs to respect local
properties (namely Local Functoriality in \ref{subsec:Local Functoriality},
Local Obliviousness in \ref{subsec:Local Obliviousness} and Local
Drop Condition in \ref{subsec:Local Drop Condition}) for $Q$ to
define a valid Global Quantum Occurrence Net.

We first define Local Annotations of Nets, and explicit how the Global
annotation is induced from the latter (Section \ref{subsec:Local Annotation of an Occurrence Net}).
Then, we show in Sections \ref{subsec:Local Functoriality}, \ref{subsec:Local Obliviousness}
and \ref{subsec:Local Drop Condition} that if $Q_{0}$ respects a
local condition, then the global annotation $Q$ defined from it respects
the associated global condition. This defines a Local Occurrence Net.

\subsection{Local Annotation of an Occurrence Net\label{subsec:Local Annotation of an Occurrence Net}}

\begin{tcolorbox}[breakable,enhanced,title=Annotation of a Net Skeleton,colback=lime!5!white,
colframe=lime!75!black,]

\begin{defn}
\label{def:Annotation-of-an} Let $\mathcal{N}=P,T,F$ be an Net skeleton
(can be an occurrence net or a Petri Net). An annotation of $\mathcal{N}$
endows it with two operators $Q_{0}$ and $H$ such that:
\begin{itemize}
\item $H$ is a map $H:T\rightarrow\mathbf{Hilb}$
\item $Q_{0}$ is a map on conditions and event with the following signatures.
\begin{description}
\item [{Conditions}] $Q_{0}:\ci a\mapsto f\in\mathbf{Hilb}$
\item [{Events}] $Q_{0}:\boxed{t}\rightarrow f\in\mathbf{CPTNI}\left[Q_{0}\left(\pre{\boxed{t}}\right)\otimes H\left(\left\{ \boxed{t}\right\} ^{\ominus}\right),Q_{0}\left(\post{\boxed{t}}\right)\otimes H\left(\left\{ \boxed{t}\right\} ^{\oplus}\right)\right]$
\end{description}
\end{itemize}
\end{defn}

\end{tcolorbox}

\subsubsection{Global annotation $Q$ induced from the local annotation $Q_{0}$}

In order to define the corresponding global annotation obtained from
$Q_{0}$, we need to look at sub-nets that are obtained by only keeping
the events and the conditions that were used when going from one marking
$\m{}$ to another $\m{}'$: they parallel the role of intervals of
configurations in Event Structures. We call them Restrictions of Nets:
\begin{tcolorbox}[breakable,enhanced,title=Restriction of an Occurrence Net $E$ to
an interval of markings,colback=lime!5!white, colframe=lime!75!black,]

\begin{defn}
\label{def:Restriction of an Occurrence Net E to an interval of markings}
It is the sub-net of $E$ : $\rest{\int{\m{}}{\m{}'}}:=\left(P,T,F\right)$
such that $P:=\bigcup_{\m i\in\int{\m{}}{\m{}'}}\m i$ and $T:=\sigma\int{\m{}}{\m{}'}$,
with the flow relation from $F$ being restricted to $P$ and $T$.
It is required that $\m{}\rightarrow^{*}\m{}'$ for the sub-net to
be well-defined.
\end{defn}

\end{tcolorbox}

To each Restriction of Occurrence net $E$ to an interval of markings
$\int{\cutpn x}{\cutpn y}$, we make correspond an operator $\operator{\rest{\int{\cutpn x}{\cutpn y}}}$
defined by the String Diagram obtained from the local annotation of
that net. More precisely, it is the successive tensoring of all $Q_{0}\left(\boxed{e}\right)$
for all $\boxed{e}$ on the same layer of the Layer Graph, followed
by a composition with the next layer (adding the necessary $\Id{}$
maps for the input and output $H$ spaces). It has the signature in
Equation \ref{eq:Signature Operator} and defines the following global
annotation:

\begin{tcolorbox}[breakable,enhanced,title=Global Annotation induced from a Local Annotation,colback=lime!5!white,
colframe=lime!75!black,]

\begin{defn}
\label{def:Global Annotation defined from a Local Annotation}If $E,Q_{0},H$
is a locally annotated Occurrence net, the associated globally annotated
net is $E,Q_{E},H$, with canonically, for all markings $\cutpn x$
and $\cutpn y$ such that $\cutpn x\rightarrow^{*}\cutpn y$: 
\begin{equation}
Q_{E}\left(\int{\cutpn x}{\cutpn y}\right):=\operator{\rest{\int{\cutpn x}{\cutpn y}}}\label{eq:Definition Global annotation from Local Annotation}
\end{equation}
\end{defn}

\end{tcolorbox}
. 
\begin{align}
\operator{\rest{\int{\cutpn x}{\cutpn y}}} & :\mathbf{CPTNI}\left(\bigotimes_{\ci a\in Q\left(\cutpn x\right)}Q_{0}\left(\ci a\right)\otimes\bigotimes_{\left\{ \boxed{t}\right\} ^{\ominus}\in\rest{\int{\cutpn x}{\cutpn y}}}H\left(\boxed{t}^{\ominus}\right)\right.\nonumber \\
 & \longrightarrow\left.\bigotimes_{\ci a\in Q\left(\cutpn y\right)}Q_{0}\left(\ci a\right)\otimes\bigotimes_{\left\{ \boxed{t}\right\} ^{\oplus}\in\rest{\int{\cutpn x}{\cutpn y}}}H\left(\boxed{t}^{\oplus}\right)\right)\label{eq:Signature Operator}
\end{align}

We give thereafter a more formal definition of the $\operator .$
operator, but a mental picture of the process is enough to understand
the rest of the work. Figure \ref{fig:Diagrammatic construction of}
illustrates the construction of $\operator .$. 
\begin{tcolorbox}[breakable,enhanced,title=Operator associated to the restriction of
an Occurrence Net,colback=lime!5!white, colframe=lime!75!black,]

\begin{defn}
\label{def:Operator associated to the restriction of an Occurrence Net}
Let $\rest{\int{\cutpn x}{\cutpn y}}=\left(P,T,F\right)$ be an Occurrence
Net restriction (see Definition \ref{def:Restriction of an Occurrence Net E to an interval of markings}),
and $Q_{0},H$ a corresponding local annotation. Let also $G_{E}$
be the graph obtained by the procedure of Algorithm \ref{alg:From-annotated-Occurrence},
and $L=\left\{ L_{0},\ldots,L_{l|}\right\} $ be its layer-graph (i.e.
where $L_{d}$ is the set of nodes of $G_{E}$ at distance $d$ from
the initial nodes of $\cutpn x$) .

Then $\operator{\rest{\int{\cutpn x}{\cutpn y}}}$ is the $CPTNI$
corresponding to the String Diagram obtained from $G_{E}$, by tensoring
the labels of the vertices on the same layers, and then composing
the successive layers {[}see Figure \ref{fig:Diagrammatic construction of}{]}.
Formally, 
\begin{align*}
\operator{\rest{\int{\cutpn x}{\cutpn y}}} & =\mathop{\bigcirc}_{L_{i}\in L}\left[\bigotimes_{\begin{array}{c}
v\in L_{i}\\
f\text{ label of }v
\end{array}}f\right]\\
 & =\mathop{\bigcirc}_{L_{i}\in L}\left[\begin{cases}
\text{if }i\text{ even} & \bigotimes_{\ci a\in L_{i}}\Id{Q_{0}\left(\ci a\right)}\\
\text{if }i\text{ odd} & \bigotimes_{\boxed{e}\in L_{i}}Q_{0}\left(\boxed{e}\right)
\end{cases}\otimes\right.\\
 & \hfill\bigotimes_{\begin{array}{c}
\boxed{e}^{\ominus}\in L_{j}\\
j>i
\end{array}}\Id{H\left(\boxed{e}^{\ominus}\right)}\otimes\bigotimes_{\begin{array}{c}
\boxed{e}^{\oplus}\in L_{j}\\
j<i
\end{array}}\Id{H\left(\boxed{e}^{\oplus}\right)}\Biggr]
\end{align*}
\end{defn}

\end{tcolorbox}

Remark that, except for layer $L_{0}$ and $L_{l}$ that explicit
the signature of the operator, the even layers can be omitted without
loss of generality since they essentially consist in a identity that
glue two odd layers together. We hence note that if the even layer
$L_{i}$ consist in an identity $\Id A$ over a space $A$, then the
co-domain of layer $L_{i-1}$ is $A$ , as is the domain of layer
$L_{i+1}$.

\begin{algorithm}[h]
\centering
\begin{lyxcode}
\textbf{Input~:~}Annotated~restriction~of~Occ.~Net~$\rest{\int{\cutpn x}{\cutpn y}}=\left(P,T,F\right)$,$Q_{0},H$

\textbf{Output:}~$G_{\rest{\int{\cutpn x}{\cutpn y}}}$~the~string~diagram~corresponding~to~the~net

-{}-{}-

$G_{\rest{\int{\cutpn x}{\cutpn y}}}:=\left(V,E\right)=\left(P\sqcup T,F\right)$~where~each~vertex
\begin{lyxcode}
$\boxed{e}\in T$~is~labeled~$Q_{0}\left(\boxed{e}\right)$,~and

$\ci a\in P$~is~labeled~$\Id{Q_{0}\left(\ci a\right)}$
\end{lyxcode}
$L:=\left\{ L_{0},\ldots L_{l}\right\} $~the~layer~graph~of~$G_{E}$

$In:=T^{\ominus};Out:=T^{\oplus}$~the~set~of~remaining~I/O~to~be~added~

\begin{tabularx}{1\columnwidth}{>{\raggedright\arraybackslash}X||>{\raggedright\arraybackslash}X}
\begin{lyxcode}
//~Adding~the~Input~$H$~Spaces
\end{lyxcode}
 & \begin{lyxcode}
//~Adding~the~Output~$H$~Spaces:
\end{lyxcode}
\tabularnewline
\begin{lyxcode}
\textbf{For~each}~layer~$L_{i}$~from~$1$~to~$l$~:

\textbf{~For~each}~$\boxed{e}^{\ominus}\in In$:

\textbf{~~~~If}~$\boxed{e}^{\ominus}\notin L_{i}$:

~~~~~~add~vertex~$b$~to~$L_{i}$

~~~~~~~w/label~$\Id{H\left(\boxed{e}^{\ominus}\right)}$

~~~~~~add~edge~$a\rightarrow b$~~if

~~~~~~~$\exists a\in L_{i-1}$~

~~~~~~~labeled~$\Id{H\left(\boxed{e}^{\ominus}\right)}$

\textbf{~~~~Else:}

~~~~~~add~edge~$a\rightarrow Q_{0}\left(\boxed{e}^{\ominus}\right)$~if~

~~~~~~~$\exists a\in L_{i-1}$~

~~~~~~~labeled~$\Id{H\left(\boxed{e}^{\ominus}\right)}$

~~~~~~remove~$\boxed{e}^{\ominus}$~from~$In$
\end{lyxcode}
 & \begin{lyxcode}
\textbf{For~each}~layer~$L_{i}$~from~$l$~downto~1:

\textbf{~For~each}~$\boxed{e}^{\oplus}\in Out$:

\textbf{~~If}~$\boxed{e}^{\oplus}\notin L_{i}$:

~~~~add~vertex~$a$~to~$L_{i}$

~~~~~w/label~$\Id{H\left(\boxed{e}^{\oplus}\right)}$

~~~~add~edge~$a\rightarrow b$~~if~

~~~~~$\exists b\in L_{i+1}$~

~~~~~labeled~$\Id{H\left(\boxed{e}^{\oplus}\right)}$

\textbf{~~Else}:

~~~~add~edge~$Q_{0}\left(\boxed{e}^{\oplus}\right)\rightarrow b$~~if~

~~~~~$\exists b\in L_{i+1}$~

~~~~~labeled~$\Id{H\left(\boxed{e}^{\oplus}\right)}$

~~~~remove~$\boxed{e}^{\oplus}$~from~$In$
\end{lyxcode}
\tabularnewline
\end{tabularx}
\end{lyxcode}
\caption{From annotated Occurrence Net to String Diagram\label{alg:From-annotated-Occurrence}}
\end{algorithm}
\begin{figure}[H]
\centering

\tikzset{every picture/.style={line width=0.75pt}} 

\begin{tikzpicture}[x=0.75pt,y=0.75pt,yscale=-1,xscale=1]

\draw  [draw opacity=0][fill={rgb, 255:red, 126; green, 211; blue, 33 }  ,fill opacity=1 ][blur shadow={shadow xshift=3pt,shadow yshift=0pt, shadow blur radius=3.75pt, shadow blur steps=5 ,shadow opacity=100}] (81,132) .. controls (81,120.95) and (89.95,112) .. (101,112) -- (428,112) .. controls (439.05,112) and (448,120.95) .. (448,132) -- (448,192) .. controls (448,203.05) and (439.05,212) .. (428,212) -- (101,212) .. controls (89.95,212) and (81,203.05) .. (81,192) -- cycle ;
\draw  [draw opacity=0][fill={rgb, 255:red, 126; green, 211; blue, 33 }  ,fill opacity=1 ][blur shadow={shadow xshift=3pt,shadow yshift=0pt, shadow blur radius=3.75pt, shadow blur steps=5 ,shadow opacity=100}] (0,275.2) .. controls (0,264.6) and (8.6,256) .. (19.2,256) -- (508.8,256) .. controls (519.4,256) and (528,264.6) .. (528,275.2) -- (528,332.8) .. controls (528,343.4) and (519.4,352) .. (508.8,352) -- (19.2,352) .. controls (8.6,352) and (0,343.4) .. (0,332.8) -- cycle ;
\draw  [draw opacity=0][fill={rgb, 255:red, 126; green, 211; blue, 33 }  ,fill opacity=1 ][blur shadow={shadow xshift=3pt,shadow yshift=0pt, shadow blur radius=3.75pt, shadow blur steps=5 ,shadow opacity=100}] (79,460.2) .. controls (79,449.6) and (87.6,441) .. (98.2,441) -- (428.8,441) .. controls (439.4,441) and (448,449.6) .. (448,460.2) -- (448,517.8) .. controls (448,528.4) and (439.4,537) .. (428.8,537) -- (98.2,537) .. controls (87.6,537) and (79,528.4) .. (79,517.8) -- cycle ;
\draw [line width=3]    (87,72) .. controls (116.77,97.47) and (237,102) .. (277,72) ;
\draw [line width=3]    (66,576) .. controls (119.7,544) and (218.7,549) .. (256,576) ;

\draw  [draw opacity=0][fill={rgb, 255:red, 255; green, 255; blue, 255 }  ,fill opacity=1 ][blur shadow={shadow xshift=0pt,shadow yshift=-3pt, shadow blur radius=3pt, shadow blur steps=4 ,shadow opacity=100}]  (91,153) .. controls (91,143.06) and (99.06,135) .. (109,135) -- (241,135) .. controls (250.94,135) and (259,143.06) .. (259,153) -- (259,166) .. controls (259,175.94) and (250.94,184) .. (241,184) -- (109,184) .. controls (99.06,184) and (91,175.94) .. (91,166) -- cycle  ;
\draw (94,139) node [anchor=north west][inner sep=0.75pt]   [align=left] {$\displaystyle \bigotimes $ of the $\displaystyle Q_{0}$'s of the\\events at distance 1 of $\displaystyle \hat{\boldsymbol{x}}$};
\draw (265,143.4) node [anchor=north west][inner sep=0.75pt]    {$\bigotimes $};
\draw  [draw opacity=0][fill={rgb, 255:red, 255; green, 255; blue, 255 }  ,fill opacity=1 ][blur shadow={shadow xshift=0pt,shadow yshift=-3pt, shadow blur radius=3pt, shadow blur steps=4 ,shadow opacity=100}]  (10,291) .. controls (10,281.06) and (18.06,273) .. (28,273) -- (160,273) .. controls (169.94,273) and (178,281.06) .. (178,291) -- (178,304) .. controls (178,313.94) and (169.94,322) .. (160,322) -- (28,322) .. controls (18.06,322) and (10,313.94) .. (10,304) -- cycle  ;
\draw (13,277) node [anchor=north west][inner sep=0.75pt]   [align=left] {$\displaystyle \bigotimes $ of the $\displaystyle Q_{0}$'s of the\\events at distance 2 of $\displaystyle \hat{\boldsymbol{x}}$};
\draw (177,292.4) node [anchor=north west][inner sep=0.75pt]    {$\bigotimes $};
\draw  [draw opacity=0][fill={rgb, 255:red, 255; green, 255; blue, 255 }  ,fill opacity=1 ][blur shadow={shadow xshift=0pt,shadow yshift=-3pt, shadow blur radius=3pt, shadow blur steps=4 ,shadow opacity=100}]  (376,280) .. controls (376,270.06) and (384.06,262) .. (394,262) -- (507,262) .. controls (516.94,262) and (525,270.06) .. (525,280) -- (525,323) .. controls (525,332.94) and (516.94,341) .. (507,341) -- (394,341) .. controls (384.06,341) and (376,332.94) .. (376,323) -- cycle  ;
\draw (379,266) node [anchor=north west][inner sep=0.75pt]   [align=left] {$\displaystyle \mathsf{Id}\left(\begin{matrix}
\text{output spaces of}\\
\text{positive events}\\
\text{above}
\end{matrix}\right)$};
\draw (350,290.4) node [anchor=north west][inner sep=0.75pt]    {$\bigotimes $};
\draw  [draw opacity=0][fill={rgb, 255:red, 255; green, 255; blue, 255 }  ,fill opacity=1 ][blur shadow={shadow xshift=0pt,shadow yshift=-3pt, shadow blur radius=3pt, shadow blur steps=4 ,shadow opacity=100}]  (87,482) .. controls (87,472.06) and (95.06,464) .. (105,464) -- (234,464) .. controls (243.94,464) and (252,472.06) .. (252,482) -- (252,496) .. controls (252,505.94) and (243.94,514) .. (234,514) -- (105,514) .. controls (95.06,514) and (87,505.94) .. (87,496) -- cycle  ;
\draw (90,468) node [anchor=north west][inner sep=0.75pt]   [align=left] {$\displaystyle \bigotimes $ of the $\displaystyle Q_{0}$'s of the\\events at distance $\displaystyle l\ $of $\displaystyle \hat{\boldsymbol{x}}$};
\draw (253,478.4) node [anchor=north west][inner sep=0.75pt]    {$\bigotimes $};
\draw (255,224.4) node [anchor=north west][inner sep=0.75pt]    {$\bigcirc $};
\draw  [draw opacity=0][fill={rgb, 255:red, 248; green, 231; blue, 28 }  ,fill opacity=0.27 ]  (296,215) -- (461,215) -- (461,249) -- (296,249) -- cycle  ;
\draw (299,219) node [anchor=north west][inner sep=0.75pt]  [font=\Large] [align=left] {(composed with)};
\draw (254,358.4) node [anchor=north west][inner sep=0.75pt]    {$\bigcirc $};
\draw  [draw opacity=0][fill={rgb, 255:red, 80; green, 227; blue, 194 }  ,fill opacity=1 ]  (542,126) -- (585,126) -- (585,176) -- (542,176) -- cycle  ;
\draw (545,130.4) node [anchor=north west][inner sep=0.75pt]  [font=\huge]  {$L_{1}$};
\draw  [draw opacity=0][fill={rgb, 255:red, 80; green, 227; blue, 194 }  ,fill opacity=1 ]  (542,276) -- (585,276) -- (585,326) -- (542,326) -- cycle  ;
\draw (545,280.4) node [anchor=north west][inner sep=0.75pt]  [font=\huge]  {$L_{2}$};
\draw  [draw opacity=0][fill={rgb, 255:red, 80; green, 227; blue, 194 }  ,fill opacity=1 ]  (542,455) -- (579,455) -- (579,505) -- (542,505) -- cycle  ;
\draw (545,459.4) node [anchor=north west][inner sep=0.75pt]  [font=\huge]  {$L_{l}$};
\draw  [fill={rgb, 255:red, 245; green, 166; blue, 35 }  ,fill opacity=0.41 ]  (102, 53) circle [x radius= 16.28, y radius= 16.28]   ;
\draw (93,44.4) node [anchor=north west][inner sep=0.75pt]    {$p_{1}$};
\draw  [fill={rgb, 255:red, 245; green, 166; blue, 35 }  ,fill opacity=0.41 ]  (163.5, 69) circle [x radius= 15.95, y radius= 15.95]   ;
\draw (155,60.4) node [anchor=north west][inner sep=0.75pt]    {$p_{2}$};
\draw  [fill={rgb, 255:red, 245; green, 166; blue, 35 }  ,fill opacity=0.41 ]  (251, 59) circle [x radius= 16.97, y radius= 16.97]   ;
\draw (241,50.4) node [anchor=north west][inner sep=0.75pt]    {$p_{...}$};
\draw (182,60.4) node [anchor=north west][inner sep=0.75pt]  [font=\Large]  {$\dotsc $};
\draw  [draw opacity=0][fill={rgb, 255:red, 80; green, 227; blue, 194 }  ,fill opacity=1 ]  (280,45) -- (338,45) -- (338,82) -- (280,82) -- cycle  ;
\draw (283,49.4) node [anchor=north west][inner sep=0.75pt]  [font=\Large]  {$\text{cut }\hat{\boldsymbol{x}}$};
\draw (254,415.4) node [anchor=north west][inner sep=0.75pt]    {$\bigcirc $};
\draw (252,376.4) node [anchor=north west][inner sep=0.75pt]  [font=\LARGE]  {$\vdots $};
\draw  [fill={rgb, 255:red, 245; green, 166; blue, 35 }  ,fill opacity=0.41 ]  (84, 597) circle [x radius= 16.97, y radius= 16.97]   ;
\draw (74,588.4) node [anchor=north west][inner sep=0.75pt]    {$p_{...}$};
\draw  [fill={rgb, 255:red, 245; green, 166; blue, 35 }  ,fill opacity=0.41 ]  (217, 587) circle [x radius= 15.62, y radius= 15.62]   ;
\draw (209,578.4) node [anchor=north west][inner sep=0.75pt]    {$p_{r}$};
\draw (129,562.4) node [anchor=north west][inner sep=0.75pt]  [font=\Large]  {$\dotsc $};
\draw  [draw opacity=0][fill={rgb, 255:red, 80; green, 227; blue, 194 }  ,fill opacity=1 ]  (281,557) -- (338,557) -- (338,594) -- (281,594) -- cycle  ;
\draw (284,561.4) node [anchor=north west][inner sep=0.75pt]  [font=\Large]  {$\text{cut }\hat{y}$};
\draw (548,374.4) node [anchor=north west][inner sep=0.75pt]  [font=\LARGE]  {$\vdots $};
\draw  [draw opacity=0][fill={rgb, 255:red, 255; green, 255; blue, 255 }  ,fill opacity=1 ][blur shadow={shadow xshift=0pt,shadow yshift=-3pt, shadow blur radius=3pt, shadow blur steps=4 ,shadow opacity=100}]  (206,277) .. controls (206,267.06) and (214.06,259) .. (224,259) -- (330,259) .. controls (339.94,259) and (348,267.06) .. (348,277) -- (348,320) .. controls (348,329.94) and (339.94,338) .. (330,338) -- (224,338) .. controls (214.06,338) and (206,329.94) .. (206,320) -- cycle  ;
\draw (209,263) node [anchor=north west][inner sep=0.75pt]   [align=left] {$\displaystyle \mathsf{Id}\left(\begin{matrix}
\text{input spaces of}\\
\text{negative events}\\
\text{below}
\end{matrix}\right)$};
\draw  [draw opacity=0][fill={rgb, 255:red, 255; green, 255; blue, 255 }  ,fill opacity=1 ][blur shadow={shadow xshift=0pt,shadow yshift=-3pt, shadow blur radius=3pt, shadow blur steps=4 ,shadow opacity=100}]  (302,140) .. controls (302,130.06) and (310.06,122) .. (320,122) -- (426,122) .. controls (435.94,122) and (444,130.06) .. (444,140) -- (444,183) .. controls (444,192.94) and (435.94,201) .. (426,201) -- (320,201) .. controls (310.06,201) and (302,192.94) .. (302,183) -- cycle  ;
\draw (305,126) node [anchor=north west][inner sep=0.75pt]   [align=left] {$\displaystyle \mathsf{Id}\left(\begin{matrix}
\text{input spaces of}\\
\text{negative events}\\
\text{below}
\end{matrix}\right)$};
\draw  [draw opacity=0][fill={rgb, 255:red, 255; green, 255; blue, 255 }  ,fill opacity=1 ][blur shadow={shadow xshift=0pt,shadow yshift=-3pt, shadow blur radius=3pt, shadow blur steps=4 ,shadow opacity=100}]  (286,464) .. controls (286,454.06) and (294.06,446) .. (304,446) -- (417,446) .. controls (426.94,446) and (435,454.06) .. (435,464) -- (435,507) .. controls (435,516.94) and (426.94,525) .. (417,525) -- (304,525) .. controls (294.06,525) and (286,516.94) .. (286,507) -- cycle  ;
\draw (289,450) node [anchor=north west][inner sep=0.75pt]   [align=left] {$\displaystyle \mathsf{Id}\left(\begin{matrix}
\text{output spaces of}\\
\text{positive events}\\
\text{above}
\end{matrix}\right)$};
\draw (35,588.4) node [anchor=north west][inner sep=0.75pt]    {$Q_{0}$};
\draw (58,577.4) node [anchor=north west][inner sep=0.75pt]  [font=\LARGE]  {$($};
\draw (99,578.4) node [anchor=north west][inner sep=0.75pt]  [font=\LARGE]  {$)$};
\draw (168,580.4) node [anchor=north west][inner sep=0.75pt]    {$Q_{0}$};
\draw (191,569.4) node [anchor=north west][inner sep=0.75pt]  [font=\LARGE]  {$($};
\draw (230,569.4) node [anchor=north west][inner sep=0.75pt]  [font=\LARGE]  {$)$};
\draw (115,34.4) node [anchor=north west][inner sep=0.75pt]  [font=\LARGE]  {$)$};
\draw (54,45.4) node [anchor=north west][inner sep=0.75pt]    {$Q_{0}$};
\draw (77,34.4) node [anchor=north west][inner sep=0.75pt]  [font=\LARGE]  {$($};
\draw (118,61.4) node [anchor=north west][inner sep=0.75pt]    {$Q_{0}$};
\draw (139,50.4) node [anchor=north west][inner sep=0.75pt]  [font=\LARGE]  {$($};
\draw (177,50.4) node [anchor=north west][inner sep=0.75pt]  [font=\LARGE]  {$)$};
\draw (206,50.4) node [anchor=north west][inner sep=0.75pt]    {$Q_{0}$};
\draw (227,39.4) node [anchor=north west][inner sep=0.75pt]  [font=\LARGE]  {$($};
\draw (265,39.4) node [anchor=north west][inner sep=0.75pt]  [font=\LARGE]  {$)$};
\draw    (110.56,66.12) -- (156.39,132.53) ;
\draw [shift={(158.09,135)}, rotate = 235.39] [fill={rgb, 255:red, 0; green, 0; blue, 0 }  ][line width=0.08]  [draw opacity=0] (8.93,-4.29) -- (0,0) -- (8.93,4.29) -- cycle    ;
\draw    (165.51,84.82) -- (171.51,132.02) ;
\draw [shift={(171.89,135)}, rotate = 262.76] [fill={rgb, 255:red, 0; green, 0; blue, 0 }  ][line width=0.08]  [draw opacity=0] (8.93,-4.29) -- (0,0) -- (8.93,4.29) -- cycle    ;
\draw    (240.76,72.54) -- (195.34,132.61) ;
\draw [shift={(193.53,135)}, rotate = 307.1] [fill={rgb, 255:red, 0; green, 0; blue, 0 }  ][line width=0.08]  [draw opacity=0] (8.93,-4.29) -- (0,0) -- (8.93,4.29) -- cycle    ;
\draw    (96.36,581.31) -- (150.19,513) ;
\draw [shift={(94.5,583.67)}, rotate = 308.24] [fill={rgb, 255:red, 0; green, 0; blue, 0 }  ][line width=0.08]  [draw opacity=0] (8.93,-4.29) -- (0,0) -- (8.93,4.29) -- cycle    ;
\draw    (208.91,570.22) -- (181.31,513) ;
\draw [shift={(210.21,572.93)}, rotate = 244.26] [fill={rgb, 255:red, 0; green, 0; blue, 0 }  ][line width=0.08]  [draw opacity=0] (8.93,-4.29) -- (0,0) -- (8.93,4.29) -- cycle    ;

\end{tikzpicture}
\caption{\label{fig:Diagrammatic construction of}Diagrammatic construction
of $\protect\operator{\protect\rest{\protect\int{\protect\cutpn x}{\protect\cutpn y}}}$
from $\protect\rest{\protect\int{\protect\cutpn x}{\protect\cutpn y}}$
{[}see Definition \ref{def:Operator associated to the restriction of an Occurrence Net}{]}}
\end{figure}

\subsection{Local Functoriality\label{subsec:Local Functoriality}}

We prove here that a Local annotation $Q_{0}$ always induces a Global
annotation that respects the Functoriality property (See Definition
\ref{def:Global Quantum Occurrence Net})
\begin{tcolorbox}[breakable,enhanced,title=Functoriality is naturally induced by a local
annotation,colback=teal!15!white, colframe=teal!60!black,]

\begin{thm}
\label{thm:Functoriality}If a global quantum Occurrence net annotation
$Q$ is defined from a local annotation $Q_{0}$, then $Q$ satisfies
the functoriality property.
\end{thm}

\begin{enumerate}
\item By def., $Q\left(\int{\cutpn x}{\cutpn x}\right)=\operator{\rest{\int{\cutpn x}{\cutpn x}}}=\Id{Q_{0}\left(\cutpn x\right)}=\Id{Q\left(x\right)}$.
\item If $\cutpn x\rightarrow^{*}\cutpn y\rightarrow^{*}\cutpn z$, we have
that $\rest{\int{\cutpn x}{\cutpn z}}$ is the Occurrence Net obtained
by ``gluing'' the respective Occurrence Nets of $\rest{\int{\cutpn y}{\cutpn z}}$
and $\rest{\int{\cutpn x}{\cutpn y}}$ by the conditions of $\cutpn y$.
Furthermore, remark that the co-domain of $\operator{\rest{\int{\cutpn y}{\cutpn z}}}$
corresponds to the domain of $\operator{\rest{\int{\cutpn x}{\cutpn y}}}$.
As such, we can write
\begin{alignat*}{1}
Q\left(\int{\cutpn x}{\cutpn z}\right) & =\operator{\rest{\int{\cutpn x}{\cutpn z}}}\\
 & =\left[\Id{H\left[y\setminus x\right]^{\oplus}}\otimes\operator{\rest{\int{\cutpn y}{\cutpn z}}}\right]\\
 & \circ\left[\operator{\rest{\int{\cutpn x}{\cutpn y}}}\otimes\Id{H\left[z\setminus y\right]^{\ominus}}\right]
\end{alignat*}
\end{enumerate}
\end{tcolorbox}

\subsection{Local Obliviousness\label{subsec:Local Obliviousness}}

We show here that if the local annotation $Q_{0}$ verifies a condition
called Local Obliviousness, then the global annotation $Q$ defined
from $Q_{0}$ respects the (global) obliviousness property ( See Definition
\ref{def:Global Quantum Occurrence Net})

\begin{tcolorbox}[breakable,enhanced,title=Local Obliviousness,colback=lime!15!white,
colframe=lime!50!black,]

\begin{defn}
\label{def:Local-Obliviousness}

A map of events $Q_{0}$ respects local obliviousness if for all negative
events $\boxed{e}^{\ominus}$
\end{defn}

\begin{itemize}
\item $Q_{0}\left(\boxed{e}^{\ominus}\right)=\Id{Q\left(\pre{\boxed{e}}\right)\otimes H\left(\boxed{e}\right)}$
\end{itemize}
\end{tcolorbox}

\begin{tcolorbox}[breakable,enhanced,title=Signature of a negative event,colback=lime!15!white,
colframe=lime!50!black,]

\begin{prop}
\label{prop:Signature-of-a-negative-event}

If $Q_{0}$ respects Local obliviousness, for all $\boxed{e}^{\ominus}$,
\[
Q_{0}\left(\post{\boxed{e}}\right)=Q_{0}\left(\pre{\boxed{e}}\right)\otimes H\left(\boxed{e}\right)
\]
\end{prop}

\begin{proof}
$Q_{0}\left(\boxed{e}^{\ominus}\right)$ is an identity with $Q_{0}\left(\pre{\boxed{e}}\right)\otimes H\left(\boxed{e}\right)$
as its domain and $Q_{0}\left(\post{\boxed{e}}\right)$ as its co-domain.
\end{proof}
\end{tcolorbox}

\begin{tcolorbox}[breakable,enhanced,title=Local Obliviousness implies Global Obliviousness,colback=lime!15!white,
colframe=lime!50!black,]

\begin{thm}
Local Obliviousness implies Global Obliviousness\label{lemma:Local-Obliviousness-implies-global}

Let $Q$ be a global Occurrence net annotation defined by the local
annotation $Q_{0}$ (Def. \ref{def:Global Annotation defined from a Local Annotation}),
that satisfies Local Obliviousness. Then if $\cutpn x\rightarrow^{*,\ominus}\cutpn y$,
\end{thm}

\begin{enumerate}
\item 
\[
Q\left(\int{\cutpn x}{\cutpn y}\right)=\Id{Q\left(\cutpn x\right)\otimes H\left[y\setminus x\right]^{\ominus}}
\]
\item 
\[
Q\left(\cutpn y\right)=Q\left(\cutpn x\right)\otimes\bigotimes_{e\in y\setminus x}H\left(\boxed{e}\right)
\]
\end{enumerate}

\tcbline{}

Since every event in the annotated net $\rest{\left[\cutpn x,\cutpn y\right]}$
is negative, $\operator{\rest{\left[\cutpn x,\cutpn y\right]}}$ is
of the form 
\begin{align*}
\operator{\rest{\left[\cutpn x,\cutpn y\right]}} & =\mathop{\bigcirc}_{L_{i}\in L}\left[\bigotimes_{\begin{array}{c}
v\in L_{i}\\
f\text{ label of }v
\end{array}}f\right]\\
 & =\mathop{\bigcirc}_{L_{i}\in L}\left[\bigotimes_{\boxed{e}\in L_{i}}Q_{0}\left(\boxed{e}\right)\otimes\bigotimes_{\begin{array}{c}
\boxed{e}^{\ominus}\in L_{j}\\
j>i
\end{array}}\Id{H\left(\boxed{e}^{\ominus}\right)}\right]
\end{align*}

where we have omitted tensors of identities on the spaces of conditions
w.l.o.g.. Then prove by recursion on $n$ that $\mathop{\bigcirc}_{\begin{array}{c}
L_{i}\in L\\
i\leq n
\end{array}}\left[\bigotimes_{\boxed{e}\in L_{i}}Q_{0}\left(\boxed{e}\right)\otimes\bigotimes_{\begin{array}{c}
\boxed{e}^{\ominus}\in L_{j}\\
j>i
\end{array}}\Id{H\left(\boxed{e}^{\ominus}\right)}\right]=\Id{Q\left(x\right)\otimes H\left[y\setminus x\right]^{\ominus}}$ .
\end{tcolorbox}

\subsection{Local Drop Condition\label{subsec:Local Drop Condition}}

We show that if the local annotation $Q_{0}$ satisfies a condition
we call the Local Drop Condition, then the corresponding global annotation
$Q$, as constructed in Definition \ref{def:Global Annotation defined from a Local Annotation},
satisfies the (global) Drop Condition (see Definition \ref{def:Global Quantum Occurrence Net}).

Since the Drop Condition is defined over positive intervals of configurations
(or markings), this global condition induces a derived version for
the local occurrence net. Specifically, it corresponds to replacing
$Q\left(\int{\cutpn x}{\cutpn y}\right)$ for $\operator{\rest{\int{\cutpn x}{\cutpn y}}}$
in the global Drop Condition. However, this global condition is computationally
impractical: verifying it requires checking an exponential number
of intervals.

To address this, we introduce a sufficient local condition--the Local
Drop Condition--which is significantly more tractable. It only needs
to be checked:
\begin{itemize}
\item on simple extensions of configurations (i.e., extensions by immediately
causally dependent events),
\item and only within conflict clusters, defined as sets of events connected
through the conflict relation $\sim$.
\end{itemize}
This approach drastically reduces the number of intervals that must
be verified. In fact, when conflict clusters form cliques, the number
of checks can become linear. We demonstrate these simplifications
in Section \ref{subsec:Only checking single extensions are enough}
and Section \ref{subsec:Only checking conflict clusters is enough}.

\subsubsection{Only checking single extensions is enough\label{subsec:Only checking single extensions are enough}}

In this section, we show that only checking the Drop Conditions on
positive intervals of single extensions is enough to enforce it for
every marking/configuration interval. 

\begin{tcolorbox}[breakable,enhanced, colback=green!8!black!5!white, colframe=green!30!black,]

A single extension is is a configuration/marking interval only composed
of events that are directly enabled from the base configuration--in
other words, they are all comprised of singled extensions. Such an
interval is of the form $\int{\cutpn x}{\cutpn y}$, where $\cutpn y=\cutpn x\sqcup\boxed{e_{1}}\sqcup\ldots\sqcup\boxed{e_{n}}$,
and where $\forall i,\,x\ext\boxed{e_{i}}$.
\begin{example}
In Figure \ref{fig:occurrence net with cut and marking}, the interval
$\int{\cutpn x}{\cutpn y}$ is a single extension, but $\int{\cutpn x}{\cutpn z}$
is not.
\end{example}

\end{tcolorbox}

We first show technical Lemmas (\ref{lem:Alternative inductive definition of the Drop Function},
\ref{lem:Drop condition on a collapsed interval},\ref{lem:Drop condition as a recursive sum}
and \ref{lem:Expansion of the Drop Condition}), and leverage them
to obtain the stated result, in Theorem \ref{conj:Considering-single-extensions}. 

\paragraph{Technical Lemmas}

\begin{tcolorbox}[breakable,enhanced,title=Alternative inductive definition of the Drop
Function ,colback=white,colframe=black!30!white,borderline={0.5mm}{0mm}{green!30!black,dashed}]

\begin{lem}
\label{lem:Alternative inductive definition of the Drop Function}Let
$x\subseteq^{0,\oplus}y_{1},\ldots,y_{n}\in C\left(E\right)$. We
have the following inductive relation.

\[
\qdv x{}=\tr_{Q_{0}\left(\cutpn x\right)}
\]
\begin{align*}
\qdv x{y_{1},\ldots,y_{n}} & =\qdv x{y_{1},\ldots,y_{n-1}}\\
 & -\left[\tr_{H\left[y_{n}\setminus x\right]^{\oplus}}\otimes\qdv{y_{n}}{y_{1}\cup y_{n},\ldots,y_{n-1}\cup y_{n}}\right]\circ\operator{\rest{\int{\cutpn x}{\cutpn{y_{n}}}}}
\end{align*}
\end{lem}

\tcbline{}

\begin{align*}
\qdv x{y_{1},\ldots,y_{n}} & =\sum_{\begin{array}{c}
I\subseteq\left\llbracket 1,n\right\rrbracket \\
\bigcup_{i\in I}y_{i}\in C(E)
\end{array}}\left(-1\right)^{|I|}\tr\circ\operator{\rest{\int{\cutpn x}{\cutpn{\bigcup_{i\in I}y_{i}}}}}\\
 & =\sum_{\begin{array}{c}
J\subseteq\left\llbracket 1,n-1\right\rrbracket \\
\bigcup_{i\in J}y_{j}\in C(E)
\end{array}}\left(-1\right)^{|I|}\tr\circ\operator{\rest{\int{\cutpn x}{\cutpn{\bigcup_{j\in I}y_{j}}}}}-\\
(1) & \sum_{\begin{array}{c}
J\subseteq\left\llbracket 1,n-1\right\rrbracket \\
\bigcup_{i\in J}y_{j}\cup y_{n}\in C(E)
\end{array}}\left(-1\right)^{|I|}\tr\circ\uwave{\operator{\rest{\int{\cutpn x}{\cutpn{\bigcup_{j\in j}y_{j}\cup y_{n}}}}}}
\end{align*}

By the functoriality property, we have (writing simplifications with
an arrow):

\begin{alignat*}{1}
\uwave{\operator{\rest{\int{\cutpn x}{\cutpn{\bigcup_{j\in j}y_{j}\cup y_{n}}}}}} & =\left[\Id{H\left[y_{n}\setminus x\right]^{\oplus}}\otimes\operator{\rest{\int{\cutpn{y_{n}}}{\cutpn{\bigcup_{j\in J}y_{j}\cup y_{n}}}}}\right]\\
 & \circ\left[\operator{\rest{\int{\cutpn x}{\cutpn{y_{n}}},}}\otimes\Id{\cancelto{\mathbb{C}}{H\left[\bigcup_{j\in J}y_{j}\cup y_{n}\setminus y_{n}\right]^{\ominus}}}\right]
\end{alignat*}

Then
\begin{align*}
(1) & =\sum_{\begin{array}{c}
J\subseteq\left\llbracket 1,n-1\right\rrbracket \\
\bigcup_{i\in J}y_{j}\cup y_{n}\in C(E)
\end{array}}\left(-1\right)^{|I|}\tr\left[\Id{H\left[y_{n}\setminus x\right]^{\oplus}}\otimes\operator{\rest{\int{\cutpn{y_{n}}}{\cutpn{\bigcup_{j\in J}y_{j}\cup y_{n}}}}}\right]\circ\operator{\rest{\int{\cutpn x}{\cutpn{y_{n}}},}}\\
 & =\left[\tr_{H\left[y_{n}\setminus x\right]^{\oplus}}\otimes\sum_{\begin{array}{c}
J\subseteq\left\llbracket 1,n-1\right\rrbracket \\
\bigcup_{i\in J}y_{j}\cup y_{n}\in C(E)
\end{array}}\left(-1\right)^{|I|}\tr\operator{\rest{\int{\cutpn{y_{n}}}{\cutpn{\bigcup_{j\in J}y_{j}\cup y_{n}}}}}\right]\circ\operator{\rest{\int{\cutpn x}{\cutpn{y_{n}}},}}\\
 & =\left[\tr_{H\left[y_{n}\setminus x\right]^{\oplus}}\otimes\qdv{y_{n}}{y_{1}\cup y_{n},\ldots,y_{n-1}\cup y_{n}}\right]\circ\operator{\rest{\int{\cutpn x}{\cutpn{y_{n}}}}}
\end{align*}

Hence , we recognize the result

\begin{align*}
\qdv x{y_{1},\ldots,y_{n}} & =\qdv x{y_{1},\ldots,y_{n-1}}\\
 & -\left[\tr_{H\left[y_{n}\setminus x\right]^{\oplus}}\otimes\qdv{y_{n}}{y_{1}\cup y_{n},\ldots,y_{n-1}\cup y_{n}}\right]\circ\operator{\rest{\int{\cutpn x}{\cutpn{y_{n}}}}}
\end{align*}
\end{tcolorbox}

\begin{tcolorbox}[breakable,enhanced,title= Drop condition on a collapsed interval ,colback=white,colframe=black!30!white,borderline={0.5mm}{0mm}{green!30!black,dashed}]

\begin{lem}
\label{lem:Drop condition on a collapsed interval}{[}adapted from
\autocite[Proposition 2,][]{winskelProbabilisticQuantumEvent2014}{]}
Let $x\subseteq y_{1},\ldots,y_{n}\in C(E)$. If for some $i$, $x=y_{i}$,
then $\qdv x{y_{1},\ldots,y_{n}}=0$ 
\end{lem}

\tcbline{}

Assume w.l.o.g. that $i=n$. It follows that $\operator{\rest{\int{\cutpn x}{\cutpn{y_{n}}}}}=\operator{\rest{\int{\cutpn x}{\cutpn x}}}=\Id{Q\left(\cutpn x\right)}$,
$\qd{x\cup y_{n}}{y_{1}\cup y_{n},\ldots,y_{n-1}\cup y_{n}}=\qdv x{y_{1},\ldots,y_{n-1}}$
and $H\left[y_{n}\setminus x\right]^{\oplus}=\mathbb{C}$ . Hence,
by \ref{lem:Alternative inductive definition of the Drop Function},
\begin{align*}
\qdv x{y_{1},\ldots,y_{n}} & =\qdv x{y_{1},\ldots,y_{n-1}}-\left[\tr_{H\left[y_{n}\setminus x\right]^{\oplus}}\otimes\qdv{y_{n}}{y_{1}\cup y_{n},\ldots,y_{n-1}\cup y_{n}}\right]\circ\operator{\rest{\int{\cutpn x}{\cutpn{y_{n}}}}}\\
 & =\qdv x{y_{1},\ldots,y_{n-1}}-\qdv x{y_{1},\ldots,y_{n-1}}\\
 & =0
\end{align*}
\end{tcolorbox}

\begin{tcolorbox}[breakable,enhanced,title= Drop condition as a recursive sum ,colback=white,colframe=black!30!white,borderline={0.5mm}{0mm}{green!30!black,dashed}]

\begin{lem}
\label{lem:Drop condition as a recursive sum}{[}adapted from \autocite[Lemma 1,][]{winskelProbabilisticQuantumEvent2014}{]}
Let $x\subseteq^{0,\oplus}y_{1},\ldots y_{n-1},y'_{n}\in C(E)$, and
let $y_{n}\subseteq y'_{n}$. Then

\begin{align*}
\qdv x{y_{1},\ldots,y'_{n}} & =\qdv x{y_{1},\ldots,y_{n}}\\
 & +\left[\tr_{H\left[y_{n}\setminus x\right]^{\oplus}}\otimes\qd{y_{n}}{y_{1}\cup y_{n},\ldots,y_{n-1}\cup y_{n},y_{n}'}\right]\circ\operator{\rest{\int{\cutpn x}{\cutpn{y_{n}}}}}
\end{align*}
\end{lem}

\tcbline{}

By applying \ref{lem:Alternative inductive definition of the Drop Function}
on each term, 
\begin{align*}
rhs & =\uwave{\qdv x{y_{1},\ldots,y_{n-1}}}\\
 & \mathbin{\color{purple}-}{\color{purple}\left[\tr_{H\left[y_{n}\setminus x\right]^{\oplus}}\otimes\qdv{y_{n}}{y_{1}\cup y_{n},\ldots,y_{n-1}\cup y_{n}}\right]\circ\operator{\rest{\int{\cutpn x}{\cutpn{y_{n}}}}}}\\
 & +\left[\tr_{H\left[y_{n}\setminus x\right]^{\oplus}}\otimes\left(A\right)\right]\circ\operator{\rest{\int{\cutpn x}{\cutpn{y_{n}}}}}
\end{align*}

with 
\begin{align*}
\left(A\right):= & \qd{y_{n}}{y_{1}\cup y_{n},\ldots,y_{n-1}\cup y_{n}}\\
 & -\left[\tr_{H\left[y'_{n}\setminus y_{n}\right]^{\oplus}}\otimes\qdv{y'_{n}}{y_{1}\cup y'_{n},\ldots,y_{n-1}\cup y_{n}'}\right]\\
 & \hfill\circ\operator{\rest{\int{\cutpn{y_{n}}}{\cutpn{y'_{n}}}}}
\end{align*}

By functoriality of the tensor product

\begin{align*}
\left[\tr_{H\left[y_{n}\setminus x\right]^{\oplus}}\otimes\left(A\right)\right]= & {\color{blue}\tr_{H\left[y_{n}\setminus x\right]^{\oplus}}\otimes\qd{y_{n}}{y_{1}\cup y_{n},\ldots,y_{n-1}\cup y_{n}}}\\
- & \left[\tr_{H\left[y'_{n}\setminus x\right]^{\oplus}}\otimes\qdv{y'_{n}}{y_{1}\cup y'_{n},\ldots,y_{n-1}\cup y_{n}'}\right]\\
 & \hfill\circ\left[\Id{H\left[y_{n}\setminus x\right]^{\oplus}}\otimes\operator{\rest{\int{\cutpn{y_{n}}}{\cutpn{y'_{n}}}}}\right]
\end{align*}

Then, by functoriality of $\operator .$ (\ref{thm:Functoriality})

\begin{align*}
\left[\tr_{H\left[y_{n}\setminus x\right]^{\oplus}}\otimes\left(A\right)\right]\circ\operator{\rest{\int{\cutpn x}{\cutpn{y_{n}}}}} & ={\color{blue}\tr_{H\left[y_{n}\setminus x\right]^{\oplus}}\otimes\qd{y_{n}}{y_{1}\cup y_{n},\ldots,y_{n-1}\cup y_{n}}}-\\
 & \left[\tr_{H\left[y'_{n}\setminus x\right]^{\oplus}}\otimes\qdv{y'_{n}}{y_{1}\cup y'_{n},\ldots,y_{n-1}\cup y_{n}'}\right]\\
 & \hfill\circ\operator{\rest{\int{\cutpn x}{\cutpn{y'_{n}}}}}
\end{align*}

Hence, reducing the whole expression and applying again \ref{lem:Alternative inductive definition of the Drop Function},

\begin{align*}
rhs & =\qdv x{y_{1},\ldots,y_{n-1}}\\
 & \mathbin{\color{purple}-}{\color{purple}\left[\tr_{H\left[y_{n}\setminus x\right]^{\oplus}}\otimes\qdv{y_{n}}{y_{1}\cup y_{n},\ldots,y_{n-1}\cup y_{n}}\right]\circ\operator{\rest{\int{\cutpn x}{\cutpn{y_{n}}}}}}\\
 & \mathbin{\color{blue}+}{\color{blue}\left[\tr_{H\left[y_{n}\setminus x\right]^{\oplus}}\otimes\qd{y_{n}}{y_{1}\cup y_{n},\ldots,y_{n-1}\cup y_{n}}\right]\circ\operator{\rest{\int{\cutpn x}{\cutpn{y_{n}}}}}}\\
 & -\left[\tr_{H\left[y'_{n}\setminus x\right]^{\oplus}}\otimes\qdv{y'_{n}}{y_{1}\cup y'_{n},\ldots,y_{n-1}\cup y_{n}'}\right]\circ\operator{\rest{\int{\cutpn x}{\cutpn{y'_{n}}}}}\\
 & =lhs
\end{align*}
\end{tcolorbox}
\begin{tcolorbox}[breakable,enhanced,title=Expansion of the Drop Condition ,colback=green!8!black!5!white,
colframe=green!30!black,]

\begin{lem}
\label{lem:Expansion of the Drop Condition}{[}adapted from \autocite[Lemma 2,][]{winskelProbabilisticQuantumEvent2014}{]}
Let $x\subseteq y_{1},\ldots,y_{n}$ a configuration. Then $\qdv x{y_{1},\ldots,y_{n}}$
expands as a is a finite production of the grammar 
\begin{align*}
A & \longrightarrow A+\left[\tr\otimes A\right]\circ\mathbb{O}\\
A & \longrightarrow\mathsf{SingleExtension}
\end{align*}
where $\mathsf{SingleExtension}$ terms are of the form $\qdv u{w_{1},\ldots,w_{k}}$,
where $x\subseteq u\ext w_{1}$ and $w_{i}\subseteq y_{1}\cup\ldots\cup y_{n}$
for all $i\in\left\llbracket 1,k\right\rrbracket $
\end{lem}

\tcbline{}

$\vartriangleright$ Define the weight of a drop function $\qdv x{y_{1},\ldots,y_{n}}$
be $w\left(\qdv x{y_{1},\ldots,y_{n}}\right):=\left|y_{1}\setminus x\right|\times\ldots\times\left|y_{n}\setminus x\right|$.
It is an upper bound on the number of configurations in the interval
$\int x{y_{1},\ldots,y_{n}}$.

Assume $x\subseteq y_{1},\ldots y_{n-1},y'_{n}\in C(E)$. By \ref{lem:Drop condition on a collapsed interval},
we can suppose that $\forall i,\,x\neq y_{i}$, otherwise $\qdv x{y_{1},\ldots,y_{n}'}=0$. 

$\vartriangleright$ If $x\subsetneq y_{n}\subsetneq y'_{n}$, by
\ref{lem:Drop condition as a recursive sum}, we have

\begin{align*}
\qdv x{y_{1},\ldots,y_{n}'} & =\qdv x{y_{1},\ldots,y_{n}}\\
 & +\left[\tr_{H\left[y_{n}\setminus x\right]^{\oplus}}\otimes\qd{y_{n}}{y_{1}\cup y_{n},\ldots,y_{n-1}\cup y_{n},y_{n}'}\right]\circ\operator{\rest{\int{\cutpn x}{\cutpn{y_{n}}}}}
\end{align*}
where we observe that $\left|y_{n}\setminus x\right|<\left|y'_{n}\setminus x\right|$,
$\left|y'_{n}\setminus y_{n}\right|<\left|y'_{n}\setminus x\right|$
and $\left|\left(y_{i}\cup y_{n}\right)\setminus y_{n}\right|\leq\left|y_{i}\setminus x\right|$
($\star$). Furthermore, $\qdv x{y_{1},\ldots,y_{n}}$ and $\qd{y_{n}}{y_{1}\cup y_{n},\ldots,y_{n-1}\cup y_{n},y_{n}'}$
satisfy the conditions of \ref{lem:Drop condition as a recursive sum}.
So we can see the expansion of the expression $\qdv x{y_{1},\ldots,y_{n}'}$--after
a finite number of applications of \ref{lem:Drop condition as a recursive sum}--as
a finite production of the grammar 
\begin{align}
A & \overset{C1}{\longrightarrow}A+\left[\tr\otimes A\right]\circ\mathbb{O}\\
A & \overset{C2}{\longrightarrow}\mathsf{SingleExtension}
\end{align}

where 
\begin{description}
\item [{(C1)}] either the term $A$ has a weight $>1$, then rule (4.3)
applies (here $A$ verifies the conditions of \ref{lem:Drop condition as a recursive sum},
and is non terminal) 
\item [{(C2)}] either the term $A$ has weight $=1$, and is of the form
$\qdv u{w_{1},\ldots,w_{k}}$ where $x\subseteq u\ext w_{1}$ and
$w_{i}\subseteq y_{1}\cup\ldots\cup y_{n}$ for all $i\in\left\llbracket 1,k\right\rrbracket $.
Then (4.4) applies and the term is marked as a single extension. (we
ignore the case weight $=0$ since the drop-function vanishes in this
case)
\end{description}
$\vartriangleright$ The iteration eventually terminates since the
weight of the terms is strictly decreasing (by ($\star$)).
\end{tcolorbox}

The previous Lemma allows for the following critical result, stating
that checking the Drop Condition on every intervals of single extension
is enough to enforce the drop condition globally.

\begin{tcolorbox}[breakable,enhanced,title=Considering single extensions is enough,
colback=green!8!black!5!white, colframe=green!30!black,]

\begin{thm}
\label{conj:Considering-single-extensions} Let $E$ be an Occurrence
net, and $Q$ be a Quantum global annotation. Then 
\begin{align*}
\forall x\subseteq^{0,\oplus}y_{1},\ldots,y_{n}, &  & \qdv x{y_{1},\ldots,y_{n}}\sqsupseteq0\\
 & \iff\\
\forall x\ext^{0,\oplus}x\sqcup e_{1},\ldots,x\sqcup e_{n} &  & \qdv x{x\sqcup e_{1},\ldots,x\sqcup e_{n}}\sqsupseteq0
\end{align*}
\end{thm}

\tcbline{}
\begin{description}
\item [{$\implies$:}] If the Drop Condition is satisfied for every positive
interval, it is also satisfied for intervals of single extensions.
\item [{$\impliedby$:}] Let $x\subseteq^{0,\oplus}y_{1},\ldots,y_{n}$
be a positive interval of configurations. By \lemref{{Expansion of the Drop Condition}},
$\qdv x{y_{1},\ldots,y_{n}}$ is a production of the grammar
\begin{align*}
A & \longrightarrow A+\left[\tr\otimes A\right]\circ\mathbb{O}\\
A & \longrightarrow\mathsf{SingleExtension}
\end{align*}
 where single extensions are terminal elements. By assumption, the
latter are positive. Then since (i) composition by a positive operator--of
which $\mathbb{\operator{}}$'s pertain, (ii) tensoring by a trace,
and (iii) summing with another positive operator all conserve positivity,
$\qdv x{y_{1},\ldots,y_{n}}$ is positive.
\end{description}
\end{tcolorbox}

\subsubsection{Only checking conflict clusters is enough\label{subsec:Only checking conflict clusters is enough}}

This section shows that the number of configuration that needs to
be checked can be further refined by only looking at conflict clusters.
\begin{tcolorbox}[breakable,enhanced,title=Conflict Cluster, colback=green!12!white,
colframe=green!40!black]

A conflict cluster is a connected components of the undirected conflict
graph induced by the symmetric (but not necessarily transitive) conflict
relation. Formally , defining a conflict graph as $G:=\left(T,\#\right)$
where each event is a node, and edges connect events in conflict,
a conflict cluster is a connected component of $G$.
\end{tcolorbox}

\paragraph{Simplification of the drop condition from the Single Extension Lemma }

We make use of the main property shown in the previous section, and
only care about drop conditions of single extensions of events, since
it is enough to enforce the global drop condition. In this setting,
we will denote by $\cutpn x$ a marking/configuration, and $\boxed{e_{1}},\ldots,\boxed{e_{n}}$
the events s.t. $x\overset{e_{i}}{\ext}y_{i}$.

Recall the general expression of the drop function (Definition \ref{def:Global Quantum Occurrence Net}):
\[
\qdv x{y_{1},\ldots,y_{n}}:=\tr_{Q\left(\cutpn x\right)}+\sum_{\substack{I\subseteq\{1,\dots,n\}\\
\text{s.t. }\cutpn{y_{I}}\in\mathsf{Marking}\left(O\right)
}
}(-1)^{|I|}\tr_{Q\left(\cutpn{y_{I}}\right)\otimes H\left(\sigma\int{\cutpn x}{\cutpn{y_{I}}}\right)}\circ Q\left(\int{\cutpn x}{\cutpn{y_{I}}}\right)
\]

This single extension case is restrictive enough that the expression
of $\operator{\rest{\int{\cutpn x}{\cutpn{y_{I}}}}}$ is straight
forward:
\begin{tcolorbox}[breakable,enhanced, colback=green!12!white, colframe=green!40!black]

\begin{alignat*}{1}
\operator{\rest{\int{\cutpn x}{\cutpn{y_{I}}}}} & =\bigotimes_{i\in I}Q_{0}\left(\boxed{e_{i}}\right)\otimes\Id{Q_{0}\left(R_{I}\right)}\otimes\Id{Q_{0}\left(X\right)}\\
 & =\bigotimes_{i\in I}Q_{0}\left(\boxed{e_{i}}\right)\otimes\Id{Q_{0}\left(K_{I}\right)}
\end{alignat*}

where we define, for $I\subseteq\left\llbracket 1,n\right\rrbracket $:

\begin{tabularx}{\columnwidth}{>{\centering\arraybackslash}X|>{\centering\arraybackslash}X}
\multicolumn{2}{c}{\begin{cellvarwidth}[t]
\centering
$K_{I}=R_{I}\sqcup X$ (K like ``Konstant'')

as the set of conditions that stay marked after the passage to $\cutpn{y_{I}}$,
with:
\end{cellvarwidth}}\tabularnewline
${\displaystyle R_{I}:=\bigsqcup_{j\notin I}\pre{\boxed{e_{j}}}}$
($R$ like ``remaining'') & $X:=\left\{ \ci a\in\cutpn x\left|\forall i,\,\ci a\notin\pre{\boxed{e_{i}}}\right.\right\} $\tabularnewline
as the conditions untouched by the $\boxed{e_{i}}$'s, $\forall i\in I$ & the conditions of the cut/marking $\cutpn x$ not involved in the
transition whatever the $I\subseteq\left\llbracket 1,n\right\rrbracket $\tabularnewline
\end{tabularx}
\end{tcolorbox}

\begin{tcolorbox}[breakable,enhanced,title=Expression of the Drop Function for single
extensions,colback=green!12!white, colframe=green!40!black]

Hence in the case of single extension the Drop Function rewrites as:

\begin{align}
\qdv x{y_{1},\ldots,y_{n}} & =\sum_{\begin{array}{c}
I\subseteq\left\llbracket 1,n\right\rrbracket \\
y_{I} \in C(E)
\end{array}}\left(-1\right)^{|I|}\tr\circ\underbrace{\left[\bigotimes_{i\in I}Q_{0}\left(\boxed{e_{i}}\right)\otimes\Id{K_{I}}\right]}_{=\operator{\rest{\int{\cutpn x}{\cutpn{y_{I}}}}} }
\end{align}

\end{tcolorbox}

\paragraph{Cluster-wise definition of the Drop Function}

We can further refine the expression of the Drop Condition for single
extension, in the form of a factorization, where each term is a single-extension
conflict cluster. 

\begin{tcolorbox}[breakable,enhanced,title=Cluster-Factorization of the Drop Function
, colback=green!12!white, colframe=green!40!black]

\begin{thm}
Let $x$ be configuration and $e_{1},\ldots,e_{k}$ be events such
that $\forall i\in\left\llbracket 1,k\right\rrbracket ,\hspace{1em}x\overset{e_{i}}{\ext}y_{i}$
and $\forall j\in\left\llbracket k+1,n\right\rrbracket ,\hspace{1em}x\overset{e_{j}}{\ext}y_{j}$,

We suppose that for all $a\in x\sqcup e_{i}$ and $b\in x\sqcup e_{j}$,
$a$ is compatible with $b$ - although conflicts within the $\left(e_{i}\right)_{i\in\left\llbracket 1,k\right\rrbracket }$
and within the $\left(e_{j}\right)_{j\in\left\llbracket k+1,k\right\rrbracket }$
are possible. Then,

\begin{gather*}
\qdv x{x\sqcup e_{1},\ldots,x\sqcup e_{n}}\otimes\tr_{Q_{0}\left(\cutpn x\right)}\\
=\\
\qdv x{x\sqcup e_{1},\ldots,x\sqcup e_{k}}\otimes\qdv x{x\sqcup e_{k+1},\ldots,x\sqcup e_{n}}
\end{gather*}
\end{thm}

\tcbline{}
\begin{proof}
We have 
\begin{alignat*}{1}
rhs & =\left[\sum_{\begin{array}{c}
I\subseteq\left\llbracket 1,k\right\rrbracket \\
x\sqcup e_{I}\in C(E)
\end{array}}\left(-1\right)^{|I|}\tr\left[\bigotimes_{i\in I}Q_{0}\left(\boxed{e_{i}}\right)\otimes\Id{K_{I}}\right]\right]\\
 & \otimes\left[\sum_{\begin{array}{c}
J\subseteq\left\llbracket k+1,n\right\rrbracket \\
x\sqcup e_{J}\in C(E)
\end{array}}\left(-1\right)^{|J|}\tr\left[\bigotimes_{j\in J}Q_{0}\left(\boxed{e_{j}}\right)\otimes\Id{K_{J}}\right]\right]\\
 & =\sum_{\begin{array}{c}
I\subseteq\left\llbracket 1,k\right\rrbracket \\
J\subseteq\left\llbracket k+1,n\right\rrbracket \\
x\sqcup e_{I}\in C(E)\\
x'\sqcup e_{J}\in C(E)
\end{array}}\left(-1\right)^{|I|+|J|}\tr\left[\bigotimes_{i\in I}Q_{0}\left(\boxed{e_{i}}\right)\otimes\Id{K_{I}}\otimes\bigotimes_{i\in J}Q_{0}\left(\boxed{e_{j}}\right)\otimes\Id{K_{J}}\right]
\end{alignat*}

Noticing that for $I\subseteq\left\llbracket 1,k\right\rrbracket $
and $J\subseteq\left\llbracket k+1,n\right\rrbracket $,
\begin{align*}
R_{I} & =\left(\bigsqcup_{i\in\left\llbracket 1,k\right\rrbracket \setminus I}\pre{\boxed{e_{i}}}\sqcup\bigsqcup_{i\in\left\llbracket k+1,n\right\rrbracket \setminus J}\pre{\boxed{e_{i}}}\right)\sqcup\bigsqcup_{i\in J}\pre{\boxed{e_{i}}}\\
 & :=\left(A\right)\sqcup B\\
R_{J} & =\bigsqcup_{i\in\left\llbracket 1,k\right\rrbracket \setminus I}\pre{\boxed{e_{i}}}\sqcup\bigsqcup_{i\in I}\pre{\boxed{e_{i}}}\sqcup\bigsqcup_{i\in\left\llbracket k+1,n\right\rrbracket \setminus J}\pre{\boxed{e_{i}}}\\
 & :=C\sqcup D\sqcup E
\end{align*}
we can rewrite the inner identity product:
\begin{align*}
\Id{K_{I}}\otimes\Id{K_{J}} & =\left[\Id{R_{I}}\otimes\Id X\right]\otimes\left[\Id{R_{J}}\otimes\Id X\right]\\
 & =\left[\Id A\otimes\Id X\right]\otimes\left[\Id{B\sqcup\left(C\sqcup D\sqcup E\right)}\otimes\Id X\right]\\
 & =\Id{K_{I\sqcup J}}\otimes\Id{Q_{0}\left(\cutpn x\right)}
\end{align*}

Hence, considering that the configurations $\begin{array}{c}
x\sqcup e_{I}\\
x\sqcup e_{J}
\end{array}$ are compatible if and only if $x\sqcup e_{I}\sqcup e_{J}$ is a configuration,
we can re-index the sum $I':=I\sqcup J$, finally yielding:
\begin{alignat*}{1}
rhs & =\sum_{\begin{array}{c}
I'\subseteq\left\llbracket 1,n\right\rrbracket \\
x\sqcup e_{I'}
\end{array}}\left(-1\right)^{|I'|}\tr\left[\bigotimes_{i'\in I'}Q_{0}\left(\boxed{e_{i'}}\right)\otimes\Id{K_{I'}}\right]\otimes\tr_{Q_{0}\left(\cutpn x\right)}
\end{alignat*}

Which proves the statement.
\end{proof}
\end{tcolorbox}

This factorization gives us a sufficient condition for the Drop Condition
to be satisfied, where only single extensions on conflict clusters
need to be checked. Indeed, the drop function is positive for the
Löwner order if and only if each of its term in its cluster-factorization
is. Hence the following definition of the Local Drop Condition.

\begin{tcolorbox}[breakable,enhanced,title=Local Drop Condition,colback=yellow!10!white,
colframe=red!60!black]

\begin{defn}
\label{def:Local Drop Condition} A quantum local annotation $Q_{0}$
on an annotated net skeleton $E=\left(P,T,F,Q_{0},H\right)$ satisfies
the Local Drop Condition if :
\end{defn}

\begin{itemize}
\item for all markings/configurations $\cutpn x$, 
\item and $\left(0,\oplus\right)$-conflict clique $C=e_{1},\ldots,e_{k}$
s.t. $x\overset{e_{i}}{\ext}y_{i}$, and s.t. the $e_{i}$'s are mutually
compatible
\end{itemize}
we have : $\qdv x{x\sqcup e_{1},\ldots,x\sqcup e_{k}}\sqsupseteq0$.
\end{tcolorbox}

\paragraph{Further simplifications for conflict cliques}

When a conflict cluster $C=\left\{ e_{1},\ldots,e_{k}\right\} $ is
in fact a clique, composed of events all in mutual conflict, we have
the further simplification: for all $I\subseteq\left\llbracket 1,n\right\rrbracket $,
$x\sqcup e_{I}\in C(E)\iff\left|I\right|\leq1$. Of course, no combination
of more than 1 elements of the $e_{i}$'s is compatible, and hence

\begin{tcolorbox}[breakable,enhanced,colback=green!12!white, colframe=green!40!black]
\begin{align*}
\qdv x{x\sqcup e_{1},\ldots,x\sqcup e_{k}} & =\sum_{\begin{array}{c}
I\subseteq\left\llbracket 1,k\right\rrbracket \\
x\sqcup e_{I}\in C(E)
\end{array}}\left(-1\right)^{|I|}\tr\left[\bigotimes_{i\in I}Q_{0}\left(\boxed{e_{i}}\right)\otimes\Id{K_{I}}\right]\\
 & =\sum_{e\in C}\tr\left[Q_{0}\left(\boxed{e}\right)\otimes\Id{Q_{0}\left(\pre{\left[C\setminus e\right]}\right)}\right]
\end{align*}
\end{tcolorbox}

In the case of conflict cliques, the Local Drop Condition is checked
in linear time $\mathcal{O}\left(\left|C\right|\right)$. 

\subsubsection{Further extensions}

Whether it could be enough to only check maximal clusters (and not
every cluster) in order to ensure the Drop Condition is an open question.
Also, it should be investigated whether a simple decomposition of
the drop condition solely in terms of conflict cliques exists. Thus,
applying the simplification expressed in the last section would yield
a practical and combinatorially efficient way to check the drop-conditions.

\subsection{Local Quantum Occurrence Net }

We conclude this section with our definition of Local Quantum Occurrence
Net, w.r.t. the local properties proved earlier.

\begin{tcolorbox}[breakable,enhanced,title=Local Quantum Occurrence Net,colback=olive!20!white,
colframe=olive!70!black,]

\begin{defn}
\label{def:Local Quantum Occurrence Net}A local Quantum Occurrence
Net $E,Q_{0},H$ is an Occurrence Net $E$ endowed with a local Quantum
Annotation $Q_{0},H$, that satisfies the \emph{Local Obliviousness}
(Definition \ref{def:Local-Obliviousness}) and the \emph{Local Drop
Condition} (Definition \ref{def:Local Drop Condition}). 
\end{defn}

\end{tcolorbox}

\section{Quantum Petri Nets\label{sec:Quantum Petri Nets}}

After having defined Quantum Occurrence Nets, we now ought to establish
a sensible framework for Quantum Petri Nets. Classical Occurrence
Nets and Petri Nets are linked through the process called ``unfolding'',
where to the Petri net is associated one unique Occurrence Net that
encaptions all the causality and concurrency relation of the initial
net. This has been thoroughly explored in the literature.

We mimic this correspondence by first defining Quantum Petri Nets
in relation with their respective causal unfolding ( Section \ref{subsec:Lifting Quantum Occurrence Nets to Quantum Petri Nets}).
Then, we prove that enforcing the local conditions (Local Obliviousness
and Local Drop Condition) on a Petri Net with a Quantum Annotation
is enough to make it a Quantum Petri Net (Section \ref{subsec:Condition for a Locally Annotated Petri Net to be a Quantum Petri Net}).

Finally, we investigate the interaction of two Quantum Petri Nets
(i.e. a stub concept for the composition of Quantum Petri Nets, that
could be developed in further work) in Section \ref{subsec:Interaction of Quantum Petri Nets}.

\subsection{Lifting Quantum Occurrence Nets to Quantum Petri Nets\label{subsec:Lifting Quantum Occurrence Nets to Quantum Petri Nets}}

We define here the unfolding of classical Petri Nets, as their maximal
branching process.

\begin{tcolorbox}[breakable,enhanced,title=Branching Process - Unfolding of a Petri
Net,, colback=purple!5!white, colframe=purple!75!black]

\begin{defn}
\label{def:Branching-Process-Unfolding} {[}from \autocite{haarComputingRevealsRelation2013}{]}
$\rightslice$ A \emph{branching process} of a net system $\mathcal{N}=\left(S,T,W,M_{0}\right)$
is a labeled occurrence net $\beta=(O;p)=(B;E;F;p)$ where the labeling
function $p$ satisfies the following properties:
\begin{enumerate}
\item $p(B)\subseteq S$ and $p(E)\subseteq T$\\
(p preserves the nature of nodes);
\item For every $\boxed{e}\in E$, the restriction of $p$ to $\boxed{e}$
is a bijection between $\pre e$ and $\pre{p\left(\boxed{e}\right)}$
similarly for $\post{\boxed{e}}$ and $\post{p\left(\boxed{e}\right)}$\\
($p$ preserves the environments of transitions);
\item the restriction of $p$ to $Min(O)$ is a bijection between $Min(O)$
and $M_{0}$\\
( $\beta$ ``starts\textquotedbl{} at $M_{0}$);
\item for every $\boxed{e_{1}},\boxed{e_{2}}\in E$, if $\pre{\boxed{e_{1}}}=\pre{\boxed{e_{2}}}$
and $p\left(\boxed{e_{1}}\right)=p\left(\boxed{e_{2}}\right)$ then
$\boxed{e_{1}}=\boxed{e_{2}}$\\
( does not duplicate the transitions ).
\end{enumerate}
\end{defn}

$\rightslice$ The \emph{unfolding of a Petri net} is the unique branching
process that unfolds as much as possible, written $\mathcal{U}\left(\mathcal{N}\right)$
.

\end{tcolorbox}

An example of Branching Process is represented in Figure \ref{fig:Occurrence Net};
it is a prefix of the unfolding of the net in Figure \ref{fig:A-safe-Petri-net},
the $p$ morphism has been made explicit by naming the events in the
unfolding by their corresponding event in the original net.

\paragraph{Defining Quantum Petri Nets from classical unfoldings}

In a first approach, we start by defining a Quantum Petri Net in terms
of the unfolding of a classical Petri Net, that is decorated with
a local Quantum Annotation, and that satisfies the properties of a
Local Quantum Occurrence Net. This yields the following definition.
\begin{tcolorbox}[breakable,enhanced,title=Quantum Petri Net - Unfolding based definition,colback=purple!5!white,
colframe=purple!75!black]

\begin{defn}
\label{def:Quantum Petri Net - Unfolding based definition}A Quantum
Petri Net $\mathtt{N}=\mathcal{N},Q_{0},H$ is formed by a couple
of a Petri net system $\mathcal{N}=\left(P,T,F,\m 0\right)$ and a
local quantum annotation $Q_{0},H$ of the unfolding of $\mathcal{N}$(see
Definition \ref{def:Annotation-of-an}), that satisfies the following
property:
\begin{itemize}
\item The annotated unfolding $\mathcal{U},Q_{0},H$ of the net system $\mathcal{N}$
is a Local Quantum Occurrence Net (Definition \ref{def:Local Quantum Occurrence Net})
- where $Q_{0},H$ are extended naturally to every conditions and
events with the same labels 
\end{itemize}
\end{defn}

\end{tcolorbox}

\begin{rem*}
The condition in the last definition is equivalent to asking for all
annotated branching processes to be Local Quantum Occurrence Net.
\end{rem*}

\subsection{Condition for a Locally Annotated Petri Net to be a Quantum Petri
Net\label{subsec:Condition for a Locally Annotated Petri Net to be a Quantum Petri Net}}

Instead of annotating the \textbackslash emph\{unfolding\} of a classical
Petri net--which is unwieldy--we annotate the Petri net \emph{directly}.
Let $Q_{0}^{\mathrm{PN}}$ denote this local quantum annotation on
the Petri net $\mathcal{N}$. By unfolding, $Q_{0}^{\mathrm{PN}}$
induces a local quantum annotation $Q_{0}^{\mathcal{U}}$ on the classical
unfolding $\mathcal{U}(\mathcal{N})$.

We seek conditions under which this induced annotation makes $\mathbf{N}=\left(\mathcal{N},Q_{0},H\right)$
a \emph{Quantum Petri Net} in the sense of Definition \ref{def:Quantum Petri Net - Unfolding based definition}.

Our main point is that it suffices to enforce the \emph{Local Drop
Condition} and \emph{Local Obliviousness} \emph{on the Petri net itself}:
if $Q_{0}^{\mathrm{PN}}$ satisfies these two local conditions, then
its induced $Q_{0}^{\mathcal{U}}$ equips $\mathcal{U}(\mathcal{N})$
accordingly, and $\mathbf{N}$ is a Quantum Petri Net.

\begin{tcolorbox}[breakable,enhanced,title=Quantum Petri Net - Local definition,colback=purple!5!white,
colframe=purple!75!black]

\begin{thm}
\label{thm:Quantum Petri Net - Local definition} Let $\mathtt{N}=\mathcal{N},Q_{0},H$
be a Locally annotated Petri Net. Then if $Q_{0}$ satisfies the Local
Drop Condition and Local Obliviousness, $\mathtt{N}$ is a Quantum
Petri Net . 
\end{thm}

\begin{defn}
This allows us to define a Quantum Petri Net as a Locally annotated
Petri Net that satisfies the Local Drop Condition and the Local Obliviousness.
\end{defn}

\begin{proof}
Let $\mathcal{U}$ be the unfolding of $\mathcal{N}$. Then let us
show that the extended annotation $Q_{0}$ of the annotated occurrence
net $\mathcal{U},Q_{0},H$ also satisfies the Local Drop Condition
and Local Obliviousness--(note the slight abuse of notation, where
we identify $Q_{0}^{PN}$ and $Q_{0}^{\mathcal{U}}$ as $Q_{0}$):
\begin{itemize}
\item There is a bijection between the conflict clusters of $\mathcal{U}$
(quotiented by their sets of labels) and those of $\mathcal{N}$.
Thus the Local Drop condition on $\mathcal{U}$ is enforced if and
only if it is also verified on $\mathcal{N}$
\item Since the environments of events of in the net and its unfolding are
in bijection, the Local Obliviousness on $\mathcal{U}$ is enforced
if and only if it is also verified on $\mathcal{N}$
\end{itemize}
\end{proof}
\end{tcolorbox}

\subsection{Interaction of Quantum Petri Nets\label{subsec:Interaction of Quantum Petri Nets}}

\global\long\def\N#1{\mathcal{N}_{#1}}%

\global\long\def\ptfm#1{\left(P_{#1},T_{#1},F_{#1},\m{#1}\right)}%

In this section, we endow Quantum Petri Net with a stub for a ``composition''
operation, that will be expanded in further work. More precisely,
we first define the Parallel Composition of two Quantum Petri Nets,
and show that the resulting net is also a Quantum Petri Nets ( Section
\ref{subsec:Parallel Composition}). This enables us to investigate
Quantum Petri Nets, that are joined together by merging positive and
negative events (i.e feeding one's inputs to the other's outputs,
irrespectively). Some conditions need to be respected for the obtained
net to remain a Quantum Petri Net. In Section \ref{subsec:Joins},
we give a sufficient condition for the preservation of the drop-condition
after a composition of the sort. Extensions will investigate other
characterizations of compositions preserving QPN properties.

\subsubsection{Parallel Composition\label{subsec:Parallel Composition}}

The parallel composition of two Quantum Petri Nets - i.e. the juxtaposition
of the two - remains a Quantum Petri Net, as per the following definitions.

\begin{tcolorbox}[breakable,enhanced,title=Parallel Composition of Quantum Petri Nets,colback=orange!5!white,
colframe=orange!75!black]

\begin{defn}
Let $\N 1=\ptfm 1,Q_{0}^{(1)},H^{(1)}$ and $\N 2=\ptfm 2,Q_{0}^{(2)},H^{(2)}$
be two Quantum Petri nets. Define the parallel composition of those
nets as 
\[
\N 1||\N 2=\left(P_{1}\sqcup P_{2},T_{1}\sqcup T{}_{2},F_{1}\sqcup F_{2},\m 1\cdot\m 2\right),Q_{0}^{(1)}||Q_{0}^{(2)},H^{(1)}||H^{(2)}
\]
where $Q_{0}^{(1)}||Q_{0}^{(2)}$ is the map that applies $Q_{0}^{(1)}$on
$P_{1}$ and $T_{1}$ (and similarly for $Q_{0}^{(2)}$). Same principle
for $H^{(1,2)}$. 

Then $\N 1||\N 2$ is a Quantum Petri Net.
\end{defn}

\tcbline{}
\begin{proof}
$Q_{0}^{(1)}||Q_{0}^{(2)}$ also verifies Local Obliviousness and
the Local Drop Condition
\end{proof}
\end{tcolorbox}

\subsubsection{Joins Preserving the Drop Condition\label{subsec:Joins}}

One can be interested in the interaction between positive and negative
events sharing the same input and output space, respectively. They
represent dual processes that can accept as input the other's output.
The wanted effect should be a net where the complementary events of
interest have been merged together. We call this operation a \emph{join.
}The single join of two events is defined as follows.

\begin{tcolorbox}[breakable,enhanced,title=Single Join of Two Events, colback=orange!5!white,
colframe=orange!75!black]

\begin{defn}
\label{def:Join of two events}Let $N=\left(P,T,F\right),Q_{0},H$
be an annotated Net Skeleton and $\boxed{e}$ and $\boxed{e'}$ be
two polarized transitions such that $p\left(e\right)=\oplus$, $p\left(e'\right)=\ominus$,
and $H\left(e\right)=H\left(e'\right)$. The annotated net skeleton
resulting from the join of the events $e$ and $e'$ is $N_{e\bowtie e'}=\left(P,T_{e\bowtie e'},F_{e\bowtie e'}\right),{Q_{0}}_{e\bowtie e'},{H}_{e\bowtie e'}$
where:
\end{defn}

\begin{itemize}
\item $T':=T\setminus\left\{ \boxed{e}\cup\boxed{e'}\right\} \cup\boxed{e\bowtie e'}$
and
\item For each arc $\ci a\rightarrow\boxed{e}$ or $\ci b\rightarrow\boxed{e'}$
in $F$ there is a corresponding arc $\ci a\rightarrow\boxed{e\bowtie e'}$
and $\ci b\rightarrow\boxed{e\bowtie e'}$ in $F_{e\bowtie e'}$ (similarly
for $\boxed{e}\rightarrow\ci a$ and $\boxed{e'}\rightarrow\ci b$)
\item Furthermore, ${Q_{0}}_{e\bowtie e'},{H}_{e\bowtie e'}$ is defined
in the natural way, where 
\begin{equation}
\forall\boxed{t}\in T_{e\bowtie e'},\hspace*{1em}{Q_{0}}_{e\bowtie e'}\left(\boxed{t}\right):=\begin{cases}
Q_{0}\left(\boxed{e}\right)\otimes\Id{Q_{0}\left(\pre{\boxed{e'}}\right)} & ,\boxed{t}=\boxed{e\bowtie e'}\\
Q_{0}\left(\boxed{t}\right) & o/w
\end{cases}
\end{equation}
(noticing that the expression $Q_{0}\left(\boxed{e}\right)\otimes\Id{Q_{0}\left(\pre{\boxed{e'}}\right)}$
is in fact the the composition $Q_{0}\left(\boxed{e}\right)\otimes\Id{Q_{0}\left(\pre{\boxed{e'}}\right)}=\left[\Id{Q\left(\post{\boxed{e}}\right)}\otimes Q_{0}\left(\boxed{e'}\right)\right]\circ\left[Q_{0}\left(\boxed{e}\right)\otimes\Id{Q\left(\pre{\boxed{e'}}\right)}\right]$)
\end{itemize}
\end{tcolorbox}

Single joins within a Quantum Petri Net can be realized in sequence,
with successive different pairs of events. But in order for the resulting
net to also be a Quantum Petri Net, some conditions on this merge
should be respected. We introduce in Definition \ref{def:Drop-preserving join}
a sufficient condition for such a sequence to preserve the Quantum
Petri Net properties. We name such joins: ``Drop-Preserving Joins''. 

\begin{tcolorbox}[breakable,enhanced,title=Drop-preserving join, colback=orange!5!white,
colframe=orange!75!black]

\begin{defn}
\label{def:Drop-preserving join}Let $S$ be a net skeleton, and $S'$
be the net obtained after a finite sequence of join operations. The
sequence of join operations is called a drop-preserving join if there
exists within $S$:
\begin{enumerate}
\item $N$ a maximal negative cluster and $P$ a strictly positive cluster
(possibly included in some maximal positive or neutral cluster $\mathds{P}\phantom{}^{0/\oplus}\supseteq P$)\label{enu:point 1}
\item $f$ a bijective map $f:N\rightarrow P$ verifying
\begin{enumerate}
\item if $\boxed{a}^{\ominus}\sim\boxed{b}^{\ominus}$ in $N$ then $f\left(\boxed{a}^{\ominus}\right)\sim f\left(\boxed{b}^{\ominus}\right)$
in $P$\label{enu:point2a}
\item $\forall\boxed{e}\in N,H\left(f\left(\boxed{e}\right)\right)=H\left(\boxed{e}\right)$
\end{enumerate}
\end{enumerate}
such that $S'$ is the successive join of the events $f\left(\boxed{e}\right)\bowtie\boxed{e}$
in $S$ for $\boxed{e}\in N$. The order of the joins does not matter.
\end{defn}

\end{tcolorbox}

\begin{rem}
\label{rem:After-a-drop-preserving} After a drop-preserving join,
the clusters of $S'$ are exactly those of $S$, minus $N$ and $\mathds{P}$,
which are replaced by a new cluster $\mathds{P}\bowtie_{f}N$ which
events are those of $\mathds{P}\setminus P\sqcup\left\{ n\bowtie p|n\in N,p\in P\right\} $.
\end{rem}

\begin{example}
An example of drop preserving join is illustrated in Figure \ref{fig:Illustration of a drop preserving join}.
\end{example}

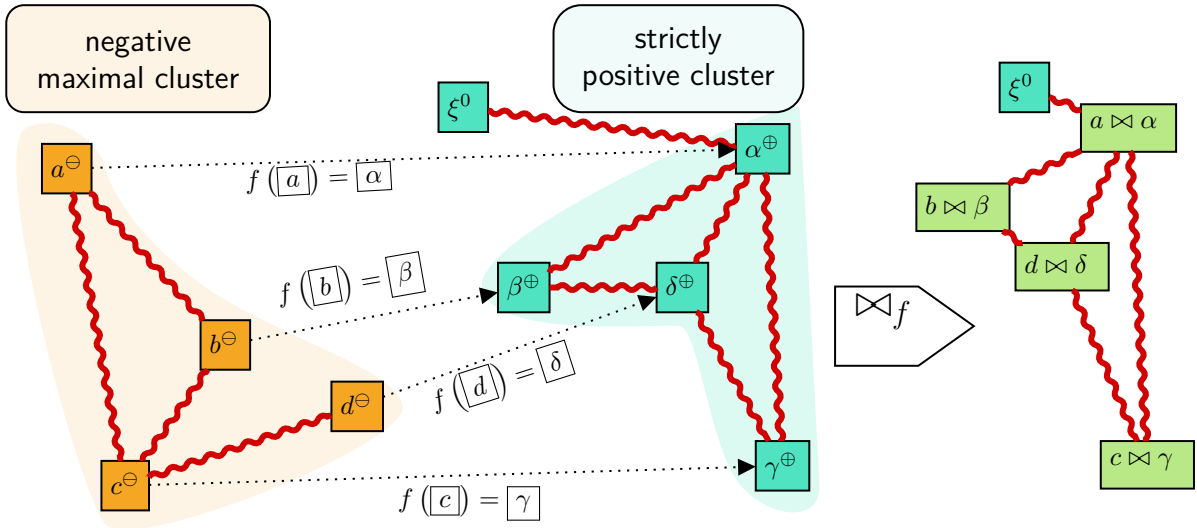
\begin{figure}[H]
\centering

\tikzset{every picture/.style={line width=0.75pt}} 

\begin{tikzpicture}[x=0.75pt,y=0.75pt,yscale=-1,xscale=1]

\draw  [draw opacity=0][fill={rgb, 255:red, 245; green, 166; blue, 35 }  ,fill opacity=0.12 ] (9.99,70) .. controls (29.99,60) and (48.99,81.82) .. (49.99,90) .. controls (50.99,98.18) and (64.99,125.18) .. (99.99,150) .. controls (134.99,174.82) and (195.99,190.82) .. (199.99,210) .. controls (203.99,229.18) and (63.99,285.82) .. (49.99,270) .. controls (35.99,254.18) and (-10.01,80) .. (9.99,70) -- cycle ;
\draw  [draw opacity=0][fill={rgb, 255:red, 80; green, 227; blue, 194 }  ,fill opacity=0.2 ] (239.99,150) .. controls (235.99,122.18) and (398.99,51.82) .. (399.99,60) .. controls (400.99,68.18) and (409,229.82) .. (410,250) .. controls (411,270.18) and (384.67,267.18) .. (370,250) .. controls (355.33,232.82) and (349,171.82) .. (340,170) .. controls (331,168.18) and (243.99,177.82) .. (239.99,150) -- cycle ;
\draw   (418,150) -- (460,150) -- (488,170) -- (460,190) -- (418,190) -- cycle ;

\draw  [fill={rgb, 255:red, 245; green, 166; blue, 35 }  ,fill opacity=1 ]  (17.99,78) -- (42.99,78) -- (42.99,103) -- (17.99,103) -- cycle  ;
\draw (20.99,82.4) node [anchor=north west][inner sep=0.75pt]    {$a^{\ominus }$};
\draw  [fill={rgb, 255:red, 245; green, 166; blue, 35 }  ,fill opacity=1 ]  (97.99,167) -- (122.99,167) -- (122.99,192) -- (97.99,192) -- cycle  ;
\draw (100.99,171.4) node [anchor=north west][inner sep=0.75pt]    {$b^{\ominus }$};
\draw  [fill={rgb, 255:red, 245; green, 166; blue, 35 }  ,fill opacity=1 ]  (163.99,198) -- (189.99,198) -- (189.99,223) -- (163.99,223) -- cycle  ;
\draw (166.99,202.4) node [anchor=north west][inner sep=0.75pt]    {$d^{\ominus }$};
\draw  [fill={rgb, 255:red, 245; green, 166; blue, 35 }  ,fill opacity=1 ]  (47.99,238) -- (71.99,238) -- (71.99,263) -- (47.99,263) -- cycle  ;
\draw (50.99,242.4) node [anchor=north west][inner sep=0.75pt]    {$c^{\ominus }$};
\draw  [fill={rgb, 255:red, 80; green, 227; blue, 194 }  ,fill opacity=1 ]  (328,138) -- (354,138) -- (354,163) -- (328,163) -- cycle  ;
\draw (331,142.4) node [anchor=north west][inner sep=0.75pt]    {$\delta ^{\oplus }$};
\draw  [fill={rgb, 255:red, 80; green, 227; blue, 194 }  ,fill opacity=1 ]  (367.99,68) -- (394.99,68) -- (394.99,93) -- (367.99,93) -- cycle  ;
\draw (370.99,72.4) node [anchor=north west][inner sep=0.75pt]    {$\alpha ^{\oplus }$};
\draw  [fill={rgb, 255:red, 80; green, 227; blue, 194 }  ,fill opacity=1 ]  (378,228) -- (405,228) -- (405,253) -- (378,253) -- cycle  ;
\draw (381,232.4) node [anchor=north west][inner sep=0.75pt]    {$\gamma ^{\oplus }$};
\draw  [fill={rgb, 255:red, 80; green, 227; blue, 194 }  ,fill opacity=1 ]  (247.99,138) -- (273.99,138) -- (273.99,163) -- (247.99,163) -- cycle  ;
\draw (250.99,142.4) node [anchor=north west][inner sep=0.75pt]    {$\beta ^{\oplus }$};
\draw  [fill={rgb, 255:red, 80; green, 227; blue, 194 }  ,fill opacity=1 ]  (218,47) -- (243,47) -- (243,72) -- (218,72) -- cycle  ;
\draw (221,51.4) node [anchor=north west][inner sep=0.75pt]    {$\xi ^{0}$};
\draw (118.36,88.39) node [anchor=north west][inner sep=0.75pt]  [rotate=-356.76]  {$f\left(\boxed{a}\right) =\boxed{\alpha }$};
\draw (133.26,143.67) node [anchor=north west][inner sep=0.75pt]  [rotate=-350.09]  {$f\left(\boxed{b}\right) =\boxed{\beta }$};
\draw (209.62,198.01) node [anchor=north west][inner sep=0.75pt]  [rotate=-342.39]  {$f\left(\boxed{d}\right) =\boxed{\delta }$};
\draw (195.99,248.4) node [anchor=north west][inner sep=0.75pt]    {$f\left(\boxed{c}\right) =\boxed{\gamma }$};
\draw  [fill={rgb, 255:red, 245; green, 166; blue, 35 }  ,fill opacity=0.12 ]  (0,17) .. controls (0,12.03) and (4.03,8) .. (9,8) -- (124,8) .. controls (128.97,8) and (133,12.03) .. (133,17) -- (133,53) .. controls (133,57.97) and (128.97,62) .. (124,62) -- (9,62) .. controls (4.03,62) and (0,57.97) .. (0,53) -- cycle  ;
\draw (66.5,35) node  [font=\large] [align=left] {\begin{minipage}[lt]{88.08pt}\setlength\topsep{0pt}
\begin{center}
negative \\maximal cluster
\end{center}

\end{minipage}};
\draw  [fill={rgb, 255:red, 80; green, 227; blue, 194 }  ,fill opacity=0.08 ]  (276,22) .. controls (276,14.27) and (282.27,8) .. (290,8) -- (388,8) .. controls (395.73,8) and (402,14.27) .. (402,22) -- (402,48) .. controls (402,55.73) and (395.73,62) .. (388,62) -- (290,62) .. controls (282.27,62) and (276,55.73) .. (276,48) -- cycle  ;
\draw (339,35) node  [font=\large] [align=left] {\begin{minipage}[lt]{83.38pt}\setlength\topsep{0pt}
\begin{center}
strictly \\positive cluster
\end{center}

\end{minipage}};
\draw  [fill={rgb, 255:red, 184; green, 233; blue, 134 }  ,fill opacity=1 ]  (509,128) -- (556,128) -- (556,152) -- (509,152) -- cycle  ;
\draw (512,132.4) node [anchor=north west][inner sep=0.75pt]    {$d\bowtie \delta $};
\draw  [fill={rgb, 255:red, 184; green, 233; blue, 134 }  ,fill opacity=1 ]  (542,58) -- (590,58) -- (590,82) -- (542,82) -- cycle  ;
\draw (545,62.4) node [anchor=north west][inner sep=0.75pt]    {$a\bowtie \alpha $};
\draw  [fill={rgb, 255:red, 184; green, 233; blue, 134 }  ,fill opacity=1 ]  (552,228) -- (599,228) -- (599,252) -- (552,252) -- cycle  ;
\draw (555,232.4) node [anchor=north west][inner sep=0.75pt]    {$c\bowtie \gamma $};
\draw  [fill={rgb, 255:red, 184; green, 233; blue, 134 }  ,fill opacity=1 ]  (459,98) -- (506,98) -- (506,122) -- (459,122) -- cycle  ;
\draw (462,102.4) node [anchor=north west][inner sep=0.75pt]    {$b\bowtie \beta $};
\draw  [fill={rgb, 255:red, 80; green, 227; blue, 194 }  ,fill opacity=1 ]  (501,37) -- (526,37) -- (526,62) -- (501,62) -- cycle  ;
\draw (504,41.4) node [anchor=north west][inner sep=0.75pt]    {$\xi ^{0}$};
\draw (442.5,163) node  [font=\LARGE]  {$\bowtie_{f}$};
\draw [color={rgb, 255:red, 208; green, 2; blue, 2 }  ,draw opacity=1 ][line width=2.25]    (41.73,103) .. controls (44.08,103.13) and (45.2,104.37) .. (45.07,106.72) .. controls (44.94,109.07) and (46.06,110.31) .. (48.41,110.44) .. controls (50.76,110.57) and (51.88,111.81) .. (51.76,114.16) .. controls (51.63,116.51) and (52.75,117.75) .. (55.1,117.87) .. controls (57.45,118) and (58.57,119.24) .. (58.44,121.59) .. controls (58.31,123.94) and (59.43,125.18) .. (61.78,125.31) .. controls (64.13,125.44) and (65.25,126.68) .. (65.13,129.03) .. controls (65,131.38) and (66.12,132.62) .. (68.47,132.75) .. controls (70.82,132.88) and (71.94,134.12) .. (71.81,136.47) .. controls (71.68,138.82) and (72.8,140.06) .. (75.15,140.19) .. controls (77.5,140.31) and (78.62,141.55) .. (78.5,143.9) .. controls (78.37,146.25) and (79.49,147.49) .. (81.84,147.62) .. controls (84.19,147.75) and (85.31,148.99) .. (85.18,151.34) .. controls (85.05,153.69) and (86.17,154.93) .. (88.52,155.06) .. controls (90.87,155.19) and (91.99,156.43) .. (91.87,158.78) .. controls (91.74,161.13) and (92.86,162.37) .. (95.21,162.5) .. controls (97.56,162.63) and (98.68,163.87) .. (98.55,166.22) -- (99.26,167) -- (99.26,167) ;
\draw [color={rgb, 255:red, 208; green, 2; blue, 2 }  ,draw opacity=1 ][line width=2.25]    (32.8,103) .. controls (34.74,104.34) and (35.04,105.98) .. (33.7,107.92) .. controls (32.37,109.86) and (32.67,111.5) .. (34.61,112.83) .. controls (36.55,114.17) and (36.85,115.81) .. (35.52,117.75) .. controls (34.18,119.69) and (34.48,121.33) .. (36.42,122.67) .. controls (38.36,124.01) and (38.66,125.65) .. (37.33,127.59) .. controls (36,129.53) and (36.3,131.17) .. (38.24,132.5) .. controls (40.18,133.84) and (40.48,135.48) .. (39.14,137.42) .. controls (37.81,139.36) and (38.11,141) .. (40.05,142.34) .. controls (41.99,143.67) and (42.29,145.31) .. (40.96,147.25) .. controls (39.62,149.19) and (39.92,150.83) .. (41.86,152.17) .. controls (43.8,153.51) and (44.1,155.15) .. (42.77,157.09) .. controls (41.44,159.03) and (41.74,160.67) .. (43.68,162.01) .. controls (45.62,163.34) and (45.92,164.98) .. (44.58,166.92) .. controls (43.25,168.86) and (43.55,170.5) .. (45.49,171.84) .. controls (47.43,173.18) and (47.73,174.82) .. (46.4,176.76) .. controls (45.06,178.7) and (45.36,180.34) .. (47.3,181.67) .. controls (49.24,183.01) and (49.54,184.65) .. (48.21,186.59) .. controls (46.88,188.53) and (47.18,190.17) .. (49.12,191.51) .. controls (51.06,192.85) and (51.36,194.49) .. (50.02,196.43) .. controls (48.69,198.37) and (48.99,200.01) .. (50.93,201.34) .. controls (52.87,202.68) and (53.17,204.32) .. (51.84,206.26) .. controls (50.5,208.2) and (50.8,209.84) .. (52.74,211.18) .. controls (54.68,212.51) and (54.98,214.15) .. (53.65,216.09) .. controls (52.32,218.03) and (52.62,219.67) .. (54.56,221.01) .. controls (56.5,222.35) and (56.8,223.99) .. (55.46,225.93) .. controls (54.13,227.87) and (54.43,229.51) .. (56.37,230.85) .. controls (58.31,232.18) and (58.61,233.82) .. (57.28,235.76) -- (57.69,238) -- (57.69,238) ;
\draw [color={rgb, 255:red, 208; green, 2; blue, 2 }  ,draw opacity=1 ][line width=2.25]    (68.88,238) .. controls (68.49,235.68) and (69.45,234.32) .. (71.78,233.93) .. controls (74.11,233.54) and (75.07,232.18) .. (74.68,229.85) .. controls (74.29,227.52) and (75.25,226.17) .. (77.58,225.78) .. controls (79.91,225.39) and (80.87,224.03) .. (80.48,221.7) .. controls (80.09,219.38) and (81.05,218.02) .. (83.37,217.63) .. controls (85.7,217.24) and (86.66,215.88) .. (86.27,213.55) .. controls (85.88,211.22) and (86.84,209.87) .. (89.17,209.48) .. controls (91.5,209.09) and (92.46,207.73) .. (92.07,205.4) .. controls (91.68,203.07) and (92.64,201.72) .. (94.97,201.33) .. controls (97.29,200.94) and (98.25,199.58) .. (97.86,197.26) .. controls (97.47,194.93) and (98.43,193.57) .. (100.76,193.18) -- (101.6,192) -- (101.6,192) ;
\draw [color={rgb, 255:red, 208; green, 2; blue, 2 }  ,draw opacity=1 ][line width=2.25]    (71.99,246.4) .. controls (73.03,244.28) and (74.61,243.74) .. (76.72,244.78) .. controls (78.84,245.82) and (80.42,245.28) .. (81.45,243.16) .. controls (82.49,241.04) and (84.07,240.5) .. (86.19,241.54) .. controls (88.3,242.58) and (89.88,242.04) .. (90.92,239.93) .. controls (91.95,237.81) and (93.53,237.27) .. (95.65,238.31) .. controls (97.77,239.35) and (99.35,238.81) .. (100.38,236.69) .. controls (101.42,234.58) and (103,234.04) .. (105.11,235.08) .. controls (107.23,236.12) and (108.81,235.58) .. (109.84,233.46) .. controls (110.87,231.34) and (112.45,230.8) .. (114.57,231.84) .. controls (116.69,232.88) and (118.27,232.34) .. (119.3,230.22) .. controls (120.35,228.11) and (121.93,227.57) .. (124.04,228.61) .. controls (126.16,229.65) and (127.74,229.11) .. (128.77,226.99) .. controls (129.8,224.87) and (131.38,224.33) .. (133.5,225.37) .. controls (135.62,226.41) and (137.2,225.87) .. (138.23,223.75) .. controls (139.27,221.64) and (140.85,221.1) .. (142.96,222.14) .. controls (145.08,223.18) and (146.66,222.64) .. (147.69,220.52) .. controls (148.72,218.4) and (150.3,217.86) .. (152.42,218.9) .. controls (154.54,219.94) and (156.12,219.4) .. (157.15,217.28) .. controls (158.19,215.17) and (159.77,214.63) .. (161.88,215.67) -- (163.99,214.94) -- (163.99,214.94) ;
\draw [color={rgb, 255:red, 208; green, 2; blue, 2 }  ,draw opacity=1 ][line width=2.25]    (348.23,138) .. controls (347.62,135.73) and (348.46,134.28) .. (350.73,133.67) .. controls (353.01,133.06) and (353.85,131.62) .. (353.24,129.34) .. controls (352.63,127.07) and (353.47,125.63) .. (355.74,125.02) .. controls (358.02,124.41) and (358.86,122.97) .. (358.25,120.69) .. controls (357.64,118.41) and (358.47,116.97) .. (360.75,116.36) .. controls (363.03,115.75) and (363.86,114.31) .. (363.25,112.03) .. controls (362.64,109.75) and (363.48,108.31) .. (365.76,107.7) .. controls (368.03,107.09) and (368.87,105.65) .. (368.26,103.38) .. controls (367.65,101.1) and (368.48,99.66) .. (370.76,99.05) .. controls (373.04,98.44) and (373.88,97) .. (373.27,94.72) -- (374.26,93) -- (374.26,93) ;
\draw [color={rgb, 255:red, 208; green, 2; blue, 2 }  ,draw opacity=1 ][line width=2.25]    (328,150.5) .. controls (326.33,152.17) and (324.67,152.17) .. (323,150.5) .. controls (321.33,148.83) and (319.67,148.83) .. (318,150.5) .. controls (316.33,152.17) and (314.67,152.17) .. (313,150.5) .. controls (311.33,148.83) and (309.67,148.83) .. (308,150.5) .. controls (306.33,152.17) and (304.67,152.17) .. (303,150.5) .. controls (301.33,148.83) and (299.67,148.83) .. (298,150.5) .. controls (296.33,152.17) and (294.67,152.17) .. (293,150.5) .. controls (291.33,148.83) and (289.67,148.83) .. (288,150.5) .. controls (286.33,152.17) and (284.67,152.17) .. (283,150.5) .. controls (281.33,148.83) and (279.67,148.83) .. (278,150.5) -- (273.99,150.5) -- (273.99,150.5) ;
\draw [color={rgb, 255:red, 208; green, 2; blue, 2 }  ,draw opacity=1 ][line width=2.25]    (273.99,142.95) .. controls (274.6,140.67) and (276.04,139.83) .. (278.32,140.44) .. controls (280.6,141.05) and (282.04,140.21) .. (282.64,137.93) .. controls (283.24,135.65) and (284.68,134.81) .. (286.96,135.41) .. controls (289.24,136.02) and (290.68,135.18) .. (291.29,132.9) .. controls (291.89,130.62) and (293.33,129.78) .. (295.61,130.39) .. controls (297.89,131) and (299.33,130.16) .. (299.93,127.88) .. controls (300.54,125.6) and (301.98,124.76) .. (304.26,125.37) .. controls (306.54,125.98) and (307.98,125.14) .. (308.58,122.86) .. controls (309.18,120.58) and (310.62,119.74) .. (312.9,120.34) .. controls (315.18,120.95) and (316.62,120.11) .. (317.23,117.83) .. controls (317.83,115.55) and (319.27,114.71) .. (321.55,115.32) .. controls (323.83,115.93) and (325.27,115.09) .. (325.87,112.81) .. controls (326.48,110.53) and (327.92,109.69) .. (330.2,110.3) .. controls (332.48,110.91) and (333.92,110.07) .. (334.52,107.79) .. controls (335.12,105.51) and (336.56,104.67) .. (338.84,105.27) .. controls (341.12,105.88) and (342.56,105.04) .. (343.17,102.76) .. controls (343.77,100.48) and (345.21,99.64) .. (347.49,100.25) .. controls (349.77,100.86) and (351.21,100.02) .. (351.81,97.74) .. controls (352.42,95.46) and (353.86,94.62) .. (356.14,95.23) .. controls (358.42,95.84) and (359.86,95) .. (360.46,92.72) .. controls (361.06,90.44) and (362.5,89.6) .. (364.78,90.21) -- (367.99,88.34) -- (367.99,88.34) ;
\draw [color={rgb, 255:red, 208; green, 2; blue, 2 }  ,draw opacity=1 ][line width=2.25]    (382.27,93) .. controls (384.04,94.56) and (384.15,96.22) .. (382.59,97.99) .. controls (381.03,99.76) and (381.13,101.42) .. (382.9,102.98) .. controls (384.67,104.54) and (384.77,106.2) .. (383.21,107.97) .. controls (381.65,109.74) and (381.75,111.4) .. (383.52,112.96) .. controls (385.29,114.52) and (385.39,116.18) .. (383.83,117.95) .. controls (382.28,119.72) and (382.38,121.38) .. (384.15,122.94) .. controls (385.92,124.5) and (386.02,126.16) .. (384.46,127.93) .. controls (382.9,129.7) and (383,131.36) .. (384.77,132.92) .. controls (386.54,134.48) and (386.64,136.14) .. (385.08,137.91) .. controls (383.53,139.68) and (383.63,141.34) .. (385.4,142.9) .. controls (387.17,144.46) and (387.27,146.12) .. (385.71,147.89) .. controls (384.15,149.66) and (384.25,151.32) .. (386.02,152.88) .. controls (387.79,154.44) and (387.89,156.1) .. (386.33,157.87) .. controls (384.77,159.64) and (384.87,161.3) .. (386.64,162.86) .. controls (388.41,164.42) and (388.51,166.08) .. (386.96,167.85) .. controls (385.4,169.62) and (385.5,171.28) .. (387.27,172.84) .. controls (389.04,174.4) and (389.14,176.06) .. (387.58,177.83) .. controls (386.02,179.6) and (386.12,181.26) .. (387.89,182.82) .. controls (389.66,184.38) and (389.76,186.04) .. (388.2,187.81) .. controls (386.65,189.58) and (386.75,191.24) .. (388.52,192.8) .. controls (390.29,194.37) and (390.39,196.03) .. (388.83,197.8) .. controls (387.27,199.57) and (387.37,201.23) .. (389.14,202.79) .. controls (390.91,204.35) and (391.01,206.01) .. (389.45,207.78) .. controls (387.9,209.55) and (388,211.21) .. (389.77,212.77) .. controls (391.54,214.33) and (391.64,215.99) .. (390.08,217.76) .. controls (388.52,219.53) and (388.62,221.19) .. (390.39,222.75) .. controls (392.16,224.31) and (392.26,225.97) .. (390.7,227.74) -- (390.72,228) -- (390.72,228) ;
\draw [color={rgb, 255:red, 208; green, 2; blue, 2 }  ,draw opacity=1 ][line width=2.25]    (367.99,78.62) .. controls (366.11,80.04) and (364.46,79.81) .. (363.04,77.93) .. controls (361.62,76.05) and (359.97,75.82) .. (358.09,77.24) .. controls (356.21,78.67) and (354.56,78.44) .. (353.14,76.56) .. controls (351.71,74.68) and (350.06,74.45) .. (348.18,75.87) .. controls (346.3,77.29) and (344.65,77.06) .. (343.23,75.18) .. controls (341.81,73.3) and (340.16,73.07) .. (338.28,74.49) .. controls (336.4,75.91) and (334.75,75.68) .. (333.33,73.8) .. controls (331.9,71.92) and (330.25,71.69) .. (328.37,73.11) .. controls (326.49,74.53) and (324.84,74.3) .. (323.42,72.42) .. controls (322,70.54) and (320.35,70.31) .. (318.47,71.73) .. controls (316.59,73.16) and (314.94,72.93) .. (313.52,71.05) .. controls (312.09,69.17) and (310.44,68.94) .. (308.56,70.36) .. controls (306.68,71.78) and (305.03,71.55) .. (303.61,69.67) .. controls (302.19,67.79) and (300.54,67.56) .. (298.66,68.98) .. controls (296.78,70.4) and (295.13,70.17) .. (293.71,68.29) .. controls (292.29,66.41) and (290.64,66.18) .. (288.76,67.6) .. controls (286.88,69.02) and (285.23,68.79) .. (283.8,66.91) .. controls (282.38,65.03) and (280.73,64.8) .. (278.85,66.22) .. controls (276.97,67.65) and (275.32,67.42) .. (273.9,65.54) .. controls (272.48,63.66) and (270.83,63.43) .. (268.95,64.85) .. controls (267.07,66.27) and (265.42,66.04) .. (263.99,64.16) .. controls (262.57,62.28) and (260.92,62.05) .. (259.04,63.47) .. controls (257.16,64.89) and (255.51,64.66) .. (254.09,62.78) .. controls (252.67,60.9) and (251.02,60.67) .. (249.14,62.09) .. controls (247.26,63.51) and (245.61,63.28) .. (244.18,61.4) -- (243,61.24) -- (243,61.24) ;
\draw [line width=0.75]  [dash pattern={on 0.84pt off 2.51pt}]  (122.99,177.09) -- (245.05,153.57) ;
\draw [shift={(247.99,153)}, rotate = 169.09] [fill={rgb, 255:red, 0; green, 0; blue, 0 }  ][line width=0.08]  [draw opacity=0] (8.93,-4.29) -- (0,0) -- (8.93,4.29) -- cycle    ;
\draw [line width=0.75]  [dash pattern={on 0.84pt off 2.51pt}]  (71.99,250.14) -- (375,241) ;
\draw [shift={(378,240.91)}, rotate = 178.27] [fill={rgb, 255:red, 0; green, 0; blue, 0 }  ][line width=0.08]  [draw opacity=0] (8.93,-4.29) -- (0,0) -- (8.93,4.29) -- cycle    ;
\draw [line width=0.75]  [dash pattern={on 0.84pt off 2.51pt}]  (42.99,90.14) -- (364.99,80.97) ;
\draw [shift={(367.99,80.88)}, rotate = 178.37] [fill={rgb, 255:red, 0; green, 0; blue, 0 }  ][line width=0.08]  [draw opacity=0] (8.93,-4.29) -- (0,0) -- (8.93,4.29) -- cycle    ;
\draw [color={rgb, 255:red, 208; green, 2; blue, 2 }  ,draw opacity=1 ][line width=2.25]    (348.01,163) .. controls (350.28,163.64) and (351.1,165.09) .. (350.46,167.36) .. controls (349.83,169.63) and (350.64,171.08) .. (352.91,171.72) .. controls (355.18,172.36) and (355.99,173.81) .. (355.35,176.08) .. controls (354.72,178.35) and (355.53,179.8) .. (357.8,180.44) .. controls (360.07,181.08) and (360.88,182.53) .. (360.25,184.8) .. controls (359.61,187.07) and (360.42,188.52) .. (362.69,189.16) .. controls (364.96,189.8) and (365.77,191.25) .. (365.14,193.52) .. controls (364.51,195.79) and (365.32,197.24) .. (367.59,197.88) .. controls (369.86,198.52) and (370.67,199.97) .. (370.03,202.24) .. controls (369.4,204.51) and (370.21,205.96) .. (372.48,206.6) .. controls (374.75,207.24) and (375.57,208.7) .. (374.93,210.97) .. controls (374.29,213.24) and (375.1,214.69) .. (377.37,215.33) .. controls (379.64,215.97) and (380.45,217.42) .. (379.82,219.69) .. controls (379.19,221.96) and (380,223.41) .. (382.27,224.05) -- (384.49,228) -- (384.49,228) ;
\draw [line width=0.75]  [dash pattern={on 0.84pt off 2.51pt}]  (189.99,205.74) -- (325.18,156.29) ;
\draw [shift={(328,155.26)}, rotate = 159.91] [fill={rgb, 255:red, 0; green, 0; blue, 0 }  ][line width=0.08]  [draw opacity=0] (8.93,-4.29) -- (0,0) -- (8.93,4.29) -- cycle    ;
\draw [color={rgb, 255:red, 208; green, 2; blue, 2 }  ,draw opacity=1 ][line width=2.25]    (538.24,128) .. controls (537.46,125.78) and (538.18,124.28) .. (540.4,123.49) .. controls (542.62,122.7) and (543.34,121.2) .. (542.56,118.98) .. controls (541.78,116.76) and (542.5,115.26) .. (544.72,114.47) .. controls (546.94,113.68) and (547.66,112.18) .. (546.88,109.96) .. controls (546.09,107.74) and (546.81,106.24) .. (549.03,105.45) .. controls (551.25,104.66) and (551.97,103.16) .. (551.19,100.94) .. controls (550.41,98.72) and (551.13,97.22) .. (553.35,96.43) .. controls (555.57,95.64) and (556.29,94.14) .. (555.51,91.92) .. controls (554.73,89.7) and (555.45,88.2) .. (557.67,87.41) .. controls (559.89,86.62) and (560.61,85.12) .. (559.83,82.9) -- (560.26,82) -- (560.26,82) ;
\draw [color={rgb, 255:red, 208; green, 2; blue, 2 }  ,draw opacity=1 ][line width=2.25]    (512.5,128) .. controls (510.21,128.57) and (508.78,127.72) .. (508.21,125.43) .. controls (507.64,123.14) and (506.22,122.29) .. (503.93,122.86) -- (502.5,122) -- (502.5,122) ;
\draw [color={rgb, 255:red, 208; green, 2; blue, 2 }  ,draw opacity=1 ][line width=2.25]    (506,98.74) .. controls (506.78,96.52) and (508.28,95.8) .. (510.51,96.58) .. controls (512.73,97.36) and (514.23,96.64) .. (515.02,94.42) .. controls (515.81,92.2) and (517.31,91.48) .. (519.53,92.26) .. controls (521.75,93.04) and (523.25,92.32) .. (524.04,90.1) .. controls (524.83,87.88) and (526.33,87.16) .. (528.55,87.94) .. controls (530.77,88.72) and (532.27,88) .. (533.06,85.78) .. controls (533.85,83.56) and (535.35,82.84) .. (537.57,83.62) -- (542,81.5) -- (542,81.5) ;
\draw [color={rgb, 255:red, 208; green, 2; blue, 2 }  ,draw opacity=1 ][line width=2.25]    (566.67,82) .. controls (568.43,83.57) and (568.52,85.24) .. (566.95,86.99) .. controls (565.38,88.75) and (565.47,90.41) .. (567.23,91.98) .. controls (568.99,93.55) and (569.08,95.22) .. (567.51,96.98) .. controls (565.94,98.74) and (566.03,100.4) .. (567.79,101.97) .. controls (569.55,103.54) and (569.64,105.2) .. (568.07,106.96) .. controls (566.5,108.71) and (566.59,110.38) .. (568.34,111.95) .. controls (570.1,113.52) and (570.19,115.19) .. (568.62,116.95) .. controls (567.05,118.71) and (567.14,120.37) .. (568.9,121.94) .. controls (570.66,123.51) and (570.75,125.17) .. (569.18,126.93) .. controls (567.61,128.69) and (567.7,130.35) .. (569.46,131.92) .. controls (571.22,133.49) and (571.31,135.15) .. (569.74,136.91) .. controls (568.17,138.67) and (568.26,140.34) .. (570.02,141.91) .. controls (571.78,143.48) and (571.87,145.14) .. (570.3,146.9) .. controls (568.73,148.66) and (568.82,150.32) .. (570.58,151.89) .. controls (572.34,153.46) and (572.43,155.12) .. (570.86,156.88) .. controls (569.29,158.63) and (569.38,160.3) .. (571.13,161.88) .. controls (572.89,163.45) and (572.98,165.11) .. (571.41,166.87) .. controls (569.84,168.63) and (569.93,170.29) .. (571.69,171.86) .. controls (573.45,173.43) and (573.54,175.09) .. (571.97,176.85) .. controls (570.4,178.61) and (570.49,180.27) .. (572.25,181.84) .. controls (574.01,183.41) and (574.1,185.08) .. (572.53,186.84) .. controls (570.96,188.6) and (571.05,190.26) .. (572.81,191.83) .. controls (574.57,193.4) and (574.66,195.06) .. (573.09,196.82) .. controls (571.52,198.58) and (571.61,200.24) .. (573.37,201.81) .. controls (575.13,203.38) and (575.22,205.05) .. (573.65,206.81) .. controls (572.08,208.56) and (572.17,210.23) .. (573.92,211.8) .. controls (575.68,213.37) and (575.77,215.03) .. (574.2,216.79) .. controls (572.63,218.55) and (572.72,220.21) .. (574.48,221.78) .. controls (576.24,223.35) and (576.33,225.01) .. (574.76,226.77) -- (574.83,228) -- (574.83,228) ;
\draw [color={rgb, 255:red, 208; green, 2; blue, 2 }  ,draw opacity=1 ][line width=2.25]    (542,60.63) .. controls (539.84,61.58) and (538.29,60.97) .. (537.34,58.81) .. controls (536.39,56.65) and (534.84,56.04) .. (532.68,56.99) .. controls (530.52,57.94) and (528.97,57.33) .. (528.03,55.17) -- (526,54.38) -- (526,54.38) ;
\draw [color={rgb, 255:red, 208; green, 2; blue, 2 }  ,draw opacity=1 ][line width=2.25]    (537.66,152) .. controls (539.85,152.87) and (540.51,154.4) .. (539.64,156.59) .. controls (538.77,158.78) and (539.42,160.31) .. (541.61,161.19) .. controls (543.8,162.06) and (544.46,163.59) .. (543.59,165.78) .. controls (542.72,167.97) and (543.37,169.5) .. (545.56,170.37) .. controls (547.75,171.24) and (548.41,172.78) .. (547.54,174.97) .. controls (546.67,177.16) and (547.32,178.69) .. (549.51,179.56) .. controls (551.7,180.43) and (552.36,181.96) .. (551.49,184.15) .. controls (550.62,186.34) and (551.27,187.87) .. (553.46,188.75) .. controls (555.65,189.62) and (556.31,191.15) .. (555.44,193.34) .. controls (554.57,195.53) and (555.22,197.06) .. (557.41,197.93) .. controls (559.6,198.8) and (560.26,200.34) .. (559.39,202.53) .. controls (558.52,204.72) and (559.17,206.25) .. (561.36,207.12) .. controls (563.55,207.99) and (564.21,209.52) .. (563.34,211.71) .. controls (562.47,213.9) and (563.12,215.43) .. (565.31,216.31) .. controls (567.5,217.18) and (568.16,218.71) .. (567.29,220.9) .. controls (566.42,223.09) and (567.07,224.62) .. (569.26,225.49) -- (570.34,228) -- (570.34,228) ;

\end{tikzpicture}

\caption{\label{fig:Illustration of a drop preserving join}Illustration of
a drop preserving join $\protect\bowtie_{f}$ (see Definition \ref{def:Drop-preserving join})}
\end{figure}

We show in what follows that if the starting net is race-free (see
\ref{def:Race-free net}), then Drop Preserving Joins preserve the
drop condition--and the obtained net is also race-free. This is the
object of the final Theorem \ref{thm:Drop-preserving joins preserve the drop condition}. 

\begin{tcolorbox}[breakable,enhanced,title=Race-freeness, colback=orange!5!white, colframe=orange!75!black]
\label{def:Race-free net}A Net $S$ is race-free if only negative
events can be in conflict with another negative event. That is, $a\sim b\implies p\left(a\right)=p\left(b\right)=\ominus\,OR\,p\left(a\right)\neq\ominus\neq p\left(b\right)$
\end{tcolorbox}

We prove this result through the following Propositions \ref{prop:Drop-preserving joins preserve race-freeness}
and \ref{prop:Drop condition properties for join of events}, starting
with a proof of the conservation of race-freeness after a drop-preserving
join.

\paragraph{Conservation of race-freeness after a drop-preserving join}

\begin{tcolorbox}[breakable,enhanced,title=Drop-preserving joins preserve race-freeness,
colback=orange!5!white, colframe=orange!75!black]

\begin{prop}
\label{prop:Drop-preserving joins preserve race-freeness}Let $S$
be a net skeleton, and $S'$ be obtained after a drop-preserving join
from clusters $P\subseteq\mathds{P}$ and $N$, according to the map
$f$ (Def. \ref{def:Drop-preserving join}).

Then if $S$ is race-free, $S'$ is also race-free.
\end{prop}

\tcbline{}

Suppose $S$ is race-free, and let $a\sim_{S'}b$ be a couple of events
in conflict in $S'$.
\begin{itemize}
\item If $a$ and $b$ are events of $S$, then they were not affected by
the join - so race-freeness is preserved.
\item If $a\in T_{S}$ and $b\notin T_{S}$, then $b=\boxed{e\bowtie e'}$
for some events $e^{\oplus}$ and $e'{}^{\ominus}$ of $S$. So in
$S$, either $a\sim_{S}e^{\oplus}$ or $a\sim_{S}e'{}^{\ominus}$.
\begin{itemize}
\item if $p\left(a\right)=\ominus$, since $S$ is race-free, $a^{\ominus}{\not\sim}_{S}e^{\oplus}$.
So $a^{\ominus}\sim_{S}e'{}^{\ominus}$. But since $a$ was not joined
the cluster $N$ was not maximal. Contradiction with Defintion \ref{def:Drop-preserving join}.\ref{enu:point 1}.
So $p\left(a\right)\neq\ominus$
\item if $p\left(a\right)=\oplus\,or\,0$, since $S$ is race-free, $a^{\oplus,0}{\not\sim}_{S}e^{\ominus}$.
So $a^{\oplus}\sim_{S}e^{\oplus}$, and thus $a^{\oplus,0}\sim_{S'}\boxed{e\bowtie e'}$.
So race-freeness is preserved.
\end{itemize}
\item ($\star$) If $a\notin T_{S}$ and $b\notin T_{S}$, then $a=\boxed{h\bowtie h'}$
and $b=\boxed{e\bowtie e'}$ in a similar fashion as that of the previous
case. Then $p\left(\boxed{h\bowtie h'}\right)=p\left(\boxed{e\bowtie e'}\right)=0$
so the race-freeness is preserved.\label{ancho dans demo}
\end{itemize}
\end{tcolorbox}

\begin{rem}
Note that in ($\star$) of \ref{ancho dans demo}, we have either
of the following four case: %
\begin{tabular}{|c|c|c|}
\hline 
and & $h'\sim e'$ & $h'\not\sim e'$\tabularnewline
\hline 
\hline 
$h\sim e$ & 1 & 2\tabularnewline
\hline 
$h\not\sim e$ & 3 & 4\tabularnewline
\hline 
\end{tabular}, where Cases 1 and 2 are the only ones that allow a drop-preserving
join. Indeed, Case 4 is not possible since $\boxed{h\bowtie h'}$
are eventually in conflict $\boxed{e\bowtie e'}$ in $S'$. Furthermore,
Case 3 is not possible as per (\ref{def:Drop-preserving join}).(\ref{enu:point2a}).
But if this latter condition were to be dropped, the intermediary
join step $S\rightarrow S_{h\bowtie h'}$ would not conserve race-freeness,
although $S\rightarrow S_{h\bowtie h'}\rightarrow S_{h\bowtie h',e\bowtie e'}$
would yield a race-free net. 
\end{rem}

\paragraph{Technical intermediary Results}

The following is a technical lemma necessary for the proof the key
properties Proposition \ref{prop:Drop condition properties for join of events}
regarding the properties verified by the drop condition of joins of
events.

\begin{tcolorbox}[breakable,enhanced,title=Configurations before and after a drop-preserving
join, colback=white,colframe=black!30!white,borderline={0.5mm}{0mm}{orange!30!black,dashed}]

\begin{lem}
\label{prop:Configurations before and after a drop-preserving join}
Let $S$ be a race-free net skeleton, and $S'$ be the net obtained
after having performed a drop-preserving join on the pairs of events:
$p_{j}\bowtie n_{j}$ for $j\in J\subseteq\left\llbracket 1,n\right\rrbracket $.

We look at a configuration $x$ in $S'$, enabling the events $f_{1},\ldots,f_{n}$
such that $x\overset{f_{i}}{\ext}y_{i}$, with $\forall j\in J,\hspace*{1em}f_{j}=p_{j}\bowtie n_{j}$.
In $S$, define the family of events $\left(e_{i}\right)_{i}$ where
the joined events are replaced by their positive contribution. Formally,
$\forall i,\hspace*{1em}e_{i}:=\begin{cases}
f_{i} & i\notin J\\
p_{i} & o/w
\end{cases}$.

Then for all $I\subseteq\left\llbracket 1,n\right\rrbracket $, $z_{I}\in C\left(S\right)\iff y_{I}\in C\left(S\right)$.
\end{lem}

\tcbline{}

Let $I\subseteq\left\llbracket 1,n\right\rrbracket $.
\begin{itemize}
\item If $I\cap J=\emptyset$, $\forall i\in I,e_{i}=f_{i}$, so $z_{I}=y_{I}$.
Thus $z_{I}\in C\left(S\right)\iff y_{I}\in C\left(S'\right)$.
\item Suppose $I\cap J=L\neq\emptyset$. we write $z_{I}=x\sqcup e_{I\setminus J}\sqcup e_{I\cap J}=x\sqcup f_{I\setminus J}\sqcup p_{I\cap J}$.
\begin{description}
\item [{$\implies$:}] If $z_{I}\in C\left(S\right)$, since $z_{I}$ is
a configuration, the events $f_{I\setminus J}$ and $p_{I\cap J}$
are mutually compatible ($\star$). Let us show that $y_{I}=x\sqcup f_{I\setminus J}\sqcup\left(p_{u}\bowtie n_{u}\right)_{u\in I\cap J}\in C\left(S'\right)$.
This is verified iff the $f_{I\setminus J}$ and $\left(p_{u}\bowtie n_{u}\right)_{u\in I\cap J}$
are mutually compatible too. That is, iff 
\[
\left\{ \begin{array}{ll}
\text{no }\boxed{p_{u}\bowtie n_{u}}\sim\boxed{p_{v}\bowtie n_{v}} & u,v\in I\cap J\\
\text{no }\boxed{f_{i}}\sim\boxed{p_{u}} & i\in I\setminus J,u\in I\cap J\\
\text{no }\boxed{f_{i}}\sim\boxed{n_{u}} & i\in I\setminus J,u\in I\cap J
\end{array}\right.
\]

\begin{itemize}
\item Case 1 cannot happen, as it would imply any of the four following
conflict :%
\begin{tabular}{|c|c|c|}
\hline 
$\sim$ & $\boxed{p_{v}}$ & $\boxed{n_{v}}$\tabularnewline
\hline 
\hline 
$\boxed{p_{u}}$ & \begin{cellvarwidth}[t]
\centering
Contradicts \\
($\star$)
\end{cellvarwidth} & \begin{cellvarwidth}[t]
\centering
Contradicts \\
Race-freeness
\end{cellvarwidth}\tabularnewline
\hline 
$\boxed{n_{u}}$ & \begin{cellvarwidth}[t]
\centering
Contradicts\\
 Race-freeness
\end{cellvarwidth} & \begin{cellvarwidth}[t]
\centering
If $\boxed{n_{u}}\sim\boxed{n_{v}}$, \\
drop-preserv. join\\
 implies $\boxed{p_{u}}\sim\boxed{p_{v}}$. \\
Contradiction
\end{cellvarwidth}\tabularnewline
\hline 
\end{tabular}.
\item Case 2 contradicts ($\star$). 
\item Finally, Case 3 would contradict race-freeness or the drop-preserving
join assumptions: if $p\left(\boxed{f_{i}}\right)=0/\oplus$, $S$
is not race-free. if $p\left(\boxed{f_{i}}\right)=\ominus$, then
the negative cluster that was joined was not maximal since $\boxed{f_{i}}$
extends it.
\end{itemize}
\end{description}
\begin{itemize}
\item So $y_{I}\in C\left(S'\right)$
\end{itemize}
\begin{description}
\item [{$\impliedby$:}] If $y_{I}\in C\left(S'\right)$, then the events
$f_{I\setminus J}=e_{I\setminus J}$ and $\left(p_{u}\bowtie n_{u}\right)_{u\in I\cap J}$
are mutually compatible. Similar reasoning to previously. 
\end{description}
\end{itemize}
\end{tcolorbox}
\begin{tcolorbox}[breakable,enhanced,title= Properties of drop-preserving joins w.r.t
the drop condition ,colback=orange!5!white, colframe=orange!75!black]

\begin{prop}
\label{prop:Drop condition properties for join of events}Let $S$
be a race-free net skeleton, and $S'$ be the net obtained after having
performed a drop-preserving join on the pairs of events: $p_{j}\bowtie n_{j}$
for $j\in J\subseteq\left\llbracket 1,n\right\rrbracket $.

We look at a configuration $x$ in $S'$, enabling the events $f_{1},\ldots,f_{n}$
such that $x\overset{f_{i}}{\ext}y_{i}$, with $\forall j\in J,\hspace*{1em}f_{j}=p_{j}\bowtie n_{j}$.
In $S$, define the family of events $\left(e_{i}\right)_{i}$ where
the joined events are replaced by their positive contribution. Formally,
$\forall i,\hspace*{1em}e_{i}:=\begin{cases}
f_{i} & i\notin J\\
p_{i} & o/w
\end{cases}$.

Then the following properties hold:
\begin{enumerate}
\item $x$ is also a configuration in $S$, with $\cutpn x_{S'}=\cutpn x_{S}={\displaystyle \bigsqcup_{i\in\left\llbracket 1,n\right\rrbracket \setminus J}\pre{\boxed{e_{i}}}\sqcup\bigsqcup_{j\in J}\pre{\boxed{p_{j}}}\sqcup\bigsqcup_{j\in J}\pre{\boxed{n_{j}}}}$.\\
$\forall i,\hspace*{1em}x\overset{e_{i}}{\ext}z_{i}$, where $z_{i}:=x\sqcup e_{i}$.\\
$\forall j\in J,z_{j}=y_{j}\setminus p_{j}\bowtie n_{j}\sqcup p_{j}$.
\item if $p\bowtie n\in x$, $\qdv x{y_{1},\ldots,y_{n}}_{S_{p\bowtie n}}=\qdv{x\setminus p\bowtie n\sqcup\left\{ p,n\right\} }{\left(y_{i}\setminus p\bowtie n\sqcup\left\{ p,n\right\} \right){}_{i}}_{S}$
\item if $\forall j\in J,\hspace*{1em}p_{j}\bowtie n_{j}\notin x$, 
\[
\qdv x{y_{1},\ldots,y_{n}}_{S'}=\qdv x{z_{1},\ldots,z_{n}}_{S}
\]
\end{enumerate}
\end{prop}

\tcbline{}
\begin{enumerate}
\item -
\item If $p\bowtie n\in x$ and $\boxed{p\bowtie n}\notin\cut x$, $\post{\boxed{p\bowtie n}}\cap\cutpn x=\emptyset$.
So $\restg{S_{p\bowtie n}}{\int{\cutpn x}{\cutpn{y_{I}}}}=\restg P{\int{\cutpn x}{\cutpn{y_{I}}}}$,
$\operator{\restg{S_{p\bowtie n}}{\int{\cutpn x}{\cutpn{y_{I}}}}}=\operator{\restg S{\int{\cutpn x}{\cutpn{y_{I}}}}}$\\
If $p\bowtie n\in x$ and $\boxed{p\bowtie n}\in\cut x$, $\post{\boxed{p\bowtie n}}\subset\cutpn x$.
But by construction $\post{\boxed{p\bowtie n}}_{\left|S_{e\bowtie e'}\right.}=\left[\post{\boxed{p}}\sqcup\post{\boxed{n}}\right]_{\left|S\right.}$.
So $\restg{S_{p\bowtie n}}{\int{\cutpn x}{\cutpn{y_{I}}}}=\restg S{\int{\cutpn x}{\cutpn{y_{I}}}}$.
\item Let $I\subseteq\left\llbracket 1,n\right\rrbracket $
\begin{description}
\item [{\uline{If~\mbox{$I\cap J=\emptyset$}:}}] $\forall i\in I,f_{i}=e_{i}$,
so $y_{I}=z_{I}$. $\operator{\restg{S'}{\int{\cutpn x}{\cutpn{y_{I}}}}}=\operator{\restg S{\int{\cutpn x}{\cutpn{z_{I}}}}}$ 
\item [{\uline{If~\mbox{$I\cap J\neq\emptyset$}:}}] For clarity, let
us distinguish the sets of constant pre-places under the action of
$\boxed{f_{I}}$ and $\boxed{e_{I}}$ respectively with an upper script:
\[
K_{I}^{f}:=R_{I}^{f}\sqcup X^{f}={\displaystyle \bigsqcup_{t\notin I}}\pre{\boxed{f_{t}}}\sqcup{\displaystyle \bigsqcup_{\begin{array}{c}
\ci a\in\cutpn x\\
\forall i,\,\ci a\notin\pre{\boxed{f_{i}}}
\end{array}}}\ci a
\]
 and 
\[
K_{I}^{e}:=R_{I}^{e}\sqcup X^{e}={\displaystyle \bigsqcup_{t\notin I}}\pre{\boxed{e_{t}}}\sqcup{\displaystyle \bigsqcup_{\begin{array}{c}
\ci a\in\cutpn x\\
\forall i,\,\ci a\notin\pre{\boxed{e_{i}}}
\end{array}}}\ci a
\]
 Remark that $K_{I}^{e}=K_{I}^{f}\sqcup\pre{\boxed{n_{I\cap J}}}$.
In that sense, $\operator{\restg{S'}{\int{\cutpn x}{\cutpn{y_{I}}}}}=\bigotimes_{i\in I}Q_{0}\left(\boxed{f_{i}}\right)\otimes\Id{K_{I}^{f}}$
and $\operator{\restg S{\int{\cutpn x}{\cutpn{z_{I}}}}}=\bigotimes_{i\in I}Q_{0}\left(\boxed{e_{i}}\right)\otimes\Id{K_{I}^{e}}$.
\end{description}
\begin{itemize}
\item Recalling that by \vref{def:Join of two events}, $\forall j\in J,\hspace*{1em}Q_{0}\left(\boxed{e_{j}}\right)=Q_{0}\left(\boxed{p_{j}\bowtie n_{j}}\right)=Q_{0}\left(\boxed{p_{j}}\right)\otimes\Id{Q_{0}\left(\pre{\boxed{n_{j}}}\right)}$,
we have
\begin{align*}
\operator{\restg{S'}{\int{\cutpn x}{\cutpn{y_{I}}}}} & =\bigotimes_{i\in I\setminus J}Q_{0}\left(\boxed{f_{i}}\right)\otimes\bigotimes_{u\in I\cap J}Q_{0}\left(\boxed{f_{u}}\right)\otimes\Id{K_{I}^{f}}\\
 & =\bigotimes_{i\in I\setminus J}Q_{0}\left(\boxed{f_{i}}\right)\otimes\left[\bigotimes_{u\in I\cap J}Q_{0}\left(\boxed{p_{u}}\right)\otimes\bigotimes_{u\in I\cap J}\Id{Q_{0}\left(\pre{\boxed{n_{u}}}\right)}\right]\otimes\Id{K_{I}^{f}}\\
 & =\bigotimes_{i\in I\setminus J}Q_{0}\left(\boxed{e_{i}}\right)\otimes\bigotimes_{u\in I\cap J}Q_{0}\left(\boxed{p_{u}}\right)\otimes\Id{Q_{0}\left(K_{I}^{f}\sqcup\pre{\boxed{n_{I\cap J}}}\right)}\\
 & =\bigotimes_{i\in I}Q_{0}\left(\boxed{e_{i}}\right)\otimes\Id{K_{I}^{e}}\\
 & =\operator{\restg S{\int{\cutpn x}{\cutpn{z_{I}}}}}
\end{align*}
\item As such, using \ref{prop:Configurations before and after a drop-preserving join}
for the correspondence of configurations in $S'$ and $S$, we obtain
\begin{align*}
\qdv x{y_{1},\ldots,y_{n}}_{S'} & =\sum_{\begin{array}{c}
I\subseteq\left\llbracket 1,n\right\rrbracket \\
y_{I}\in C(S')
\end{array}}\left(-1\right)^{|I|}\tr\operator{\restg S{\int{\cutpn x}{\cutpn{y_{I}}}}}\\
 & =\sum_{\begin{array}{c}
I\subseteq\left\llbracket 1,n\right\rrbracket \\
z_{I}\in C(S)
\end{array}}\left(-1\right)^{|I|}\tr\operator{\restg S{\int{\cutpn x}{\cutpn{z_{I}}}}}\\
 & =\qdv x{y_{1},\ldots,y_{n}}_{S}
\end{align*}
\end{itemize}
\end{enumerate}
\end{tcolorbox}

\paragraph{Conservation of the Drop Condition after a drop-preserving join}

We now conclude with the following final theorem, stating that drop-preserving
join preserve the drop condition. This builds the ground for a sensible
notion of composition for Quantum Petri Nets.

\begin{tcolorbox}[breakable,enhanced,title=Drop-preserving joins preserve the drop condition,colback=orange!5!white,
colframe=red!75!black]

\begin{thm}
\label{thm:Drop-preserving joins preserve the drop condition}Let
$S$ be a net skeleton, and $S'$ obtained after one drop-preserving
join from clusters $P\subseteq\mathds{P}$ and $N$, according to
the map $f$ (Def. \ref{def:Drop-preserving join}).

Then if $S$ is race-free and satisfies the Drop Condition, $S'$
is also race-free and satisfies the drop condition.

An immediate induction also yields that any finite sequence of Drop-preserving
joins preserve the drop condition.
\end{thm}

\tcbline{}

$S'$ is race-free as per \ref{prop:Drop-preserving joins preserve race-freeness}.

By the Cluster property, the Drop Condition is satisfied iff it is
satisfied on every cluster of single extensions of $\oplus/0$ events,
for every configuration $x$.

Let $x$ be a configuration in $S'$, and $F$ be such a cluster enabled
by $x$. Let $E$ be the set of events of $S$ such that $E$ is the
set $F$, where one has replaced the new joined elements $\boxed{p_{u}\bowtie n_{u}}$
by their positive contribution $\boxed{p_{u}}$. 
\begin{itemize}
\item If $F$ contains no new joined events (every event of $F$ is also
in $S$), it is unchanged by the join operation, and $E=F$ satisfies
the Drop condition in $S$.
\item Otherwise, we are exactly in the situation of Property \#3 of Proposition
\ref{prop:Drop condition properties for join of events}; and the
drop condition on $E$ and $F$ are equal; and the drop-condition
on $F$ is also satisfied.
\end{itemize}
\end{tcolorbox}

\subsubsection{Future Extensions}

A complete compositional operation for Quantum Petri Net should preserve
all the Quantum Petri Net axioms--not only the local Drop Condition
should be preserved after the transformation, but also should be the
Local Obliviousness and the local Functoriality of the quantum valuation.
We have shown that performing a drop-preserving join preserves the
local drop-condition. It is not known whether the drop-preserving
joins are minimally restrictive w.r.t. conservation of the drop-condition. 

Further extensions need to derive sufficient (and necessary) minimal
conditions for the other properties to be preserved.

\section{Conclusion}

We have defined \emph{Quantum Petri Nets} (QPNs) as Petri nets annotated
with a quantum valuation $Q_{0}$, whose properties are compatible
with the quantum event structure semantics of \autocite{clairambaultConcurrentQuantumStrategies2019,clairambaultGameSemanticsQuantum2019}.
This yields a natural unfolding semantics, paralleling the categorical
correspondence between classical Petri nets, their unfoldings, and
event structures. Our main technical contributions are: (i) a local
definition of \emph{Quantum Occurrence Nets} compatible with quantum
event structures, (ii) an extension to QPNs via unfolding semantics,
and (iii) a base for a compositional framework for QPNs--where \emph{drop-preserving
joins} are introduced as a sufficient criteria to preserve the drop-condition
after a composition operation.

Future directions include establishing explicit categorical links
between Quantum Petri Nets, Quantum Occurrence Nets, and Quantum Event
Structures, and understanding the expressiveness of the new models
and their relations with existing Petri net classes. Characterizing
precisely the properties for joins to preserve QPNs properties is
also necessary. Finally, additional combinatorial simplifications
of the drop-conditions could be investigated for verificational purposes. 

\printbibliography

\end{document}